\pdfoutput=1
\documentclass[letterpaper,11pt]{article}
\usepackage[margin=1in]{geometry}
\usepackage[CJKbookmarks=true,
            bookmarksnumbered=true,
            bookmarksopen=true,
            colorlinks=true,
            citecolor=red,
            linkcolor=blue,
            anchorcolor=red,
            urlcolor=blue
            ]{hyperref}
\usepackage[title]{appendix}
\usepackage{graphicx}
\usepackage{caption}
\usepackage{subcaption}
\usepackage{booktabs}
\usepackage{algorithm}
\usepackage{algorithmic}
\usepackage{amsmath,amsthm,amsfonts,amssymb}
\usepackage{hyperref}
\usepackage{color}
\usepackage{enumitem}
\usepackage{bm}
\usepackage{multirow}
\usepackage{bbm}
\usepackage{subfiles}
\usepackage{xspace}
\usepackage{mathtools} 
\usepackage{tikz}
\usepackage{natbib}

\newtheorem{theorem}{Theorem}
\newtheorem*{theorem*}{Theorem}
\newtheorem{lemma}{Lemma}

\newtheorem{proposition}{Proposition}

\theoremstyle{definition}
\newtheorem{remark}{Remark}
\newtheorem{definition}{Definition}

\tikzset{global scale/.style={
    scale=#1,
    every node/.append style={scale=#1}
  }
}

\newcommand{\maG}{\mathcal{G}}
\newcommand{\maH}{\mathcal{H}}

\newcommand{\maM}{\mathcal{M}}
\newcommand{\maN}{\mathcal{N}}

\newcommand{\maP}{\mathcal{P}}

\newcommand{\maS}{\mathcal{S}}
\newcommand{\maT}{\mathcal{T}}

\newcommand{\maV}{\mathcal{V}}

\newcommand{\maX}{\mathcal{X}}

\newcommand{\maZ}{\mathcal{Z}}

\newcommand{\prob}[1]{ \mathbb{P}\left[ #1 \right] }
\newcommand{\expect}[1]{ \mathbb{E}\left[ #1 \right] }

\newcommand{\en}{\mathsf{{e}}}
\newcommand{\vn}{\mathsf{{v}}}

\newcommand{\pth}[1]{\left( #1 \right)}
\newcommand{\qth}[1]{\left[ #1 \right]}
\newcommand{\sth}[1]{\left\{ #1 \right\}}

\newcommand{\calM}{{\mathcal{M}}}

\newcommand{\ti}{\tilde}

\newcommand{\sfE}{{\mathsf{E}}}

\newcommand{\sfH}{{\mathsf{H}}}
\newcommand{\sfI}{{\mathsf{I}}}

\newcommand{\sfM}{{\mathsf{M}}}

\newcommand{\sfS}{{\mathsf{S}}}
\newcommand{\sfT}{{\mathsf{T}}}

\newcommand{\sfe}{{\mathsf{e}}}

\newcommand{\sfn}{{\mathsf{n}}}

\newcommand{\sfv}{{\mathsf{v}}}

\newcommand{\indc}[1]{{\mathbf{1}_{\left\{{#1}\right\}}}}

\newcommand{\ER}{Erd\H{o}s-R\'enyi }

\newcommand{\aut}{\mathsf{aut}}

\renewcommand{\hom}{\mathsf{hom}}
\newcommand{\inj}{\mathsf{inj}}
\newcommand{\vp}{\varphi}

\def\var{\mathrm{Var}}
\def\E{\mathbb{E}}
\def\P{\mathbb{P}}

\def\nodebullet{\node[fill=black,circle,inner sep=0pt,minimum size=3pt]}

\usepackage{xcolor}
\definecolor{red1}{RGB}{255, 0, 0}  
\definecolor{red2}{RGB}{220, 47, 2}   
\definecolor{red3}{RGB}{189, 58, 85}

\usetikzlibrary{intersections,calc,decorations.pathreplacing,through,arrows}
\usetikzlibrary{ decorations.markings,positioning}

\title{Testing Correlation in Graphs by Counting Bounded Degree Motifs}
\date{}
\author{Dong Huang and Pengkun Yang\thanks{D.\ Huang and P.\ Yang are with the Department of Statistics and Data Science, Tsinghua University. P. Yang is supported in part by National Key R\&D Program of China 2024YFA1015800, 
Tsinghua University Dushi Program 2025Z11DSZ001, and High Performance Computing Center, Tsinghua University.}}

\begin{document}

\maketitle

\begin{abstract}
We investigate the problem of detecting correlation between two \ER graphs $\maG(n,p)$, formulated as a hypothesis testing problem: under the null hypothesis, the two graphs are independent, while under the alternative hypothesis, they are correlated through a latent bijective mapping between their vertex sets. We develop a polynomial-time test by counting bounded degree motifs and prove its effectiveness for any constant correlation coefficient $\rho$ when the edge connection probability satisfies $p\ge n^{-1+\delta}$ for some constant $\delta>0$. 
In particular, our guarantee improves the constrain of motif-counting methods from $\rho\ge\sqrt{\alpha}$ to any constant $\rho = \Omega(1)$, where $\alpha\approx 0.338$ is the Otter's constant.
   
\end{abstract}



\section{Introduction}\label{sec-intro}

Correlation analysis between datasets is one of the most fundamental problems in statistics. 
In the classical vector setting, the problem of testing independence between two random vectors has been extensively studied, including low-dimensional~\citep{pearson1900x,kendall1938new,spearman1987proof} and high-dimensional settings~\citep{szekely2007measuring,zhu2017projection,gretton2007kernel}.
In contrast, correlation analysis for graph-structured data remains much less explored. 
Indeed, since shared randomness increases the overlap of network patterns across graphs, correlation often manifests as structural similarity between graphs, which is common in real scenarios.
Recently, there has been a surge of interest in the problem of analyzing correlated graphs, as in many applications the observations are more naturally represented as graphs rather than vectors. 
Such problems arise from various domains:
\begin{itemize}
    \item In social network analysis, whether two friendship networks on different social network platforms share structural similarities is a crucial task in privacy protection~\citep{narayanan2008robust}.
    \item In computer vision, 3-D shapes can be represented as graphs, and a significant problem is determining whether two graphs represent the same object under deformations~\citep{berg2005shape}.
    \item In natural language processing, one important problem is the ontology alignment problem, which refers to uncovering the correlation between two different knowledge graphs~\citep{haghighi2005robust}.
    \item In computational biology, proteins can be regarded as vertices and the interactions between them can be formulated as weighted edges, and the protein-protein interactions (PPI) can be represented as a graph~\citep{singh2008global}.
\end{itemize}

A common strategy for handling graph data is through graph embedding. Among such approaches, spectral embedding is a widely used method that maps graphs into low-dimensional vectors~\citep{rohe2011spectral}. 
This strategy has been applied to testing independence between graphs~\citep{lee2019network} and to defining measures of graph correlation~\citep{fujita2017correlation}.
Despite its popularity, spectral embedding has several limitations. First, it requires selecting an embedding dimension which is often heuristic and lacks theoretical guarantees. Second, reducing graphs to vector representations inevitably sacrifices structural information compared with analyzing graphs directly. Third, spectral embedding relies on singular value decomposition, which can be computationally prohibitive for large-scale networks.
These challenges highlight the need for correlation measures that operate directly on graph topology while remaining both statistically sound and computationally efficient. Motivated by this gap, we aim to develop a new method for testing correlation between two graphs.


Building on the hypothesis testing framework proposed in~\cite{barak2019nearly}, for two graphs $G_1,G_2$ with vertex sets $V(G_1),V(G_2)$ and edge sets $E(G_1), E(G_2)$,
we consider the following graph correlation detection problem: under the null hypothesis $\maH_0$, $G_1$ and $G_2$ are independent; under the alternative hypothesis $\maH_1$, there exists a latent vertex bijection $\pi:V(G_1)\mapsto V(G_2)$ that induces correlation between the edges of the two graphs. 
Specifically, for any $uv\in E(G_1)$ with $u,v\in V(G_1)$, the corresponding pair $\pi(u)\pi(v)$ lies in $E(G_2)$, and the edges $uv$ and $\pi(u)\pi(v)$ are statistically correlated.
Under both \(\maH_0\) and \(\maH_1\), each graph marginally follows the same random-graph model; what distinguishes \(\maH_1\) is the presence of statistical dependence between corresponding edges across the two graphs.
Given $G_1$ and $G_2$, the goal is to test $\maH_0$ against $\maH_1$ by the latent structure
under $\maH_1$. 

In this paper, we focus on the \ER model~\cite{paul1959random}.
Specifically, the \ER random graph $\maG(n,p)$ is the graph on $n$ vertices where each edge connects with probability $0<p<1$ independently. 
Under the null hypothesis $\maH_0$, the two graphs $G_1$ and $G_2$ follow $\maG(n,p)$ independently; under the alternative hypothesis $\maH_1$, $G_1$ and $G_2$ follow the following \emph{correlated \ER graph} $\maG(n,p,\rho)$.
\begin{definition}[Correlated \ER graph]
    For two random graphs $G_1,G_2$ with vertex sets $V(G_1),V(G_2)$ and edge sets $E(G_1),E(G_2)$, let $\pi$ denote a latent bijective mapping from $V(G_1)$ to $V(G_2)$. We say $(G_1,G_2)$ follows \emph{correlated \ER graph} $\maG(n,p,\rho)$ if both marginal distributions are \ER graph $\maG(n,p)$ and each pair of edges $(uv,\pi(u)\pi(v))$ follows the correlated bivariate Bernoulli distribution with correlation coefficient $\rho$ for any $u,v\in V(G_1)$.
\end{definition}


Let $\maP_0$ and $\maP_1$ denote the probability measures for $(G_1,G_2)$ under $\maH_0$ and $\maH_1$, respectively. We say a test statistic $\maT(G_1,G_2)$ with a threshold $\tau$ succeeds in detection, if the sum of Type I and Type II errors is bounded by $0.05$ as $n\to \infty$:
\begin{align}\label{eq:detection-cretiria}
    \limsup_{n\to \infty}\qth{\maP_0(\maT\ge\tau)+\maP_1(\maT< \tau)}\le 0.05.
\end{align}
%
%
%
%
%
In fact, the detection threshold can be improved from $0.05$ to any prescribed positive constant; see Remark~\ref{rmk:005} for details.
It is well-known that the minimal value of the sum of Type I and Type II errors between $\maP_0$ and $\maP_1$ is achieved by the likelihood ratio test (see, e.g., \cite[Theorem 13.1.1]{lehmann2005testing}). 
However, the likelihood ratio test requires the evaluation over the space of latent permutations incurring a computational cost of $n!$.
In order to design scalable tests, one must instead exploit informative graph properties that can be computed efficiently while still being identifiable for the underlying models. 
One common and significant methodology is to look at the graphs from a motif perspective, identifying the characteristic and recurrent connection patterns.
For correlation detection between two graphs, \cite{barak2019nearly} analyzed counts of balanced graphs (denser than any of their subgraphs), whereas \cite{mao2024testing} used tree counts to obtain detection guarantees. Both approaches aggregate over large motif families, reflecting that usable correlation signal accumulates with the family size.

Despite their theoretical appeal, balanced graphs or tree structures rarely occur in real-world settings such as social networks.
A line of research instead focuses on counting other types of motifs. 
However, relying on a single class of motifs often yields weak signals and can significantly limit performance.
Thus, the central challenge is to design test statistics that are not only powerful and computationally efficient, but also flexible to capture the structural characteristics of real networks.

In this paper, we introduce a new approach based on counting bounded degree motifs, namely motifs whose vertex degrees are bounded by a universal constant. This family is broad—encompassing balanced graphs, trees, triads, cliques, and stars—and therefore provides richer structural information.
Importantly, it includes commonly observed patterns such as triangles and quadrilaterals, which frequently appear in real-world networks.
Moreover, we establish rigorous theoretical guarantees for the bounded degree motif family, showing that it offers both statistical power and practical relevance.
In addition, bounded degree motifs can be efficiently estimated even on large networks: they admit scalable counting via local exploration or graph sampling techniques, thereby avoiding exhaustive enumeration over all subgraphs.

\subsection{Main Results}

In this subsection, we present the main results. 

\begin{theorem}\label{thm:main}
Assume that $p\ge n^{-1+\delta}$ for some constant $\delta>0$. Then for any constant correlation $\rho=\Omega(1)$, there exists a polynomial-time computable test statistic $\maT$ and a threshold $\tau$ such that
\[
\limsup_{n\to\infty}\big[\maP_0(\maT\ge \tau)+\maP_1(\maT<\tau)\big]\le 0.05.
\]
\end{theorem}

We will show in Section~\ref{sec:general-motif} that any \emph{$C$-admissible} motif family is sufficient for detection.  In Section~\ref{sec:bd-sub-count}, we prove that when $p\ge n^{-2/3}$ and $\rho=\Omega(1)$, there exists a \emph{$C$-admissible} motif family by constructing a bounded degree motif family.  In Section~\ref{sec:refined-results}, we further develop a refined bounded degree construction that is sufficient for detection under the weaker condition $p\ge n^{-1+\delta}$ and $\rho=\Omega(1)$, thereby establishing Theorem~\ref{thm:main}.

It is shown in~\cite{barak2019nearly} that a polynomial-time algorithm based on counting balanced graphs succeeds for correlation detection when $p\in [n^{-1+\epsilon},n^{-1+1/153}]\cup [n^{-1/3},n^{-\epsilon}]$ for any sufficiently small constant $\epsilon>0$ and $\rho=\Omega(1)$, where a balanced graph is denser than each of its nontrivial subgraphs. The recent work~\cite{mao2024testing} proposed a polynomial-time algorithm based on counting trees, which succeeds in detection when $p\ge n^{-1+o(1)}$ and $\rho^2\ge \alpha\approx 0.338$, where $\alpha$ is Otter's constant~\cite{otter1948number}. In view of Theorem~\ref{thm:main}, our test succeeds over the broad range $p\ge n^{-1+\delta}$ for any constant $\rho=\Omega(1)$ by counting bounded degree motifs.

For the remaining regimes,~\cite{ding2023low} provided evidence from the low-degree framework that the condition $\rho^2\ge \alpha$ is essential for detection when $p=n^{-1+o(1)}$. When $p\le n^{-1-\delta}$ for some constant $\delta>0$, detection is information-theoretically impossible~\cite{wu2023testing,ding2023detection}. Taken together, these results suggest that Theorem~\ref{thm:main} attains the strongest detection guarantee achievable by polynomial time. We note that~\cite{ding2023polynomial} also provided a polynomial-time matching algorithm (based on a nontrivial iterative procedure) that works when $p\ge n^{-1+\delta}$ and $\rho=\Omega(1)$, and hence can be adapted to detection. In contrast, our approach achieves the same detection regime via a simple and natural motif counting statistic, avoiding the more intricate iterative procedure.

Following~\cite{ding2023low,li2025algorithmic} and motivated by the low-degree framework~\cite{hopkins2018statistical}, there is evidence that any degree-$o(\rho^{-1})$ polynomial fails for detection in the \ER model. Consequently, when $\rho=o(1)$, no constant-degree polynomial is expected to succeed. This suggests that our assumption $\rho=\Omega(1)$ is essentially necessary for polynomial-time algorithms within this framework. 

\subsection{Test Statistics}\label{sec:method}

To obtain a computationally efficient test, a natural approach is to use summary statistics rather than searching over all possible bijective mappings. 
Ideally, the graph can be uniquely identified from a sufficiently rich set of summary statistics. 
Graph homomorphism numbers provide a particularly prominent class of such statistics. 
Specifically, for two simple graphs $\sfM$ and $G$,
a homomorphism of $\sfM$ into $G$ is an edge-preserving mapping from $V(\sfM)$ to $V(G)$. Let $\hom(\sfM,G)$ be the number of homomorphisms of $\sfM$ into $G$.  
It is well-known that the function of homomorphism numbers $\hom(\cdot,G)$ uniquely determines a simple graph $G$ (see, e.g., \cite[Theorem 5.29]{lovasz2012large}). By~\cite{muller1977edge}, when $\en(G)\ge \vn(G)\log \vn(G)$, the homomorphism numbers $\hom(\sfM,G)$ for motifs $\sfM$ with $\en(\sfM)<\en(G)$ determine $G$. It is further conjectured that $\hom(\sfM,G)$ for all $\vn(\sfM)<V(G)$ or $\en(\sfM)<\en(G)$ determine $G$ if $\vn(G)\ge 3$ and $\en(G)\ge 4$ \cite[Conjectures 5.30 and  5.31]{lovasz2012large}.
However, computing $\hom(\sfM,G)$ for all $\sfM$ up to the scale of $G$ is still computationally prohibitive. 
We instead consider the number of injective homomorphisms of $\sfM$ into $G$ denoted by $\inj(\sfM,G)$, which can be applied to evaluate $\hom(\sfM,G)$ \cite[(5.16)]{lovasz2012large}.
Indeed, $\inj(\sfM,G)$ indicates the motif counts of $\sfM$ in $G$. 
We only compute a subset of injective homomorphism numbers over an informative family of motifs $\sfM\in \calM$.

In our correlation detection problem, given a motif $\sfM$, the injective homomorphism numbers $\inj(\sfM,G_1)$ and $\inj(\sfM,G_2)$ are independent under the null hypothesis $\maH_0$, while they are correlated under the alternative $\maH_1$. 
The quantity $$\pth{\inj(\sfM,G_1)-\expect{\inj(\sfM,G_1)}}\pth{\inj(\sfM,G_2)-\expect{\inj(\sfM,G_2)}}$$ indicates the correlation between $\inj(\sfM,G_1)$ and $\inj(\sfM,G_2)$, which serves as a basis for distinguishing $\maH_0$ from $\maH_1$.
Naturally, the definition of homomorphism numbers can be extended to the case where $G$ is a weighted graph associated with vertex set $V(G)$ and weighted edge set $\sth{\beta_{uv}(G):u,v\in V(G)}$. For any mapping $\varphi:V(\sfM)\mapsto V(G)$, we define $\hom_\varphi(\sfM,G) = \prod_{uv\in E(\sfM)}\beta_{\varphi(u)\varphi(v)}(G)$ and \begin{align}\label{eq:def_of_inj}
    \inj(\sfM,G) = \sum_{\substack{\varphi:V(\sfM)\mapsto V(G)\\\varphi\text{ injective}}} \hom_\varphi(\sfM,G).
\end{align}

Given a graph $G$, we first center the weights and obtain a weighted graph $\bar{G}$ with weighted edges 
$\beta_{uv}(\bar{G}) = \indc{uv\in E(G)} - \expect{\indc{uv\in E(G)}}$ for $u,v\in V(G)$, where $\expect{\indc{uv\in E(G)}}$ can be estimated by the average degree of the graph. 
Then, for a given motif family $\maM$, our test statistic is defined as \begin{align}\label{eq:est-sub-count}
    \maT_{\maM }\pth{G_1,G_2} = \sum_{\sfM \in \maM } \omega_\sfM\, \inj(\sfM,\bar{G}_1)\inj(\sfM,\bar{G}_2),
\end{align}
where $\omega_\sfM$ is a weight function to be specified.
This test statistic can be interpreted as an inner product between the two vectors $\qth{\inj(\sfM,\bar{G}_1)}_{\sfM\in \maM}$ and $\qth{\inj(\sfM,\bar{G}_2)}_{\sfM\in \maM}$. 
By picking an appropriate threshold $\tau$, we define the test that rejects the null hypothesis $\maH_0$ whenever $\maT_\maM(G_1,G_2)\ge \tau$. We will theoretically analyze the resulting Type I and Type II errors in Section~\ref{sec:general-motif}.


A richer motif family captures more graph properties and can strengthen the effectiveness of the test at a higher computational cost. 
Our motif-counting statistic has a computational cost at most $O(n^{\en(\maM)})$, where $\en(\maM)\triangleq\max_{\sfM\in \maM}\en(\sfM)$ is the maximum number of edges among motifs in the family $\maM$.
Our theory requires $\en(\maM)\ge f(0.05)$ for some function $f$ to achieve a prescribed error probability $0.05$; see Sections~\ref{sec:general-motif},~\ref{sec:bd-sub-count} and~\ref{sec:refined-results} for further details. This setting illustrates a fundamental trade-off between statistical accuracy and computational efficiency: as $\en(\maM)\to\infty$, the sum of Type I and Type II errors vanishes, but at the expense of an increasing runtime of $O(n^{\en(\maM)})$.

The previous work~\cite{mao2024testing} adopted the motif counting statistic with the tree motifs and showed that detection is possible when the correlation coefficient is beyond some constant under the \ER random graph model. 
The effectiveness of tree counting statistics relies significantly on the tree-like substructures inherent in the graph model.
As a result, such statistics may become less effective for graph models or datasets that lack a tree-like structure,
raising the natural question of whether we can count more general motifs.
In this paper, we consider bounded degree motifs that are commonly observed in practice. Let $\maM(N_\sfe,d)$ denote the set of all connected bounded degree motifs with $N_\sfe$ edges and maximum degree bounded by $d$. For example, 
\[\maM_{3,2} = \sth{\begin{tikzpicture}[scale=0.5, baseline={([yshift=-.5ex]current bounding box.center)}]
    \nodebullet  (a) at (0,0) {};
    \nodebullet (b) at (1,0) {};
    \nodebullet (c) at (2,0) {};
    \nodebullet (d) at (3,0) {};
    \draw (a) -- (b) -- (c) -- (d);
\end{tikzpicture}\,,\quad 
\begin{tikzpicture}[scale=0.5, baseline={([yshift=-.5ex]current bounding box.center)}]
    \nodebullet (b) at (0,0) {};
    \nodebullet (c) at (0.5,-0.866) {};
    \nodebullet (d) at (-0.5,-0.866) {};
    \draw (c) -- (d) -- (b) -- (c);
\end{tikzpicture}
},\quad 
\maM_{4,3} = \left\{ 
\begin{tikzpicture}[scale=0.5, baseline={([yshift=-.5ex]current bounding box.center)}]
    \nodebullet  (a) at (0,0) {};
    \nodebullet (b) at (1,0) {};
    \nodebullet (c) at (2,0) {};
    \nodebullet (d) at (3,0) {};
    \nodebullet (e) at (4,0) {};
    \draw (a) -- (b) -- (c) -- (d) -- (e);
\end{tikzpicture}\, ,
\quad
\begin{tikzpicture}[scale=0.5, baseline={([yshift=-.5ex]current bounding box.center)}]
    \nodebullet (a) at (0,0) {};
    \nodebullet (b) at (1,0) {};
    \nodebullet (c) at (1,1) {};
    \nodebullet (d) at (0,1) {};
    \draw (a) -- (b) -- (c) -- (d) -- (a);
\end{tikzpicture}\, ,
\quad
\begin{tikzpicture}[scale=0.5, baseline={([yshift=-.5ex]current bounding box.center)}]
    \nodebullet (a) at (0,0) {};
    \nodebullet (b) at (1,0) {};
    \nodebullet (c) at (2,0) {};
    \nodebullet (e) at (1,1) {};
    \nodebullet (f) at (1,2) {};
    \draw (a) -- (b) -- (c);
    \draw (b) -- (e);
    \draw (e) -- (f);
\end{tikzpicture}\, ,\quad 
\begin{tikzpicture}[scale=0.5, baseline={([yshift=-.5ex]current bounding box.center)}]
    \nodebullet (a) at (0,1) {};
    \nodebullet (b) at (0,0) {};
    \nodebullet (c) at (0.5,-0.866) {};
    \nodebullet (d) at (-0.5,-0.866) {};
    \draw (c) -- (d) -- (b) -- (c);
    \draw (b) -- (a);
\end{tikzpicture}
\right\}.
\]
While the bounded degree counting test statistic $\maT_{\maM(N_\sfe,d)}$  remains valid for detection, we will consider two subsets of bounded degree motifs $\maM(N_\sfv,N_\sfe,d),\overline{\maM}(N_\sfv,N_\sfe,d)\subseteq \maM(N_\sfe,d)$ that are more simplified and easier to analyze. See Sections~\ref{sec:bd-sub-count} and~\ref{sec:refined-results} for further details. Indeed, our approach unifies and generalizes several existing motif counting methods: \cite{barak2019nearly} counted balanced graphs, \cite{mao2024testing} counted trees, and \cite{jin2025counting} counted cycles. By counting all bounded degree motifs, our statistic subsumes these as special cases and captures a richer range of structural correlations.
\subsection{Related Work}


\emph{Polynomial-time algorithms and computational hardness.}
It was shown in~\cite{barak2019nearly} that counting balanced subgraphs succeeds in detecting correlation in correlated \ER graphs for any constant $\rho$, provided that the edge probability $p$ lies in a suitable regime. Extending this line of work,~\cite{mao2024testing} demonstrated that counting trees achieves successful correlation detection over a broader range of $p$, as long as $\rho$ exceeds a fixed constant.
From the computational-hardness perspective, the low-degree conjecture—motivated by the sum-of-squares framework—is widely believed to offer a unifying approach to predicting computational lower bounds in a broad class of high-dimensional statistical problems~\citep{hopkins2017efficient,hopkins2018statistical,schramm2022computational,sohn2025sharp}. This conjecture has yielded tight hardness predictions for many problems, including graph matching, planted clique, planted dense subgraph, community detection, tensor PCA, and sparse PCA~\citep{hopkins2017efficient,hopkins2017power,hopkins2018statistical,bandeira2019computational,schramm2022computational,ding2023low,kunisky2024tensor,dhawan2025detection,li2025algorithmic}. 

\emph{Information-theoretic analysis.}
\cite{wu2023testing} characterized a sharp detection threshold—at which the optimal testing error exhibits a phase transition—by analyzing the maximum likelihood test for dense correlated \ER graphs with edge probability $p=n^{-o(1)}$. For sparse graphs with $p=n^{-\Omega(1)}$, they also derived the threshold up to constant factors. Building on this, \cite{ding2023detection} further sharpened the sparse-graph threshold via an analysis of the densest subgraphs. More recently, \cite{huang2025sample,huang2026information} investigated correlation detection for induced subgraphs under the Gaussian Wigner model and the \ER model.

\emph{Graph matching.}
Closely related to correlation detection is the graph matching problem, which seeks a node correspondence maximizing the edge agreement (or edge correlation) between a pair of correlated graphs~\citep{conte2004thirty}. A variety of polynomial-time algorithms have been proposed, including approaches based on subgraph counting~\citep{mao2025random,maier2025asymmetric}, neighborhood statistics~\citep{dai2019analysis,ding2021efficient,mao2021random}, spectral methods~\citep{umeyama1988eigendecomposition,singh2008global,fan2019spectral,ganassali2022spectral}, convex relaxations~\citep{aflalo2015convex,vogelstein2015fast,varma2025graph}, greedy algorithms~\citep{du2025algorithmic}, and iterative refinements~\citep{piccioli2022aligning,ding2023polynomial,ganassali2024statistical}. From an information-theoretic viewpoint, sharp thresholds and phase transitions for exact and partial recovery, along with impossibility regimes, have been investigated in a series of works~\citep{cullina2016improved,cullina2017exact,cullina2020partial,ganassali2020tree,ganassali2021impossibility,wu2022settling,ding2023matching,du2025optimal,hall2023partial}.

\subsection{Notation and Paper Organization}

For any $n\in \mathbb{N}$, let $[n]\triangleq \{1,2,\cdots,n\}$. 
We use standard asymptotic notation: for two positive sequences $\{a_n\}$ and $\{b_n\}$, we write $a_n = O(b_n)$ or $a_n\lesssim b_n$, if $a_n\le C b_n$ for some absolute constant $C$ and all $n$; $a_n = \Omega(b_n)$ or $a_n\gtrsim b_n$, if $b_n = O(a_n)$; $a_n = \Theta(b_n)$ or $a_n\asymp b_n$, if $a_n = O(b_n)$ and $a_n = \Omega(b_n)$; $a_n = o(b_n)$ or $b_n =\omega(a_n)$, if $a_n/b_n\to 0$ as $n\to \infty$.
Let $[x]$ denote the greatest integer less than or equal to $x$.

For a given weighted graph $G$, let $V(G)$ denote its vertex set and $E(G)$ its edge set. 
We write $uv$ to represent an edge $\{u,v\}$, and $\beta_{e}(G)$ for the weight of the edge $e$. 
 For an unweighted graph $G$, we define $\beta_{uv}(G)=\indc{uv\in E(G)}$.
Let $\vn(G)=|V(G)|$ denote the number of vertices in $G$, and $\en(G)=\sum_{e\in E(G)} \beta_{e}(G)$ the total weight of its edges. For any bijective mapping $\pi:V(G_1)\mapsto V(G_2)$, we define  $\pi(uv) = \pi(u)\pi(v)$ for any $u,v\in V(G_1)$. For simplicity, we write $\pi(e)$ to denote $\pi(uv)$ for any edge $e =uv$.
We write $\sfH_1 = \sfH_2$ if and only if they are isomorphic, that is, there exists a bijection $\pi:V(\sfH_1)\mapsto V(\sfH_2)$ such that $uv\in E(\sfH_1)$ if and only if $\pi(u)\pi(v)\in E(\sfH_2)$.
For any bijective mapping $\pi: V(G_1)\mapsto V(G_2)$ and subgraph $\sfH\subseteq G_1$, we define $\pi(\sfH)$ as the graph with\begin{align}\label{eq:def-piH}
    E(\pi(\sfH)) = \sth{\pi(u)\pi(v):uv\in E(\sfH)},V(\pi(\sfH)) = \sth{\pi(v):v\in V(\sfH)}.
\end{align}
For two graphs $G$ and $G'$, let $G\cap G'$ denote the graph with\begin{align}\label{eq:intersec-graph}
    E(G\cap G') = E(G)\cap E(G'), V(G\cap G') = \sth{v\in V(G)\cup V(G'):\exists u, uv\in E(G\cap G')}.
\end{align}
Let $G\cup G'$ denote the graph with
\begin{align}\label{eq:union-graph}
    E(G\cup G') = E(G)\cup E(G'),V(G\cup G') = V(G)\cup V(G').
\end{align}
Let $G\triangle G'$ denote the graph with 
\begin{align}\label{eq:triangle-graph}
    E(G\triangle G') = E(G)\triangle E(G'), V(G\triangle G') = \sth{v\in V(G)\cup V(G'):\exists u, uv\in E(G\triangle G')}.
\end{align}
The intersection $G \cap G'$ represents the subgraph consisting of all edges shared by $G$ and $G'$, along with the vertices incident to those edges.
The symmetric difference $G \triangle G'$ represents the subgraph containing edges that appear in exactly one of $G$ or $G'$. Moreover, we have \begin{align*}
    |V(G\triangle G')| = \vn(G)+\vn(G')-2|V(G)\cap V(G')|+|V(G\triangle G')\cap (V(G)\cap V(G'))|.
\end{align*}




The remainder of the paper is organized as follows. 
In Section~\ref{sec:general-motif}, we establish theoretical guarantees for the proposed
statistics under general conditions.
Section~\ref{sec:bd-sub-count} introduces a specific motif family and shows
that the corresponding statistic performs well on graph models.
In Section~\ref{sec:refined-results}, we extend the condition $p\ge n^{-2/3}$ to the sparser regime $p\ge n^{-1+\delta}$ by counting a more delicate class of motifs.
In Section~\ref{sec:num-results}, we provide simulation studies. We finish in Section~\ref{sec:discussion}. Some technical proofs are deferred to the Appendices.


\section{General Motif Counting Statistic}\label{sec:general-motif}


In this section, we establish the basic properties for the motif counting statistic $\maT_\maM$ in~\eqref{eq:est-sub-count} under the \ER model for general motif family $\maM$. To achieve the detection criterion~\eqref{eq:detection-cretiria}, it is crucial to select an appropriate motif family $\maM$. To this end, we will show that, by choosing appropriate weights $\omega_\sfM$ in $\maT_\maM$,
\[
\E_{\maP_1}\!\left[\maT_\maM\right] = \var_{\maP_0}\!\left[\maT_\maM\right] = \sum_{\sfM\in \maM}\rho^{2\en(\sfM)}.
\]
Since $\E_{\maP_0}[\maT_\maM] = 0$, the signal-to-noise ratio of the test statistic is given by
\begin{align}\label{eq:signal-to-noise}
    \frac{\E_{\maP_1}[\maT_\maM]-\E_{\maP_0}[\maT_\maM]}{\sqrt{\var_{\maP_0}[\maT_\maM]}}
= \sqrt{\sum_{\sfM\in \maM}\rho^{2\en(\sfM)}}.
\end{align}
As a result, by choosing the threshold $\tau=\frac{1}{2}\E_{\maP_1}[\maT_\maM]$, the quantity $\sum_{\sfM\in \maM}\rho^{2\en(\sfM)}$ plays a key role in controlling the Type~I error $\maP_0(\maT_\maM\ge \tau)$ by Chebyshev's inequality:
\begin{align}\label{eq:Chebyshev}
    \maP_0(\maT_\maM\ge \tau)\le \frac{4\var_{\maP_0}(\maT_\maM)}{(\E_{\maP_1}[\maT_\maM])^2} = \frac{4}{\sum_{\sfM\in \maM}\rho^{2\en(\sfM)}}.
\end{align}
On the other hand, controlling the Type~II error requires bounding $\maP_1(\maT_\maM<\tau)$. By Chebyshev's inequality,
\begin{align}
    \maP_1(\maT_\maM<\tau)\le \maP_1\pth{|\maT_\maM-\E_{\maP_1}[\maT_\maM]|\ge \frac{1}{2}\E_{\maP_1}[\maT_\maM]}\le \frac{4\var_{\maP_1}(\maT_\maM)}{(\E_{\maP_1}[\maT_\maM])^2},
\end{align}
so it remains to control $\var_{\maP_1}(\maT_\maM)$. Specifically, we show that motif counting based on the following family is sufficient for detection.

\begin{definition}\label{def:admissible}
We say that a motif family $\maM$ is \emph{$C$-admissible} if \begin{enumerate}
\item \label{cond:connect}For all $\sfM\in \maM$, $\sfM$ is connected;
    \item\label{cond:num} There exists $C=o\pth{\frac{\log n}{\max\pth{\log\log n,-\log\rho}}}$ such that $\vn(\sfM)\vee\en(\sfM)\le C$ for all $\sfM\in \maM$;
    \item\label{cond:signal-strength} $\sum_{\sfM\in \maM} \rho^{2\en(\sfM)}\ge 400$;
    \item\label{cond:subgraph} There exists a small constant $\epsilon_0$ such that, $n^{\vn(\sfM')} p^{\en(\sfM')}\ge n^{\epsilon_0}$ for all $\sfM\in \maM$ and subgraph $\emptyset\neq \sfM'\subseteq \sfM$. 
\end{enumerate}
\end{definition}

Condition~\ref{cond:num} ensures that the size of each motif is bounded by $C$. Under this condition, the computation time of the statistic $\maT_\maM$ is $O(n^{C})$. 
Condition~\ref{cond:signal-strength} sets a lower bound on the overall signal strength, requiring that $\sum_{\sfM\in \maM} \rho^{2\en(\sfM)}\ge 400$. This requirement can be challenging to meet in practice: when the number of edges $\en(\sfM)$ grows, although $\maM$ contains more motifs, the term $\rho^{2\en(\sfM)}$ decreases quickly if $\rho$ is small. Consequently, one needs to balance the size $|\maM|$ and the quantity $\rho^{2\en(\sfM)}$ to ensure the family still carries enough signal.
Thus, Condition~\ref{cond:signal-strength} directly contributes to the Type~I error analysis.
Condition~\ref{cond:subgraph} is a technical requirement that ensures control of the variance $\var_{\maP_1}[\maT_\maM]$, which is essential for Type~II error analysis.
Intuitively, when a motif contains overly dense subgraphs, their occurrences become highly dependent, which inflates the variance $\var_{\maP_1}\!\left[\maT_\maM\right]$ and undermines the power of the test.
The requirement
\(
n^{\vn(\sfM')}\, p^{\en(\sfM')}\ge n^{\epsilon_0}
\)
ensures that each substructure appears with sufficient frequency to stabilize the test statistic, preventing such variance explosion.
This condition naturally motivates the use of bounded degree motifs, whose subgraphs are not excessively dense and thus allow for uniform variance control across the motif family.
The following theorem provides a theoretical guarantee for the counting statistic based on any \emph{$C$-admissible} motif family.


\begin{theorem}\label{thm:admissible}
    For any \emph{$C$-admissible} motif family $\maM$, there exist $\tau,\omega_\sfM\in \mathbb{R}$ such that,
    \begin{align*}
        \maP_0\pth{\maT_\maM\ge\tau}+\maP_1(\maT_\maM< \tau)\le 0.05.
    \end{align*}
\end{theorem}
Theorem~\ref{thm:admissible} shows that the test statistic based on any \emph{$C$-admissible} motif family suffices for detection. In Section~\ref{sec:bd-sub-count}, we will construct a specific sub-family of bounded degree motifs $\maM'\subseteq \maM(N_\sfe,d)$ and prove that $\maM'$ is \emph{$C$-admissible}. Consequently, the statistic $\maT_{\maM'}$ achieves successful detection. In the following, we outline a general recipe for controlling the Type I and Type II errors on $\maT_\maM$.

\begin{remark}\label{rmk:005}
It follows from~\eqref{eq:signal-to-noise} and~\eqref{eq:Chebyshev} that the Type~I error can be made arbitrarily small provided that
$\sum_{\sfM\in\maM}\rho^{2\en(\sfM)}$ is sufficiently large.
For the Type~II error, Proposition~\ref{prop:admissible-TypeII} implies that it can also be made arbitrarily small.
In fact, for the constructions in Sections~\ref{sec:bd-sub-count} and~\ref{sec:refined-results}, Condition~\ref{cond:signal-strength} can be strengthened by replacing the constant $400$ with any prescribed constant (while keeping the same assumptions on $p$ and $\rho$), which in turn allows the target error level (e.g., $0.05$) to be replaced by any fixed constant.
Since the proof is identical up to constant changes, we omit the details.
\end{remark}

\subsection{Type I Error Control via Signal Score Estimation}

In this subsection, we show that the Type I error for the motif counting statistic with a \emph{$C$-admissible} motif family can be bounded by $O\pth{\frac{1}{\sum_{\sfM\in \maM} \rho^{2\en(\sfM)}}}$, where the quantity $\sum_{\sfM\in \maM} \rho^{2\en(\sfM)}$ is defined as the \emph{signal score} of the motif family, since it captures the strength of the signal-to-noise ratio. In order to distinguish $\maH_0$ from $\maH_1$ by the test statistic $\maT_\maM$, one natural choice for $\tau$ is to pick $\tau = \frac{1}{2}\pth{\E_{\maP_0}\qth{\maT_\maM}+\E_{\maP_1}\qth{\maT_\maM}}=\frac{1}{2}\E_{\maP_1}\qth{\maT_\maM}$. Applying Chebyshev's inequality yields the Type I error: \begin{align*}
    \maP_0\pth{\maT_\maM\ge \tau}\le \frac{\var_{\maP_0}\qth{\maT_\maM}}{\tau^2}=\frac{4\var_{\maP_0}\qth{\maT_\maM}}{\pth{\E_{\maP_1}\qth{\maT_\maM}}^2}.
\end{align*}
With appropriately chosen weights $\omega_\sfM$, we have $\var_{\maP_0}\qth{\maT_\maM} = \E_{\maP_1}\qth{\maT_\maM} =\sum_{\sfM\in \maM} \rho^{2\en(\sfM)}$. Therefore, the Type I error is bounded by $\frac{4}{\sum_{\sfM\in \maM} \rho^{2\en(\sfM)}}$. 
Let $\aut(\sfM)$ denote the number of automorphisms of $\sfM$.
In particular, we have the following proposition.
\begin{proposition}\label{prop:TypeI-admissible}
    For the statistic $\maT_\maM$ with weight $\omega_\sfM = \frac{\rho^{\en(\sfM)} (n-\vn(\sfM))!}{n! (p(1-p))^{\en(\sfM)}\aut(\sfM)}$, if $\tau = \frac{1}{2}\E_{\maP_1}\qth{\maT_\maM}$, then 
    \begin{align*}
        \maP_0\pth{\maT_\maM\ge \tau}\le \frac{4}{\sum_{\sfM\in \maM} \rho^{2\en(\sfM)}}.
    \end{align*}
\end{proposition}
The proof of Proposition~\ref{prop:TypeI-admissible} is deferred to Appendix~\ref{apd:proof-prop-admi-typeI}.
In view of Proposition~\ref{prop:TypeI-admissible}, the Type I error can be controlled as long as the \emph{signal score} is sufficiently large. 
If the motif family $\maM$ is \emph{$C$-admissible},
Condition~\ref{cond:signal-strength} in Definition~\ref{def:admissible} ensures that the Type I error is bounded by $0.01$. In practice, however, the parameters $p$ and $\rho$ may not be known a priori. A natural approach is to estimate these parameters from the observed graph data. While the edge probability $p$ can be reasonably approximated by the empirical edge density, estimating the correlation coefficient $\rho$ is substantially more challenging. To circumvent this challenge, we often restrict our attention to motif families satisfying
\[
\en(\sfM) = \en(\sfM') \quad \text{for all } \sfM,\sfM' \in \maM.
\]
Since all motifs in \(\mathcal M\) have the same number of edges, the factor $\rho^{\en(\sfM)}\,[p(1-p)]^{-\en(\sfM)}$
is constant across \(\sfM\in \maM\) and can be absorbed into a global normalization. Hence, instead of $\omega_{\sfM}=\frac{\rho^{\en(\sfM)}(n-\vn(\sfM))!}{n!(p(1-p))^{\en(\sfM)}\aut(\sfM)}$,
we may take the  choice $\omega_{\sfM}=\frac{(n-\vn(\sfM))!}{n!\aut(\sfM)}$.
With this choice, the statistic $\maT_\maM$ does not depend on the unknown correlation coefficient \(\rho\) (nor on \(p\) through the weights); the omitted constant is absorbed by the final normalization of the test.
A concrete example of such a \emph{$C$-admissible} motif family with equal edge numbers will be presented in Section~\ref{sec:bd-sub-count}.

\subsection{Type II Error Control via Second Moment Analysis}

In this subsection, we establish the main results for controlling the Type II error. By Chebyshev's inequality, we obtain \begin{align*}
    \maP_1\pth{\maT_\maM<\tau} \le \maP_1\pth{\left| \maT_\maM-\E_{\maP_1}\qth{\maT_\maM}\right|>\tau}\le \frac{4\var_{\maP_1}\qth{\maT_\maM}}{\pth{\E_{\maP_1}\qth{\maT_\maM}}^2}.
\end{align*}
With a suitable choice of weights $\omega_\sfM$, we have already shown that $\E_{\maP_1}[\maT_\maM]$ is characterized by the \emph{signal score} $\sum_{\sfM\in \maM} \rho^{2\en(\sfM)}$. The remaining task is to control the variance, which requires estimating the second moment $\E_{\maP_1}[\maT_\maM^2]$. This analysis involves carefully handling the correlated terms under $\maP_1$, especially the off-diagonal terms in the second-moment expansion. Indeed, Condition~\ref{cond:subgraph} in the definition of a \emph{$C$-admissible} motif family is precisely designed to ensure such control. In particular, we obtain the following Proposition for controlling the Type II error.

\begin{proposition}\label{prop:admissible-TypeII}
    For motif counting statistic $\maT_\maM$ with \emph{$C$-admissible}  family $\maM$ and weight $\omega_\sfM = \frac{\rho^{\en(\sfM)}(n-\vn(\sfM))!}{n!(p(1-p))^{\en(\sfM)}\aut(\sfM)}$, if $\tau = \frac{1}{2}\E_{\maP_1}\qth{\maT_\maM}$, then \begin{align*}
        \maP_1\pth{\maT_\maM<\tau}\le 4\Bigg(3n^{-\epsilon_0/2}(4C)^{8C}\rho^{-2C}+\frac{\exp\pth{\frac{C^2}{n-2C+1}}+1}{\sum_{\sfM\in \maM}\rho^{2\en(\sfM)}}+\exp\pth{\frac{C^2}{n-2C+1}}-1\Bigg).
    \end{align*}
\end{proposition}
The proof of Proposition~\ref{prop:admissible-TypeII} is deferred to Appendix~\ref{apd:proof-prop-admi-typeII}. In view of Proposition~\ref{prop:admissible-TypeII}, the term $\exp\left(\tfrac{C^2}{n-2C+1}\right)-1$ can be upper bounded by $\tfrac{1}{400}$ when $n$ is sufficiently large. Moreover, when $\rho$ is a constant, the term $3n^{-\epsilon_0/2}(4C)^{8C}\rho^{-2C}$ can also be upper bounded by $\tfrac{1}{400}$ for large $n$. Consequently, the Type II error is bounded by $0.04$ when $n$ is large enough. Combining this with Proposition~\ref{prop:TypeI-admissible}, we are now ready to prove Theorem~\ref{thm:admissible}.

\begin{proof}[Proof of Theorem~\ref{thm:admissible}]
Pick \begin{align*}
    \omega_\sfM = \frac{\rho^{\en(\sfM)}(n-\vn(\sfM))!}{n!(p(1-p))^{\en(\sfM)}\aut(\sfM)},\quad \tau = \frac{1}{2}\E_{\maP_1}\qth{\maT_\maM} = \frac{1}{2}\sum_{\sfM\in\maM} \rho^{2\en(\sfM)}.
\end{align*}
By the Condition~\ref{cond:signal-strength} in Definition~\ref{def:admissible}, we have $\sum_{\sfM\in \maM} \rho^{2\en(\sfM)}\ge 400$ for a \emph{$C$-admissible} motif family $\maM$.
By Proposition~\ref{prop:TypeI-admissible}, the Type I error is upper bounded by \begin{align}\label{eq:TypeI-error}
    \maP_0\pth{\maT_\maM\ge \tau}\le \frac{4}{\sum_{\sfM\in \maM} \rho^{2\en(\sfM)}}\le 0.01.
\end{align}
By Proposition~\ref{prop:admissible-TypeII}, the Type II error is upper bounded by \begin{align}\label{eq:TypeII-error}
    \maP_1\pth{\maT_\maM<\tau}\le 4\Bigg(3n^{-\epsilon_0/2}(4C)^{8C}\rho^{-2C}+\frac{\exp\pth{\frac{C^2}{n-2C+1}}+1}{\sum_{\sfM\in \maM}\rho^{2\en(\sfM)}}+\exp\pth{\frac{C^2}{n-2C+1}}-1\Bigg).
\end{align}
For any $C = o\pth{\frac{\log n}{\max\sth{\log\log n,-\log\rho}}}$ and sufficiently large $n$, we have $3n^{-\epsilon_0/2}(4C)^{8C}\rho^{-2C}\le \frac{1}{400}$ and $\exp\pth{\frac{C^2}{n-2C+1}}-1\le \frac{1}{400}$. By~\eqref{eq:TypeII-error}, we obtain \begin{align*}
    \maP_1\pth{\maT_\maM<\tau}&\le 4\pth{\frac{1}{400}+\frac{1+1/400}{\sum_{\sfM\in \maM}\rho^{2\en(\sfM)}}+\frac{1}{400}}\\&\le 4\pth{\frac{1}{400}+\frac{2}{400}+\frac{1}{400}}=0.04,
\end{align*}
where the last inequality applies the fact that  $\sum_{\sfM\in \maM} \rho^{2\en(\sfM)}\ge 400$ for a \emph{$C$-admissible} motif family $\maM$.
Consequently, combining this with~\eqref{eq:TypeI-error}, we obtain 
\begin{equation*}
    \maP_0\pth{\maT_\maM\ge \tau}+\maP_1\pth{\maT_\maM<\tau}\le 0.05.\qedhere
\end{equation*}
\end{proof}


\subsection{Enumeration of Motif Families}
\label{subsec:comb-contrib}

In this subsection, we highlight the combinatorial ingredients underlying our detection results.
To meet the detection criterion, it follows from~\eqref{eq:signal-to-noise} and~\eqref{eq:Chebyshev} that the quantity $\sum_{\sfM\in \maM}\rho^{2\en(\sfM)}$ plays a central role, where we often take $\en(\sfM)$ to be the same across the family, namely $\en(\sfM)=N_\sfe$ for all $\sfM\in\maM$.

A key observation from prior work is that the growth rate of the motif family can become the bottleneck.
For instance, \cite{wu2023testing} counted trees, whose family size is on the order of $(\alpha^{-1})^{N_\sfe}$~\cite{otter1948number}, where $\alpha\approx 0.338$ is Otter's constant.
In that case, unless $\rho^2$ exceeds a constant threshold, the sum $\sum_{\sfM\in\maM}\rho^{2\en(\sfM)}$ remains small as $N_\sfe\to\infty$.
This motivates seeking motif classes whose cardinality grows on the order of $N_\sfe^{\Theta(N_\sfe)}$, which is sufficient to achieve detection for any constant correlation $\rho=\Omega(1)$.


Beyond trees, classical enumeration results show that bounded degree graph classes (e.g., graphs with prescribed degree sequences) already have super-exponential cardinality $N_{\sfe}^{\Theta(N_{\sfe})}$; see, e.g.,~\cite{mckay1991asymptotic,liebenau2017asymptotic}. However, these results concern labeled graphs and are not tailored to the structural regularity required by our moment bounds. In particular, we need a large family of motifs that are  distinct up to isomorphism (rather than merely labeled) and for which automorphism sizes and overlap patterns are uniformly controlled, so that second-moment contributions from overlapping embeddings can be bounded. Moreover, we require that every nonempty subgraph $\sfM'\subseteq \sfM$ satisfies the worst-case sparsity constraint in Condition~\ref{cond:subgraph}. These additional requirements motivate us to go beyond generic bounded degree enumeration and to construct explicit motif families that retain the $N_{\sfe}^{\Theta(N_{\sfe})}$ multiplicity while enforcing the necessary subgraph constraints.

\noindent\textbf{Bounded degree motif families via bijective matchings.}
Our construction in Section~\ref{sec:bd-sub-count} achieves the growth rate $|\maM| = N_\sfe^{\Theta(N_\sfe)}$ by introducing a combinatorial degree of freedom in the form of a bijective mapping within the motif.
Specifically, we identify a collection of about $N_\sfe/d$ building blocks (with $d$ the maximum degree of the motif) and choose an arbitrary bijection between two copies of this collection.
The number of choices is about $(N_\sfe/d)!$, which yields the desired $N_\sfe^{\Theta(N_\sfe)}$ factor when $d=O(1)$, thereby motivating bounded degree motifs.

At the same time, our moment bounds require the technical condition that $n^{\vn(\sfM')}p^{\en(\sfM')}\ge n^{\epsilon_0}$ for all $\emptyset\neq \sfM'\subseteq \sfM$ and some constant $\epsilon_0$; see Condition~\ref{cond:subgraph}.
When $d=2$, the admissible motifs are essentially limited to paths and cycles, which is too restrictive for our purposes, and we therefore focus on $d\ge 3$.
However, for the bounded degree family in Section~\ref{sec:bd-sub-count}, Condition~\ref{cond:subgraph} is guaranteed only when $p\ge n^{-2/3}$, whereas the optimal sparsity range targeted by our results is $p\ge n^{-1+\delta}$.

\noindent\textbf{Refined motif families for $p\ge n^{-1+\delta}$.}
To push the sparsity requirement down to $p\ge n^{-1+\delta}$, we construct in Section~\ref{sec:refined-results} a more delicate motif class in which every subgraph $\sfM'\subseteq \sfM$ satisfies a suitable sparsity constraint, informally of the form $\en(\sfM')/\vn(\sfM')\le 1/(1-\delta)$.
The key idea is to separate vertices of degree $3$ by long stretches of degree-$2$ vertices, which forces any subgraph to remain sufficiently sparse and hence improves the worst-case behavior in Condition~\ref{cond:subgraph}.
Meanwhile, the construction retains an arbitrary bijective mapping component, so that the refined family still has size $|\maM| = N_\sfe^{\Theta(N_\sfe)}$.

\section{Admissible Bounded Degree Motifs: Construction and Detection Guarantees}\label{sec:bd-sub-count}
We have established in Section~\ref{sec:general-motif} that counting \emph{$C$-admissible} motifs is sufficient for successful detection. In this section, we present an explicit construction of a motif family that satisfies the \emph{$C$-admissible} conditions.
In Section~\ref{sec:method}, we introduced the bounded degree motif family $\maM(N_\sfe, d)$, defined as the collection of motifs with $N_\sfe$ edges and maximum degree at most $d$. However, directly analyzing all motifs in $\maM(N_\sfe, d)$ is challenging due to the large heterogeneity of their structures and correlations. To obtain a simpler yet effective statistic, we focus on a more structured subclass.
Specifically, we construct a motif family $\maM(N_\sfv, N_\sfe, d) \subseteq \maM(N_\sfe, d)$ consisting of motifs with $N_\sfv$ vertices, $N_\sfe$ edges, and maximum degree $d$. In Section~\ref{subsec:bd-degree-specific}, we show that the associated motif-counting statistic $\maT_{\maM(N_\sfv, N_\sfe, d)}$ is both polynomial-time computable and \emph{$C$-admissible}, and hence is sufficient for detection. Furthermore, we establish that the broader bounded degree motif statistic $\maT_{\maM(N_\sfe, d)}$ is also \emph{$C$-admissible} in Section~\ref{subsec:general-bd-degree}.

The overall proof flow is summarized in Figure~\ref{fig:logic-flow}:
Section~\ref{sec:general-motif} shows the toolbox
established in  Theorem~\ref{thm:admissible}; Section~\ref{subsec:bd-degree-specific} proves that
$\mathcal{T}_{\mathcal{M}(N_{\mathsf v},N_{\mathsf e},d)}$ is \emph{$N_\sfe$-admissible} and yields detection in Theorem~\ref{thm:bd-degree-motif-main}; and
Section~\ref{subsec:general-bd-degree} upgrades the argument to the cleaner statistic
$\mathcal{T}_{\mathcal{M}(N_{\mathsf e},d)}$, which is \emph{$(N_\sfe{+}1)$-admissible} and also successful for detection, as stated in Theorem~\ref{thm:bd-degree-motif-general}.




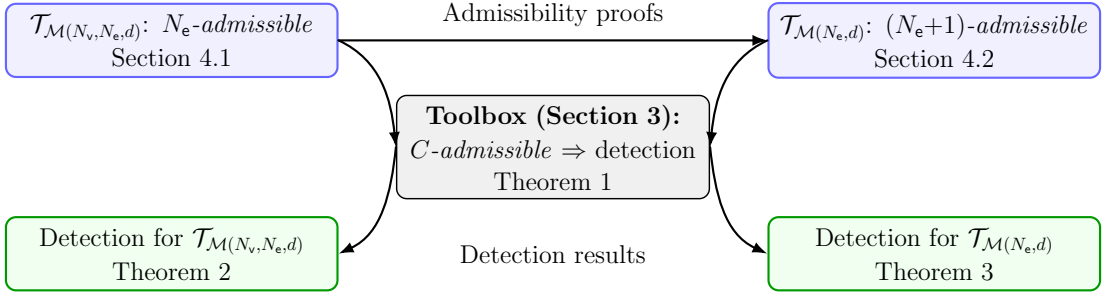
\begin{figure}[htbp]
\centering
\resizebox{0.95\linewidth}{!}{%
\begin{tikzpicture}[font=\Large, node distance=22mm and 32mm,
  >=stealth,
  nodeA/.style = {rectangle, rounded corners=2mm, draw=blue!60, fill=blue!6,
                  very thick, minimum height=12mm, text width=61mm,
                  align=center, inner sep=6pt},
  nodeD/.style = {rectangle, rounded corners=2mm, draw=green!60!black, fill=green!6,
                  very thick, minimum height=12mm, text width=61mm,
                  align=center, inner sep=6pt},
  toolbox/.style = {rectangle, rounded corners=2mm, draw=black, fill=black!6,
                    thick, minimum height=10mm, text width=58mm,
                    align=center, inner sep=5pt},
  flow/.style = {->, very thick, line cap=round},
  lab/.style  = {font=\scriptsize, text=black!65},
  ttl/.style  = {font=\Large}
]

\node[nodeA] (A1) at (-7.5,  2.1)
{$\mathcal{T}_{\mathcal{M}(N_{\mathsf v},N_{\mathsf e},d)}$: \emph{$N_{\mathsf{e}}$-admissible}\\[1pt] Section~\ref{subsec:bd-degree-specific}};

\node[nodeD] (A2) at (-7.5, -2.1)
{Detection for $\mathcal{T}_{\mathcal{M}(N_{\mathsf v},N_{\mathsf e},d)}$\\[1pt] Theorem~\ref{thm:bd-degree-motif-main}};

\node[nodeA] (D1) at ( 7.5,  2.1)
{$\mathcal{T}_{\mathcal{M}(N_{\mathsf e},d)}$: \emph{$(N_{\mathsf{e}}{+}1)$-admissible}\\[1pt] Section~\ref{subsec:general-bd-degree}};

\node[nodeD] (D2) at ( 7.5, -2.1)
{Detection for $\mathcal{T}_{\mathcal{M}(N_{\mathsf e},d)}$\\[1pt] Theorem~\ref{thm:bd-degree-motif-general}};

\node[ttl] at ($(A1)!0.5!(D1)+(0,5mm)$) {Admissibility proofs};
\node[ttl] at ($(A2)!0.5!(D2)$) {Detection results};

\coordinate (TopMid) at ($(A1)!0.5!(D1)$);
\coordinate (BotMid) at ($(A2)!0.5!(D2)$);
\node[toolbox] (TB) at ($(TopMid)!0.5!(BotMid)$)
{\textbf{Toolbox (Section~\ref{sec:general-motif}):}\\[1pt] \emph{$C$-admissible} $\Rightarrow$ detection\\ Theorem~\ref{thm:admissible}};

\draw[flow] (A1.east) to[out=-25,  in=96, looseness=1.00] (TB.west);
\draw[flow] (TB.west) to[out=264, in=25, looseness=1.00] (A2.east);

\draw[flow] (A1.east) to[out=0, in=180, looseness=1.05] (D1.west);

\draw[flow] (D1.west) to[out=205, in=84,  looseness=1.00] (TB.east);
\draw[flow] (TB.east) to[out=276, in=155, looseness=1.00] (D2.west);

\end{tikzpicture}%
}
\caption{Logical flow from admissibility to detection.}
\label{fig:logic-flow}
\end{figure}

\subsection{Construction of a Specific Bounded Degree Family}\label{subsec:bd-degree-specific}
For any integers $N_\sfv,N_\sfe$ such that $N_\sfv= \ell(d-1)+4$ and $N_\sfe = \binom{d}{2}\ell+d+1$ for some $\ell,d\in \mathbb N$, let $\maM(N_\sfv,N_\sfe,d)$ denote a special subset of bounded degree motifs with $N_\sfv$ vertices, $N_\sfe$ edges, and maximum degree $d$. Specifically, each motif $\sfM\in \maM(N_\sfv,N_\sfe,d)$ consists of $d-1$ paths of length $\ell$ between  two \emph{central vertices}, with each \emph{central vertex} connecting to an \emph{extremity vertex} of degree 1.
Additionally, there are $\ell$ red edges connecting distinct pairs of vertices between any two paths. The motif family $\maM(N_\sfv,N_\sfe,d)$ consists of all such motifs.
As illustrated in Figure~\ref{fig:bd-degree-graph}, each motif in $\maM(N_\sfv,N_\sfe,d)$ consists of 
$d-1$ paths, each in blue, with $\ell$ vertices of degree $d$. 
Specifically, each path is defined as $P_i\triangleq \sth{{v_{i,1}},v_{i,2},\cdots, v_{i,\ell}}$ for any $1\le i\le d-1$.
There are $\ell$ edges in red connecting distinct pairs of vertices between any two paths. 
Additionally, the \emph{central vertices} are $v_{01}, v_{02}$ and the \emph{extremity vertices} are $v_{00},v_{03}$.

Indeed, since there are exactly two vertices of degree $2$ in $\sfM$, we may view $\sfM$ as a partially labeled graph with two distinguished vertices $v_{0,0}$ and $v_{0,3}$ (see, e.g., \citep[Section~3.2]{lovasz2012large}). When counting such bounded degree motifs in $G_1$ and $G_2$, one can start from these labeled vertices and extend along the prescribed paths.
Furthermore, this convention does not affect the statistic: if $\sfM$ admits an automorphism exchanging $v_{0,0}$ and $v_{0,3}$, the partially labeled count equals twice the unlabeled count; otherwise the two counts are identical.
This is equivalent to labeling the two \emph{central vertices} $v_{0,1}$ and $v_{0,2}$ while deleting the \emph{extremity vertices} $v_{0,0}$ and $v_{0,3}$. In the simple case $d=3$ and $\ell=1$, the motif then becomes a partially labeled square with one cross edge, illustrating the basic structure of this family.


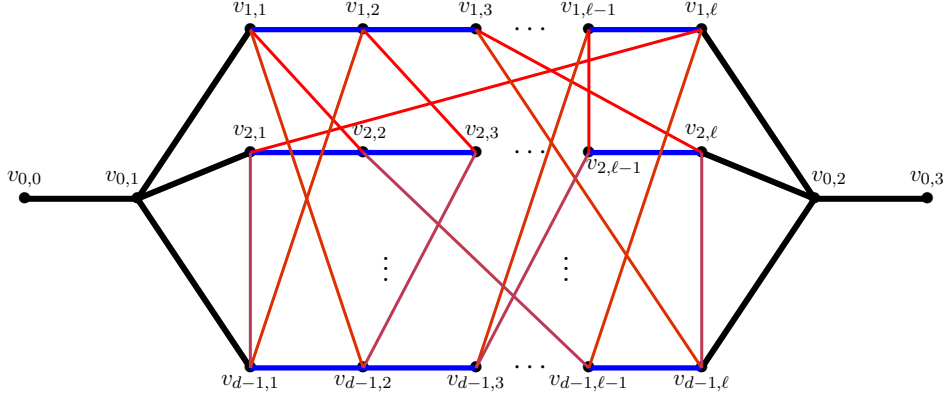
\begin{figure}
    \begin{center}
        \begin{tikzpicture}[yscale = 1.3, xscale=1.3, line width=0.4mm, line cap=round, line join=round]

            \tikzset{
              bullet/.style={circle, fill=black, inner sep=1.25pt},
              ep/.style={circle, draw=black, fill=white, inner sep=1.1pt}
            }

            \node[ep] (v00) at (-1,0) {};
            \node[ep] (v03) at (7,0) {};

            \node[above,scale=0.8,yshift=1] at (v00) {$v_{0,0}$};
            \node[above,scale=0.8,xshift=-7,yshift=1] at (0,0) {$v_{0,1}$};
            \node[above,scale=0.8,yshift=1] at (1,1.5) {$v_{1,1}$};
            \node[above,scale=0.8,yshift=1] at (1,0.41) {$v_{2,1}$};
            \node[below,scale=0.8,yshift=1] at (1,-1.5) {$v_{d-1,1}$};
            \node[above,scale=0.8,yshift=1] at (2,1.5) {$v_{1,2}$};
            \node[above,scale=0.8,xshift=7,yshift=1] at (2,0.41) {$v_{2,2}$};
            \node[below,scale=0.8,yshift=1] at (2,-1.5) {$v_{d-1,2}$};
            \node[above,scale=0.8,yshift=1] at (3,1.5) {$v_{1,3}$};
            \node[above,scale=0.8,xshift=3,yshift=1] at (3,0.41) {$v_{2,3}$};
            \node[below,scale=0.8,yshift=1] at (3,-1.5) {$v_{d-1,3}$};
            \node[above,scale=0.8,yshift=1] at (4,1.5) {$v_{1,\ell-1}$};
            \node[below,scale=0.8,xshift=12,yshift=1] at (4,0.41) {$v_{2,\ell-1}$};
            \node[below,scale=0.8,yshift=1] at (4,-1.5) {$v_{d-1,\ell-1}$};
            \node[above,scale=0.8,yshift=1] at (5,1.5) {$v_{1,\ell}$};
            \node[above,scale=0.8,yshift=1] at (5,0.41) {$v_{2,\ell}$};
            \node[below,scale=0.8,yshift=1] at (5,-1.5) {$v_{d-1,\ell}$};
            \node[above,scale=0.8,yshift=1,xshift=7] at (6,0) {$v_{0,2}$};
            \node[above,scale=0.8,yshift=1] at (v03) {$v_{0,3}$};

            \foreach \P in {
              (0,0),
              (1,1.5),(1,0.41),(1,-1.5),
              (2,1.5),(2,0.41),(2,-1.5),
              (3,1.5),(3,0.41),(3,-1.5),
              (4,1.5),(4,0.41),(4,-1.5),
              (5,1.5),(5,0.41),(5,-1.5),
              (6,0)
            }{
              \node[bullet] at \P {};
            }

            \draw[line width=0.35mm] (v00) -- (0,0);
            \draw[line width=0.35mm] (0,0) -- (1,1.5);
            \draw[line width=0.35mm] (0,0) -- (1,0.41);
            \draw[line width=0.35mm] (0,0) -- (1,-1.5);
            \draw[line width=0.35mm] (6,0) -- (v03);
            \draw[line width=0.35mm] (6,0) -- (5,1.5);
            \draw[line width=0.35mm] (6,0) -- (5,0.41);
            \draw[line width=0.35mm] (6,0) -- (5,-1.5);

            \draw[blue, line width=0.35mm] (1,1.5) -- (2,1.5);
            \draw[blue, line width=0.35mm] (2,1.5) -- (3,1.5);
            \draw[blue, line width=0.35mm] (1,0.41) -- (2,0.41);
            \draw[blue, line width=0.35mm] (2,0.41) -- (3,0.41);
            \draw[blue, line width=0.35mm] (1,-1.5) -- (2,-1.5);
            \draw[blue, line width=0.35mm] (2,-1.5) -- (3,-1.5);
            \draw[blue, line width=0.35mm] (4,1.5) -- (5,1.5);
            \draw[blue, line width=0.35mm] (4,0.41) -- (5,0.41);
            \draw[blue, line width=0.35mm] (4,-1.5) -- (5,-1.5);

            \draw (3.5,1.5) node{$\cdots$};
            \draw (3.5,0.41) node{$\cdots$};
            \draw (2.2,-0.55) node{$\vdots$};
            \draw (3.8,-0.55) node{$\vdots$};
            \draw (3.5,-1.5) node{$\cdots$};

            \draw[red1,line width=0.35mm] (1,1.5) -- (2,0.41);
            \draw[red1,line width=0.35mm] (2,1.5) -- (3,0.41);
            \draw[red1,line width=0.35mm] (3,1.5) -- (5,0.41);
            \draw[red1,line width=0.35mm] (5,1.5) -- (1,0.41);

            \draw[red2,line width=0.35mm] (2,1.5) -- (1,-1.5);
            \draw[red2,line width=0.35mm] (1,1.5) -- (2,-1.5);
            \draw[red2,line width=0.35mm] (3,1.5) -- (5,-1.5);
            \draw[red2,line width=0.35mm] (4,1.5) -- (3,-1.5);

            \draw[red3,line width=0.35mm] (1,0.41) -- (1,-1.5);
            \draw[red3,line width=0.35mm] (2,0.41) -- (4,-1.5);
            \draw[red3,line width=0.35mm] (4,0.41) -- (3,-1.5);
            \draw[red3,line width=0.35mm] (5,0.41) -- (5,-1.5);
            \draw[red3,line width=0.35mm] (3,0.41) -- (2,-1.5);

            \draw[red1,line width=0.35mm] (4,1.5) -- (4,0.41);
            \draw[red2,line width=0.35mm] (5,1.5) -- (4,-1.5);

        \end{tikzpicture}
    \end{center}
    \caption{A special bounded degree motif with vertex set size $N_\sfv$, edge set size $N_\sfe$, and maximum degree $d$.}
    \label{fig:bd-degree-graph}
\end{figure}

By Theorem~\ref{thm:admissible}, in order to provide a theoretical guarantee for $\maT_{\maM(N_\sfv,N_\sfe,d)}$, it suffices to verify the four \emph{$N_\sfe$-admissible} Conditions in Definition~\ref{def:admissible}. 
The connectivity for $\sfM\in \maM(N_\sfv,N_\sfe,d)$ yields the Condition~\ref{cond:connect}.
Since the vertices and edges for any $\sfM\in \maM(N_\sfv,N_\sfe,d)$ are bounded by $N_\sfe= \binom{d}{2}\ell+d+1$, Condition~\ref{cond:num} holds. We then verify Conditions~\ref{cond:signal-strength} and~\ref{cond:subgraph}, respectively. For Condition~\ref{cond:signal-strength}, since $\en(\sfM) = \en(\sfM')$ for any $\sfM,\sfM'\in \maM(N_\sfv,N_\sfe,d)$, the \emph{signal score} is characterized by $\sum_{\sfM\in \maM(N_\sfv,N_\sfe,d)}\rho^{2\en(\sfM)} = \rho^{2N_\sfe} |\maM(N_\sfv,N_\sfe,d)|$. The following lemma provides an estimate of $|\maM(N_\sfv,N_\sfe,d)|$.

\begin{lemma}\label{lem:motif-family-lwbd}
    For the motif family $ \maM(N_\sfv,N_\sfe,d)$  with $d\ge 3$, we have\begin{align}\label{eq:motif-family-lwbd}
        \frac{1}{2}\pth{\frac{2(N_\sfe-d-1)}{ed^{\frac{d}{d-2}}(d-1)}}^{\frac{d-2}{d}(N_\sfe-d-1)}\le |\maM(N_\sfv,N_\sfe,d)|\le \pth{\frac{2(N_\sfe-d-1)}{d(d-1)}}^{\frac{d-2}{d}(N_\sfe-d-1)}.
    \end{align}
\end{lemma}
 
The proof of Lemma~\ref{lem:motif-family-lwbd} is deferred to Appendix~\ref{apd:proof-lem-motif-family-lwbd}.
It follows from Lemma~\ref{lem:motif-family-lwbd} that, if $N_\sfe\ge d+1+ C(d) \rho^{-\frac{2d}{d-2}}$ with some constant $C(d)$, then \begin{align*}
    \sum_{\sfM\in \maM(N_\sfv,N_\sfe,d)}\rho^{2\en(\sfM)}&\ge \rho^{2N_\sfe}|\maM(N_\sfv,N_\sfe,d)|\\&\ge \frac{1}{2}\pth{\frac{2\rho^{\frac{2d}{d-2}}(N_\sfe-d-1)}{ed^{\frac{d}{d-2}}(d-1)}}^{\frac{d-2}{d}(N_\sfe-d-1)}\rho^{2d+2}\\&\ge\frac{1}{2}\pth{\frac{2C(d)}{ed^{\frac{d}{d-2}}(d-1)}}^{\frac{d-2}{d}C(d)\rho^{-2d/(d-2)}}\ge 400.
\end{align*} Hence, the requirement on the \emph{signal score} is satisfied. As for Condition~\ref{cond:subgraph}, we have the following Lemma.
\begin{lemma}\label{lem:subgraph}
    For any motif $\sfM\in \maM(N_\sfv,N_\sfe,d)$ and subgraph $\emptyset\neq \sfM'\subseteq \sfM$, we have 
    \begin{align*}
        d\,\vn(\sfM')\ge 2\,\en(\sfM')+1.
    \end{align*}
\end{lemma}

\begin{proof}

We note that for any $v\in V(\sfM')$, the degree of $v$ is at most $d$. Therefore, we have $d \vn(\sfM')\ge 2 \en(\sfM')$. It remains to prove that the equality cannot be achieved.

    Suppose $d\vn(\sfM') = 2\en(\sfM')$. Then, for any $v\in V(\sfM')$, the degree of $v$ in $\sfM'$ is $d$. If there exists $1\le i\le d-1,1\le j\le \ell$ such that $v_{i,j}\in \sfM'$, since the degree of $v_{i,j}$ in $\sfM'$ is $d$, then $v_{i,{j+1}} \in \sfM'$. Similarly, we obtain that $v_{i,j},v_{i,j+1},\cdots, v_{i,\ell},v_{0,2},v_{0,3}\in V(\sfM')$. Since the degree of $v_{0,3}$ in $\sfM'$ is at most 1 and $d\ge 3$, we have $d\vn(\sfM')>2\en(\sfM')$. If there exists $0\le i\le 3$ such that $v_{0,i}\in V(\sfM')$, we can similarly obtain that $v_{0,0}\in V(\sfM')$ or $v_{0,3}\in V(\sfM')$, and thus $d\vn(\sfM')>2\en(\sfM')$. Therefore, we conclude that $d\vn(\sfM')\ge 2\en(\sfM')+1$ for any $\sfM\in \maM(N_\sfv,N_\sfe,d)$ and $\sfM'\subseteq \sfM$.
\end{proof}

Indeed, the presence of two \emph{extremity vertices} ensures that all motifs in \(\mathcal{M}(N_\sfv, N_\sfe, d)\) are non-regular.
For any connected motif $\mathsf{M}$ with maximum degree d, we have $2\en(\mathsf{M}')  \le d\vn(\mathsf{M}')$ for any subgraph $\mathsf{M}'\subseteq \mathsf{M}$.
The equality cannot hold, as the existence of two \emph{extremity vertices} of degree 1 prevents any subgraph from being $d$-regular.
Hence the inequality is strict, and hence the lemma follows.
By Lemma~\ref{lem:subgraph}, we have $n^{\vn(\sfM')}p^{\en(\sfM')}\ge n^{\frac{2\en(\sfM')+1}{d}}p^{\en(\sfM')}\ge n^{1/d}$ if $p\ge n^{-2/d}$. Taking $\epsilon_0 = \frac{1}{d}$ yields Condition~\ref{cond:subgraph}. Since $\sfM$ is connected and $N_\sfv\le N_\sfe$ by our construction, we obtain that the motif family $\maM(N_\sfv,N_\sfe,d)$ is \emph{$N_\sfe$-admissible}. Consequently, we obtain the following theorem regarding the detection performance of $\maT_{\maM(N_\sfv,N_\sfe,d)}$.
Let $[x]$ denote the greatest integer less than or equal to $x$.

\begin{theorem}\label{thm:bd-degree-motif-main}
     If $p=n^{-a}$ with $0<a\le \frac{2}{3}$ and  $d=\qth{\frac{2}{a}}$, when $N_\sfe=o\pth{\frac{\log n}{\max\pth{\log\log n,-\log\rho}}}$ and $N_\sfe\ge \frac{C_1(d)}{\rho^{2d/(d-2)}}$ for some constant $C_1(d)$ depending on $d$,\begin{align*}
        \maP_0\pth{\maT_{\maM(N_\sfv,N_\sfe,d)}\ge \tau}+\maP_1\pth{\maT_{\maM(N_\sfv,N_\sfe,d)}<\tau}\le 0.05.
    \end{align*}

    If $p=n^{-o(1)}$, then for any constant $\epsilon>0$, when $N_\sfe = o\pth{\frac{\log n}{\max\pth{\log\log n,-\log\rho}}}$ and $N_\sfe\ge \frac{C_2(\epsilon)}{\rho^{2+\epsilon}}$ for some constant $C_2(\epsilon)$ depending on $\epsilon$,\begin{align*}
        \maP_0\pth{\maT_{\maM(N_\sfv,N_\sfe,d)}\ge \tau}+\maP_1\pth{\maT_{\maM(N_\sfv,N_\sfe,d)}<\tau}\le 0.05.
    \end{align*}

    Furthermore, $\maT_{\maM(N_\sfv,N_\sfe,d)}$ is  computable in $O(n^{N_\sfe})$.
\end{theorem}

\begin{proof}
We first show that $\maT_{\maM(N_\sfv,N_\sfe,d)}$ is computable in time $O(n^{N_\sfe})$. For any $\sfM\in \maM$, since there are $\binom{n}{\vn(\sfM)}(\vn(\sfM)!)$ injections from $V(\sfM)$ to $V(\bar{G}_1)$,
the injective homomorphism number $\inj(\sfM,\bar{G}_1)$ takes $\binom{n}{\vn(\sfM)}(\vn(\sfM)!)\le n^{\vn(\sfM)}$ time for computation. Similarly, the injective homomorphism number $\inj(\sfM,\bar{G}_2)$ can be computed in time $n^{\vn(\sfM)}$. Consequently, the time complexity for ${\maT_{\maM(N_\sfv,N_\sfe,d)}}$ is bounded by $2n^{\vn(\sfM)} |{\maM(N_\sfv,N_\sfe,d)}|$, and we have \begin{align*}
    2n^{\vn(\sfM)} |{\maM(N_\sfv,N_\sfe,d)}|&\overset{\mathrm{(a)}}{\le} 2n^{\ell(d-1)+4} \pth{\frac{2(N_\sfe-d-1)}{d(d-1)}}^{\frac{d-2}{d}\cdot (N_\sfe-d-1)}\\&\overset{\mathrm{(b)}}{\le} 2n^{2(N_\sfe+d-1)/d}(N_\sfe)^{N_\sfe}\overset{\mathrm{(c)}}{\le} n^{N_\sfe},
\end{align*}
where $\mathrm{(a)}$ is due to $\vn(\sfM)=\ell(d-1)+4$ and~\eqref{eq:motif-family-lwbd} in Lemma~\ref{lem:motif-family-lwbd}; $\mathrm{(b)}$ is because $\ell(d-1)+4 = \frac{2(N_\sfe+d-1)}{d}$, $\frac{2(N_\sfe-d-1)}{d(d-1)}\le N_\sfe$, and $\frac{d-2}{d}\cdot (N_\sfe-d-1)\le N_\sfe$; $\mathrm{(c)}$ is because $N_\sfe = o(\frac{\log n}{\log \log n})$ implies $2(N_\sfe)^{N_\sfe} = n^{o(1)}$ and $\frac{2(N_\sfe+d-1)}{d}\le \frac{2N_\sfe}{3}+2<N_\sfe$.

We then show the theoretical guarantee on $\maP_0\pth{\maT_{\maM(N_\sfv,N_\sfe,d)}\ge \tau}+\maP_1\pth{\maT_{\maM(N_\sfv,N_\sfe,d)}< \tau}$. By Theorem~\ref{thm:admissible}, it suffices to show that $\maM(N_\sfv,N_\sfe,d)$ is \emph{$C$-admissible}. Since $N_\sfv\le N_\sfe$ and $N_\sfe = o\pth{\frac{\log n}{\max \pth{\log \log n, -\log \rho}}}$, taking $C = N_\sfe$ yields Condition~\ref{cond:num}. By Lemma~\ref{lem:subgraph}, for all $\sfM\in \maM(N_\sfv,N_\sfe,d)$ and $\sfM'\subseteq \sfM$, we have \begin{align*}
    n^{\vn(\sfM')}p^{\en(\sfM')}\ge n^{(2\en(\sfM')+1)/d}p^{\en(\sfM')}\ge n^{1/d},
\end{align*}
where the last inequality is because $np^{d/2}\ge 1$ by the choice of $d$. Pick $\epsilon_0 = 1/d$, we conclude the Condition~\ref{cond:subgraph} in Definition~\ref{def:admissible}. Since all $\sfM\in \maM(N_\sfv,N_\sfe,d)$ are connected, it remains to verify the Condition~\ref{cond:signal-strength} $\sum_{\sfM\in \maM}\rho^{2\en(\sfM)}\ge 400$.

We first focus on the case $p = n^{-a}$ with $0<a\le \frac{2}{3}$. By Lemma~\ref{lem:motif-family-lwbd}, \begin{align*}
    \sum_{\sfM\in \maM(N_\sfv,N_\sfe,d)}\rho^{2\en(\sfM)}\ge \frac{1}{2}\pth{\frac{2(N_\sfe-d-1)}{ed^{\frac{d}{d-2}}(d-1)}}^{\frac{d-2}{d}(N_\sfe-d-1)}\rho^{2d+2}.
\end{align*}
We first prove that there exists a constant $C(d)$ depending on $d$ such that \begin{align*}
    \frac{1}{2}(C(d))^{\frac{d-2}{d}\rho^{-2d/(d-2)}}\rho^{2d+2}\ge 400
\end{align*}
for any $0\le \rho\le 1$ and integer $d\ge 3$.
Pick $C(d) = \exp\pth{3\log 800+\frac{d(2d+2)}{d-2}}$. Then, \begin{align*}
    &~\frac{d-2}{d}\rho^{-\frac{2d}{d-2}}\log (C(d))+(2d+2)\log\rho \\=&~ \frac{d-2}{d}\rho^{-\frac{2d}{d-2}}\pth{3\log 800+\frac{d(2d+2)}{d-2}}+(2d+2)\log\rho\\\overset{\mathrm{(a)}}{\ge}&~ \frac{3(d-2)\log(800)}{d}+(2d+2)\pth{\frac{1}{\rho^2}+\log \rho}\overset{\mathrm{(b)}}{\ge}   \log 800,
\end{align*}
where $\mathrm{(a)}$ is because $\rho\le 1$ and $\frac{2d}{d-2}\ge 2$; $\mathrm{(b)}$ follows from $\log x+\frac{1}{x^2}\ge 0$ for any $x>0$ and $\frac{3(d-2)}{d}\ge 1$ for any $d\ge 3$. Therefore, we obtain $\frac{1}{2}(C(d))^{\frac{d-2}{d}\rho^{-2d/(d-2)}} \rho^{2d+2}\ge 400$. Let $C_1(d)\triangleq d+1+\frac{1}{2}e d^{d/(d-2)} (d-1) C(d)$. When $N_\sfe\ge \frac{C_1(d)}{\rho^{2d/(d-2)}}$, we have \begin{align*}
    \frac{2(N_\sfe-d-1)\rho^{2d/(d-2)}}{ed^{d/(d-2)}(d-1)} &\ge \frac{2C_1(d) - 2(d+1) \rho^{2d/(d-2)}}{ed^{d/(d-2)}(d-1)} \ge \frac{2C_1(d) - 2(d+1)}{ed^{d/(d-2)}(d-1)} = C(d)
\end{align*}
and \begin{align*}
    \frac{d-2}{d}\cdot (N_\sfe-d-1) \ge \frac{d-2}{d}\cdot\pth{\frac{d+1}{\rho^{2d/(d-2)}}-d-1} \ge \frac{d-2}{d} \rho^{-\frac{2d}{d-2}}.
\end{align*}
Therefore, \begin{align*}
    \sum_{\sfM\in \maM(N_\sfv,N_\sfe,d)}\rho^{2\en(\sfM)} &\ge \frac{1}{2} \pth{\frac{2(N_\sfe-d-1)\rho^{2d/(d-2)}}{ed^{d/(d-2)}(d-1)}}^{\frac{d-2}{d}\cdot (N_\sfe-d-1)}\rho^{2d+2}\\&\ge \frac{1}{2}(C(d))^{\frac{d-2}{d}\rho^{-2d/(d-2)}} \rho^{2d+2}\ge 400.
\end{align*} 

We then focus on the case $p = n^{o(1)}$. For any $\epsilon>0$, pick $d = \qth{\frac{4}{\epsilon}}+3$. Then $\frac{2d}{d-2}<2+\epsilon$. Let $C_2(\epsilon)\triangleq C_1\pth{\qth{\frac{4}{\epsilon}}+3}$. We have shown that when $N_\sfe\ge \frac{C_1(d)}{\rho^{2d/(d-2)}}$, we have $\sum_{\sfM\in \maM(N_\sfv,N_\sfe,d)}\rho^{2\en(\sfM)}\ge 400$.
Since $\frac{C_1(d)}{\rho^{2d/(d-2)}} \le \frac{C_2(\epsilon)}{\rho^{2+\epsilon}}$, we also have $\sum_{\sfM\in \maM(N_\sfv,N_\sfe,d)}\rho^{2\en(\sfM)}\ge 400$.
Consequently, $\maM(N_\sfv,N_\sfe,d)$ is \emph{$N_\sfe$-admissible}. By Theorem~\ref{thm:admissible}, we obtain that \begin{equation*}
\maP_0\pth{\maT_{\maM(N_\sfv,N_\sfe,d)}\ge \tau}+\maP_1\pth{\maT_{\maM(N_\sfv,N_\sfe,d)}< \tau}\le 0.05.\qedhere
\end{equation*}
\end{proof}

By Lemmas~\ref{lem:motif-family-lwbd} and~\ref{lem:subgraph}, the bounded degree motif family $\maM(N_\sfv,N_\sfe,d)$ is \emph{$N_\sfe$-admissible}.
In view of Theorem~\ref{thm:bd-degree-motif-main}, if $\rho$ is a constant, one can choose a constant $N_\sfe$ when $p \ge n^{-2/3}$, making the test statistic computable in polynomial time. Although the test statistic $\maT_{\maM(N_\sfv,N_\sfe,d)}$ suffices for correlation detection, the motifs in $\maM(N_\sfv,N_\sfe,d)$ are highly specialized and uncommon in practical scenarios; for instance, triangles and quadrilaterals are not included in this family. To obtain a more broadly applicable motif-counting statistic, we will show in Section~\ref{subsec:general-bd-degree} that counting all bounded degree motifs also suffices for correlation detection. 

\subsection{A General Admissible Statistic}\label{subsec:general-bd-degree}
In this subsection, we consider an implementable test statistic $\maT_{\maM(N_\sfe,d)}$, where $\maM(N_\sfe,d)$ denotes the set of all connected motifs with $N_\sfe$ edges and maximum degree at most $d$. In order to provide theoretical guarantee, we show that $\maT_{\maM(N_\sfe,d)}$ is  \emph{$(N_\sfe{+}1)$-admissible}. We then verify the conditions in Definition~\ref{def:admissible}:\begin{enumerate}
    \item Since all $\sfM\in \maM(N_\sfe,d)$ are connected, we have  $\vn(\sfM)\le \en(\sfM)+1\le N_\sfe+1$, where the equality holds when $\sfM$ is a tree;
    \item Since $\maM(N_\sfv,N_\sfe,d)\subseteq \maM(N_\sfe,d)$, we have $$\sum_{\sfM\in \maM(N_\sfe,d)} \rho^{2\en(\sfM)}\ge \sum_{\sfM\in \maM(N_\sfv,N_\sfe,d)} \rho^{2\en(\sfM)}\ge 400;$$
    \item For all $\sfM\in\maM(N_\sfe,d)$ and subgraph $\sfM'\subseteq \sfM$, since the maximum degree of $\sfM'$ is bounded by $d$, we have $\en(\sfM')\le \frac{d}{2}\vn(\sfM')$. Consequently, when $p\ge n^{-a}$ for some constant $a<\frac{2}{3}$, one can pick $d\ge 3$  such that $1-\frac{da}{2}>0$, yielding $n^{\vn(\sfM')} p^{\en(\sfM')}\ge (np^{d/2})^{\vn(\sfM')}\ge n^{1-\frac{da}{2}}$.
    \item For all $\sfM\in \maM(N_\sfe,d)$, $\sfM$ is connected.
\end{enumerate}

Therefore, the bounded degree motif counting statistic $\maT_{\maM(N_\sfe,d)}$ is \emph{$(N_\sfe{+}1)$-admissible} whenever $\maT_{\maM(N_\sfv,N_\sfe,d)}$ is \emph{$N_\sfe$-admissible} and $p\ge n^{-a}$ for some constant $a<\frac{2}{3}$. Specifically, we have the following Theorem.
\begin{theorem}\label{thm:bd-degree-motif-general}
      If $p=n^{-a}$ with constant $0<a< \frac{2}{3}$, then for $d=3\cdot \indc{\frac{2}{5}<a<\frac{2}{3}}+\pth{\qth{\frac{2}{a}}-1}\cdot \indc{0<a\le \frac{2}{5}}$, when $N_\sfe=o\pth{\frac{\log n}{\max\pth{\log\log n,-\log\rho}}}$ and $N_\sfe\ge \frac{C_1(d)}{\rho^{2d/(d-2)}}$ for some constant $C_1(d)$ depending on $d$,\begin{align*}
        \maP_0\pth{\maT_{\maM(N_\sfe,d)}\ge \tau}+\maP_1\pth{\maT_{\maM(N_\sfe,d)}<\tau}\le 0.05.
    \end{align*}

    If $p=n^{-o(1)}$, then for any constant $\epsilon>0$, when $N_\sfe = o\pth{\frac{\log n}{\max\pth{\log\log n,-\log\rho}}}$ and $N_\sfe\ge \frac{C_2(\epsilon)}{\rho^{2+\epsilon}}$ for some constant $C_2(\epsilon)$ depending on $\epsilon$,\begin{align*}
        \maP_0\pth{\maT_{\maM(N_\sfe,d)}\ge \tau}+\maP_1\pth{\maT_{\maM(N_\sfe,d)}<\tau}\le 0.05.
    \end{align*}

    Furthermore, $\maT_{\maM(N_\sfe,d)}$ is  computable in time $O(n^{N_\sfe+1+o(1)})$.
\end{theorem}
\begin{proof}

We first show that the statistic $\maT_{\maM(N_\sfv,N_\sfe,d)}$ is computable in time $O(n^{N_\sfe+1+o(1)})$. Recall that the injective homomorphism number $\inj(\sfM,\bar{G}_1),\inj(\sfM,\bar{G}_2)$ can be computed in time $2n^{\vn(\sfM)}$. Since for connected  motif $\sfM$ with $N_\sfe$ edges, the number of vertices is bounded by $N_\sfe+1$, the total computation time is then bounded by $2n^{N_\sfe+1}|\maM(N_\sfe,d)|$. We then upper bound $|\maM(N_\sfe,d)|$. We note that \begin{align*}
    |\maM(N_\sfe,d)|
\le \sum_{N_\sfv=1}^{N_\sfe+1}\binom{\binom{N_\sfv}{2}}{N_\sfe}
\le (N_\sfe+1)\binom{\binom{N_\sfe+1}{2}}{N_\sfe}
\le (N_\sfe+1)\pth{\frac{e(N_\sfe+1)}{2}}^{N_\sfe},
\end{align*}
where the first inequality is because $1\le N_\sfv\le N_\sfe+1$ and there are at most $\binom{N_\sfv}{2}$ edges with $N_\sfv$ vertices; the last inequality is because $\binom{m_1}{m_2}\le\pth{\frac{em_1}{m_2}}^{m_2}$ for all $m_1,m_2\in \mathbb{N}$. Since $N_e = o\pth{\frac{\log n}{\log \log n}}$, we have $(N_\sfe+1)\pth{\frac{e(N_\sfe+1)}{2}}^{N_\sfe} = n^{o(1)}$. Consequently, the overall computation time is $n^{N_\sfe+1+o(1)}$.

We then show the theoretical guarantee. We prove that $\maM(N_\sfe,d)$ is \emph{$(N_\sfe{+}1)$-admissible}.
For Condition~\ref{cond:num}, since there are at most $N_\sfe+1$ vertices in a connected motif with $N_\sfe$ edges, choosing $C = N_\sfe +1$ satisfies the Condition~\ref{cond:num}. Since $\maM(N_\sfv,N_\sfe,d)\subseteq \maM(N_\sfe,d)$ and $\maM(N_\sfv,N_\sfe,d)$ is \emph{$N_\sfe$-admissible}, we have \begin{align*}
    \sum_{\sfM\in \maM(N_\sfe,d)}\rho^{2\en(\sfM)}\ge  \sum_{\sfM\in \maM(N_\sfv, N_\sfe,d)}\rho^{2\en(\sfM)}\ge 400.
\end{align*}
For Condition~\ref{cond:subgraph}, since the maximum degree of all $\sfM\in\maM$ is bounded by $d$, we have $\en(\sfM')\le \frac{d}{2}\vn(\sfM')$ for all $\sfM'\subseteq \sfM$. When $p=n^{o(1)}$, we have $n^{\vn(\sfM')}p^{\en(\sfM')}\ge n^{1/2}$.
When $p = n^{-a}$ with constant $0<a\le \frac{2}{5}$, we pick $d = \qth{\frac{2}{a}}-1$. Since $d\le \frac{2}{a}-1$,  we have \begin{align*}
    n^{\vn(\sfM')}p^{\en(\sfM')}&\ge n^{\vn(\sfM')}p^{d\vn(\sfM')/2}\ge (np^{(2/a-1)/2})^{\vn(\sfM')}\\&=(n^{a/2})^{\vn(\sfM')}\ge n^{a/2}.
\end{align*}
When $p=n^{-a}$ with constant $\frac{2}{5}<a<\frac{2}{3}$, we pick $d=3$.
Consequently, 
\begin{align*}
    n^{\vn(\sfM')}p^{\en(\sfM')}&\ge n^{\vn(\sfM')}p^{d\vn(\sfM')/2}\\&=(n^{1-3a/2})^{\vn(\sfM')}\ge n^{1-3a/2}.
\end{align*}
Picking $\epsilon_0 = \min\pth{\frac{a}{2},1-\frac{3a}{2}}$ yields the Condition~\ref{cond:subgraph} in Definition~\ref{def:admissible}. Since all $\sfM\in \maM(N_\sfe,d)$ are connected, we conclude that $\maM(N_\sfe,d)$ is \emph{$C$-admissible}. By Theorem~\ref{thm:admissible}, we have \begin{equation*}
    \maP_0(\maT_{\maM(N_\sfe,d)}\ge\tau)+\maP_1(\maT_{\maM(N_\sfe,d)}<\tau)\le 0.05.\qedhere
\end{equation*}
\end{proof}
In view of Theorem~\ref{thm:bd-degree-motif-general}, we have shown that the bounded degree motif counting statistic $\maT_{\maM(N_\sfe,d)}$ is sufficient for detection when $p \ge n^{-a}$ for some constant $0 < a < \frac{2}{3}$. 
The parameter $d$ is selected to ensure the Condition~\ref{cond:subgraph} in Definition~\ref{def:admissible}.
Moreover, for constant correlation $\rho$, one can choose a constant $N_\sfe$ satisfying the required conditions, making the test statistic computable in polynomial time $O(n^{N_\sfe+1+o(1)})$. 
Compared with the statistic $\maT_{\maM(N_\sfv,N_\sfe,d)}$ in Theorem~\ref{thm:bd-degree-motif-main}, the theoretical guarantee for $\maT_{\maM(N_\sfe,d)}$ requires a slightly stronger condition on the edge probability, namely $p \ge n^{-a}$ with $0<a<\frac{2}{3}$. Nevertheless, $\maT_{\maM(N_\sfe,d)}$ is better aligned with applications, as the bounded degree motif family $\maM(N_\sfe,d)$ includes many commonly occurring motifs. 
Although the condition $p \ge n^{-2/3}$ plays an important role in the analysis of the motif-counting statistics $\maT_{\maM(N_\sfv, N_\sfe,d)}$ and $\maT_{\maM(N_\sfe,d)}$, it is not a necessary requirement for admissible motif families. In fact, for sparser graphs with $p < n^{-2/3}$, \cite{mao2024testing} showed that tree counting remains successful for detection under an additional constraint $\rho^2 \ge \alpha \approx 0.338$; moreover, the corresponding tree-counting statistic is also admissible. More generally, in such sparser regimes, detection remains feasible as long as the admissibility is satisfied.

Indeed, one could consider a larger motif family than $\maM(N_\sfe,d)$, such as the family containing all bounded degree motifs with at most $N_\sfe$ edges rather than exactly $N_\sfe$. While this would increase the signal and the theoretical guarantee could be established similarly, such a choice may introduce practical challenges.
Recall that we select the weight $\omega_\sfM = \frac{\rho^{\en(\sfM)} (n-\vn(\sfM))!}{n!\,(p(1-p))^{\en(\sfM)}\aut(\sfM)}.$
Since the correlation parameter $\rho$ is challenging to estimate in practice, it is natural to restrict attention to motifs with the same number of edges. Notably, the family $\maM(N_\sfe,d)$ satisfies this property and remains practical. We will demonstrate in Section~\ref{sec:num-results} that the statistic $\maT_{\maM(N_\sfe,d)}$ performs well on synthetic data.





Theorems~\ref{thm:bd-degree-motif-main} and~\ref{thm:bd-degree-motif-general} provide sufficient conditions for detection, leading to a  degree-$O(\rho^{-2d/(d-2)})$ algorithm  when $p = n^{-a}$ for any $0<a\le \tfrac{2}{3}$. When $p=n^{-o(1)}$, there exists a degree-$O(\rho^{-2-\epsilon})$ algorithm for any $\epsilon>0$.
When the maximum degree satisfies $d \le 2$, the  connected motifs reduce to paths and cycles, which provide limited structural information for detection. Hence, we consider bounded degree motifs with $d \ge 3$; this choice underlies the coefficient $\tfrac{2}{3}$ appearing in the detection threshold.
Indeed, the size of the motif family $\maM$ plays a crucial role in controlling errors. Specifically, in order to achieve detection for any constant $\rho$, it is necessary that $|\maM|\asymp \en(\maM)^{\en(\maM)}$. As a result, it is natural to consider the motifs with maximum degree $d\ge 3$. 
The condition  $\en(\maM) = o(\frac{\log n}{\max\sth{\log \log n,-\log \rho}})$, $\en(\maM)\ge \frac{C_1\pth{d}}{\rho^{ {2d/(d-2)}}}$, and $\en(\maM)\ge \frac{C_2(\epsilon)}{\rho^{2+\epsilon}}$ implies a necessary constraint on the correlation coefficient, namely, $\rho \gtrsim \frac{1}{\log n}$. In fact, counting motifs with $\en(\maM)$ edges in $\bar{G}_1$ and $\bar{G}_2$ corresponds to a degree-$\en(\maM)$ polynomial algorithm.
In particular, when $\rho$ is a constant, the time complexity remains polynomial. 
For  $p\ge n^{-1/3}$ and any constant $\rho$, \cite{barak2019nearly}
achieved detection criteria by counting balanced graphs, whereas~\cite{mao2024testing} succeeded when $p\ge n^{-1+o(1)}$ and $\rho \ge \sqrt{\alpha}$ by counting trees.
Our bounded degree counting method bridges the gap between these regimes by establishing detection for $p\in [n^{-2/3},n^{-1/3}]$ with constant $\rho <\sqrt{\alpha}$.
Beyond motif counting approaches, \cite{ding2023polynomial} proposed an iterative method that can be applied to both detection and the recovery of $\pi$ under $\maH_1$, achieving reliable performance for any constant $\rho$ when $p \ge n^{-1+\delta}$ with a small constant $\delta$. We will show in Section~\ref{sec:refined-results} that detection remains achievable in this regime via a more refined construction.


Indeed, our results align with computational hardness conjectures in this problem.
It has been postulated in~\cite{hopkins2017efficient,hopkins2018statistical} that the framework of low-degree polynomial algorithms  captures the hardness of detecting and recovering latent structures. Based on the low-degree conjecture, \cite{ding2023low} showed that any degree-$O(\rho^{-1})$ polynomial algorithm fails for detection with vanishing error when $p = n^{-\alpha}$ for any constant $\alpha\in (0,1)$. The more recent work~\cite{li2025algorithmic} provided evidence on the detection problem and conjectured that any degree-$o(\rho^{-1})$ polynomial algorithm fails for detection with constant error. In summary, our bounded degree motif counting statistic provides a polynomial-time algorithm that aligns with the low-degree conjecture, achieving a gap of $\frac{2d}{d-2}$ in the sparse regime and a gap of $2+\epsilon$ in the dense regime in terms of the exponent of $1/\rho$.

\begin{remark}[Trade-off between computation and statistical efficiency]
As shown in Theorems~\ref{thm:bd-degree-motif-main} and~\ref{thm:bd-degree-motif-general}, there exist conditions on $N_\sfe$ to satisfy the detection criterion $0.05$. 
To achieve stronger statistical efficiency, it is necessary to count motifs with a larger number of edges. However, if the detection criterion satisfies $$\limsup_{n\to\infty}\qth{\maP_0(\maT\ge \tau)+\maP_1(\maT<\tau)} = o(1),$$ then $N_\sfe = \omega(1)$ as $n \to \infty$, and the corresponding test statistic is no longer computable in polynomial time. Consequently, there exists a trade-off between computation time and statistical efficiency.
\end{remark}


\section{Extension to the Sparser Regime via Refined Motif Classes}\label{sec:refined-results}

In Section~\ref{sec:bd-sub-count}, we show that counting a bounded degree motif family suffices for correlation detection for any constant $\rho=\Omega(1)$ when $p\ge n^{-2/3}$; see Theorem~\ref{thm:bd-degree-motif-main}. This exponent $2/3$ is not intrinsic. Indeed, it stems from the technical Condition~\ref{cond:subgraph} in Definition~\ref{def:admissible}, which requires $n^{\vn(\sfM')}p^{\en(\sfM')}\ge n^{\epsilon_0}$ for all $\sfM\in\maM$ and all nonempty subgraphs $\emptyset\neq \sfM'\subseteq \sfM$. When choosing $\maM(N_\sfv,N_\sfe,d)$ (or $\maM(N_\sfe,d)$), we can only control the density ratio $\en(\sfM')/\vn(\sfM')$ by a constant close to $3/2$. Consequently, enforcing Condition~\ref{cond:subgraph} yields the constraint $p\ge n^{-2/3}$. On the other hand, it follows from~\cite{ding2023low} that when $p=n^{-1+o(1)}$, there is evidence from the low-degree perspective that the correlation constraint $\rho^2\ge \alpha\approx 0.338$ is tight. It is then natural to ask:
\begin{center}
\emph{Can we design a finer motif family so that motif counting succeeds when $p\ge n^{-1+\delta}$?}
\end{center}

To extend the regime $p\ge n^{-2/3}$ to the sparser regime $p\ge n^{-1+\delta}$ for any given constant $\delta>0$, we need to construct a motif family $\maM$ such that, for every $\sfM\in\maM$ and every nonempty subgraph $\emptyset\neq \sfM'\subseteq \sfM$, the edge--vertex density satisfies $\en(\sfM')/\vn(\sfM')\le 1/(1-\delta)$, while keeping the family size at most $(c_1 N_\sfe)^{c_2 N_\sfe}$. Recall the construction in Figure~\ref{fig:bd-degree-graph} for $d=3$. There, most vertices have degree $3$, which leads to a relatively large edge--vertex density. In what follows, we construct a new motif family in which degree-$3$ vertices are rare. In fact, their proportion is controlled by $\delta$ and all remaining vertices have degree $1$ or $2$.

We then introduce the construction for $\overline{\maM}(N_\sfv,N_\sfe,3)$ such that $N_\sfv = t(6\ell -1)-2\ell+3$ and $N_\sfe= t(6\ell-1)+2$, where $t,\ell\ge 1$ are integers.
See Figure~\ref{fig:refined-motif}.
For simplicity, we denote the dashed edge by a path of $t+1$ vertices connecting two vertices. For instance, there are $t-1$ vertices between $v_{1,1}$ and $v_{1,2}$. For each $\overline{\maM}(N_\sfv,N_\sfe,3)$, the construction consists of the following three steps:\begin{enumerate}
    \item Construct $\ell$ cycles of length $4t$ (in red color), with $4\ell$ special vertex $v_{i,1},v_{i,2},v_{i,3},v_{i,4}$ for $1\le i\le \ell$ and remaining normal vertices. The special vertices will be of degree 3 and the remaining vertices will be of degree 2. There are $t-1$ vertices between $(v_{i,1},v_{i,2}),(v_{i,2},v_{i,3}),(v_{i,3},v_{i,4}),(v_{i,4},v_{i,1})$, respectively. 
    \item Connect $v_{i,3}$ and $v_{i+1,1}$ with a dashed edge (i.e., there are $t-1$ vertices between them) and connect $(v_{0,0},v_{1,1}),(v_{\ell,3},v_{0,1})$, where $v_{0,0},v_{0,1}$ are two extremity vertices.
    \item\label{item:refined-step3} Arbitrarily select a permutation $\sigma:[\ell]\mapsto [\ell]$. Connect $v_{i,2}$ to $v_{\sigma(i),4}$ with a dashed edge (in blue color) for all $1\le i\le \ell$.
\end{enumerate}

\begin{figure}[htbp]
  \begin{center}
    \begin{tikzpicture}[yscale=1.5, xscale=1.5, line cap=round, line join=round]

      \def\stubfrac{0.18}
      \newcommand{\stubdashed}[2]{%
        \draw[line width=0.4mm]
          (#1) -- ($(#1)!\stubfrac!(#2)$);
        \draw[densely dashed, line width=0.4mm]
          ($(#1)!\stubfrac!(#2)$) -- ($(#2)!\stubfrac!(#1)$);
        \draw[line width=0.4mm]
          ($(#2)!\stubfrac!(#1)$) -- (#2);
      }
      \newcommand{\redstubdashed}[2]{%
        \draw[red, line width=0.35mm]
          (#1) -- ($(#1)!\stubfrac!(#2)$);
        \draw[red, densely dashed, line width=0.35mm]
          ($(#1)!\stubfrac!(#2)$) -- ($(#2)!\stubfrac!(#1)$);
        \draw[red, line width=0.35mm]
          ($(#2)!\stubfrac!(#1)$) -- (#2);
      }
      \newcommand{\bluestubdashed}[2]{%
        \draw[blue, line width=0.35mm]
          (#1) -- ($(#1)!\stubfrac!(#2)$);
        \draw[blue, densely dashed, line width=0.35mm]
          ($(#1)!\stubfrac!(#2)$) -- ($(#2)!\stubfrac!(#1)$);
        \draw[blue, line width=0.35mm]
          ($(#2)!\stubfrac!(#1)$) -- (#2);
      }

      \coordinate (L) at (0,0);
      \coordinate (R) at (8.9,0); 

      \coordinate (D1l) at (1.0,0);
      \coordinate (D1t) at (1.6,0.95);
      \coordinate (D1r) at (2.2,0);
      \coordinate (D1b) at (1.6,-0.95);

      \coordinate (D2l) at (3.2,0);
      \coordinate (D2t) at (3.8,0.95);
      \coordinate (D2r) at (4.4,0);
      \coordinate (D2b) at (3.8,-0.95);

      \coordinate (D3l) at (6.7,0);
      \coordinate (D3t) at (7.3,0.95);
      \coordinate (D3r) at (7.9,0);
      \coordinate (D3b) at (7.3,-0.95);

      \node[above,scale=0.8,yshift=1.5pt] at (L) {$v_{0,0}$};
      \node[above,scale=0.8,yshift=1.5pt] at (R) {$v_{0,1}$};

      \node[right, scale=0.8, xshift=1pt] at (D1l) {$v_{1,1}$};
      \node[above,scale=0.8, yshift=1.5pt] at (D1t) {$v_{1,2}$};
      \node[left,scale=0.8, xshift=-1pt] at (D1r) {$v_{1,3}$};
      \node[below,scale=0.8, yshift=-1.5pt] at (D1b) {$v_{1,4}$};

      \node[right, scale=0.8, xshift=1pt] at (D2l) {$v_{2,1}$};
      \node[above,scale=0.8, yshift=1.5pt] at (D2t) {$v_{2,2}$};
      \node[left,scale=0.8, xshift=-1pt] at (D2r) {$v_{2,3}$};
      \node[below,scale=0.8, yshift=-1.5pt] at (D2b) {$v_{2,4}$};

      \node[right, scale=0.8, xshift=1pt] at (D3l) {$v_{\ell,1}$};
      \node[above,scale=0.8, yshift=1.5pt] at (D3t) {$v_{\ell,2}$};
      \node[left,scale=0.8, xshift=-1pt] at (D3r) {$v_{\ell,3}$};
      \node[below,scale=0.8, yshift=-1.5pt] at (D3b) {$v_{\ell,4}$};

      \draw[line width=0.4mm] (L) -- (D1l);
      \draw[line width=0.4mm] (D3r) -- (R);

      \redstubdashed{D1l}{D1t}
      \redstubdashed{D1t}{D1r}
      \redstubdashed{D1r}{D1b}
      \redstubdashed{D1b}{D1l}

      \stubdashed{D1r}{D2l}

      \redstubdashed{D2l}{D2t}
      \redstubdashed{D2t}{D2r}
      \redstubdashed{D2r}{D2b}
      \redstubdashed{D2b}{D2l}

      \node at ($(D2r)!0.5!(D3l)$) {$\cdots$};

      \draw[line width=0.4mm] (D2r) -- ++(0.35,0);
      \draw[line width=0.4mm] (D3l) -- ++(-0.35,0);

      \redstubdashed{D3l}{D3t}
      \redstubdashed{D3t}{D3r}
      \redstubdashed{D3r}{D3b}
      \redstubdashed{D3b}{D3l}

      \bluestubdashed{D1t}{D2b}
      \bluestubdashed{D1b}{D2t}
      \bluestubdashed{D3t}{D3b}

      \draw (L) node[circle,draw,fill=white,inner sep=1.1pt]{};
      \draw (R) node[circle,draw,fill=white,inner sep=1.1pt]{};

      \foreach \p in {D1l,D1t,D1r,D1b,D2l,D2t,D2r,D2b,D3l,D3t,D3r,D3b}{
        \fill (\p) circle (1.25pt);
      }

    \end{tikzpicture}
  \end{center}
  \caption{A special motif family $\overline{\maM}(N_\sfv,N_\sfe,d)\subseteq \maM(N_\sfe,d)$ with $d=3.$ Each dashed edge denotes a path including $t-1$ nodes.}
  \label{fig:refined-motif}
\end{figure}

By Step~\ref{item:refined-step3}, there are at most $\ell!$ choices of $\sfM\in \overline{\maM}(N_\sfv,N_\sfe,3)$. To lower bound $|\overline{\maM}(N_\sfv,N_\sfe,3)|$, it remains to estimate the number of automorphisms for $\sfM\in \overline{\maM}(N_\sfv,N_\sfe,3)$. We provide a lower bound using an argument similar to the proof of Lemma~\ref{lem:motif-family-lwbd}. 
Indeed, once $\{v_{i,1},v_{i,2},v_{i,3}\}_{1\le i\le \ell}$ and $v_{0,0},v_{0,1}$ are specified, the vertices $\{v_{i,4}\}_{1\le i\le \ell}$ are uniquely determined, since $v_{i,1}$ and $v_{i,3}$ have degree $3$, and the only remaining degree-$3$ vertex adjacent to both of them is $v_{i,4}$. Therefore, it suffices to count the possible images of $\{v_{i,1},v_{i,2},v_{i,3}\}_{1\le i\le \ell}$ together with $v_{0,0},v_{0,1}$. Note that there are exactly two degree-$1$ vertices, hence there are only two choices for the images of $v_{0,0}$ and $v_{0,1}$. Moreover, for each $v_{i,j}$ with $1\le i\le \ell$ and $1\le j\le 3$, the image of the successor vertex (i.e., $v_{i,j+1}$ for $j=1,2$ and $v_{i+1,1}$ for $j=3$) has at most two possible choices. Consequently, each $\sfM\in\overline{\maM}(N_\sfv,N_\sfe,3)$ has at most $2^{3\ell+1}$ automorphisms. Since $(\ell/e)^{\ell}\le \ell!\le \ell^\ell$, we conclude that
\begin{align}\label{eq:refined-construct}
    \frac{1}{2}\Big(\frac{\ell}{8e}\Big)^{\ell}\le \frac{\ell!}{2^{3\ell+1}}\le |\overline{\maM}(N_\sfv,N_\sfe,3)|\le \ell!\le \ell^\ell.
\end{align}
Pick $t=\lfloor \delta^{-1}\rfloor+1$, where $\lfloor x\rfloor$ denotes the greatest integer less than or equal to $x$. We first introduce the following lemma on $\en(\sfM')/\vn(\sfM')$.
\begin{lemma}\label{lem:refined-aut-lwbd}
    For any $\emptyset\neq \sfM'\subseteq \sfM$ and $\sfM\in \overline{\maM}(N_\sfv,N_\sfe,3)$,  $\en(\sfM')/\vn(\sfM')\le 1+\delta$.
\end{lemma}
\begin{proof}
    We separate it into two cases.

    \textbf{Case 1: $\sfM'$ is a tree.} We note that for any tree $\sfM'$, we have $\en(\sfM') = \vn(\sfM')-1$, and thus $\en(\sfM')/\vn(\sfM')<1<1+\delta$.

    \textbf{Case 2: $\sfM'$ contains at least one cycle.}
    If $\sfM'$ is disconnected, we may ignore all tree components since they satisfy $\en/\vn<1$ and can only decrease the ratio $\en(\sfM')/\vn(\sfM')$.
    Therefore, it suffices to consider the connected component of $\sfM'$ that contains a cycle, and we may assume without loss of generality that $\sfM'$ is connected and contains a cycle.  
    
    Let $\maV\triangleq \sth{v_{i,j}:1\le i\le \ell,1\le j\le 4}$.
    By construction, any two  vertices in $\maV$ are well separated: along any path connecting them, there are at least $t-1$ intermediate vertices (equivalently, their graph distance is at least $t$). Consequently, any cycle must contain at least 3 vertices in $\maV$, and the segments between consecutive such vertices contribute at least $t$ vertices each, so the shortest possible cycle has length at least $3t$, and thus $\vn(\sfM')\ge 3t$. Moreover, in any subgraph $\sfM'\subseteq \sfM$, each degree-$3$ vertex requires at least $t$ vertices (itself plus at least $t-1$ separating vertices), hence the number of degree-$3$ vertices in $\sfM'$ is at most $\lfloor \vn(\sfM')/t\rfloor+1$. Since $2\en(\sfM') = \sum_{v\in V(\sfM')}\operatorname{deg}(v)$, where $\operatorname{deg}(v)$ denotes the degree of vertex $v$, we then have \begin{align*}
        2\en(\sfM') = \sum_{v\in V(\sfM')}\operatorname{deg}(v)\le  \sum_{v\in V(\sfM')}(2+\indc{\operatorname{deg}(v)=3})\le 2\vn(\sfM')+\vn(\sfM')/t+1, 
    \end{align*}
    where the last inequality uses $\sum_{v\in V(\sfM')}\indc{\operatorname{deg}(v)=3}\le \lfloor \vn(\sfM')/t\rfloor+1\le  \vn(\sfM')/t+1$. Since $\vn(\sfM')\ge 3t$, we then obtain \begin{align*}
        \frac{\en(\sfM')}{\vn(\sfM')}\le 1+\frac{1}{2t}+\frac{1}{2\vn(\sfM')}\le 1+\frac{1}{2t}+\frac{1}{6t} = 1+\frac{2}{3t}\le 1+\delta,
    \end{align*}
    where the last inequality is because $t = \lfloor \delta^{-1}\rfloor+1\ge \delta^{-1}$.
\end{proof}
We are now ready to prove the main theorem in this section, which also implies Theorem~\ref{thm:main}.
\begin{theorem}\label{thm:refined-main-thm}
    Assume that $p\ge n^{-1+\delta}$ with some constant $\delta>0$. When $N_\sfe = o\pth{\frac{\log n}{\max(\log\log n,-\log \rho)}}$ and $N_\sfe \ge \frac{C(\delta)}{\rho^{24/\delta}}$ for a constant $C(\delta)$, then there exists a threshold $\tau$ such that 
    \begin{align*}
        \maP_0(\maT_{\overline{\maM}(N_\sfv,N_\sfe,3)}\ge \tau)+\maP_1(\maT_{\overline{\maM}(N_\sfv,N_\sfe,3)}< \tau)\le 0.05.
    \end{align*}
\end{theorem}

\begin{proof}
    By Theorem~\ref{thm:admissible}, it suffices to show that $\overline{\maM}(N_\sfv,N_\sfe,3)$ is \emph{$C-$admissible}. Conditions~\ref{cond:connect} and~\ref{cond:num} follow from the construction of $\overline{\maM}(N_\sfv,N_\sfe,3)$ and $N_\sfv\le N_\sfe =o\pth{\frac{\log n}{\max(\log\log n,-\log \rho)}}$.
    By Lemma~\ref{lem:refined-aut-lwbd}, for all $\sfM\in \maM$ and $\emptyset\neq \sfM'\subseteq \sfM$, we have \begin{align*}
        n^{\vn(\sfM')}p^{\en(\sfM')}\ge n^{\vn(\sfM')}p^{(1+\delta)\vn(\sfM')}\ge n^{\vn(\sfM')\delta^2}\ge n^{\delta^2}.
    \end{align*}
    Taking $\epsilon_0 = \delta^2$ yields Condition~\ref{cond:subgraph}. It remains to verify Condition~\ref{cond:signal-strength}. Recall that $N_\sfe = t(6\ell-1)+2$ and $t = \lfloor \delta^{-1}\rfloor+1$. Thus $\ell\ge \frac{N_\sfe}{6t}\ge \frac{\delta N_\sfe}{12}$. By~\eqref{eq:refined-construct},\begin{align*}
        \sum_{\sfM\in \overline{\maM}(N_\sfv,N_\sfe,3)} \rho^{2\en(\sfM)}\ge |\overline{\maM}(N_\sfv,N_\sfe,3)|\rho^{2N_\sfe}\ge \frac{1}{2}\pth{\frac{\delta N_\sfe}{96 e}}^{\delta N_\sfe /12}\rho^{2N_\sfe} = \frac{1}{2}\Big[\Big(\frac{\delta N_\sfe}{96e}\Big)^{\delta /12}\rho^2\Big]^{N_\sfe}.
    \end{align*}
    Consequently, when $N_\sfe\ge (800/\rho^2)^{12/\delta}\cdot 96e/\delta$, we have $\sum_{\sfM\in \overline{\maM}(N_\sfv,N_\sfe,3)} \rho^{2\en(\sfM)}\ge 400$, which yields Condition~\ref{cond:signal-strength}. We then prove the theorem.
\end{proof}

In view of Theorem~\ref{thm:refined-main-thm}, we obtain a polynomial-time test for any constant $\rho=\Omega(1)$ when $p\ge n^{-1+\delta}$. Moreover, while keeping $N_{\sfe}$ as a fixed constant, the numerical constant $0.05$ in the risk bound can be replaced by any prescribed constant. To drive the testing error to $o(1)$ as $n\to\infty$, however, one would need to take $N_{\sfe}=\omega(1)$; the resulting procedure is no longer polynomial-time, since the running time scales as $O(n^{N_{\sfe}})$. There is substantial evidence suggesting that this limitation is intrinsic to our construction. We briefly summarize two such pieces of evidence below.

The treewidth of a motif $\sfM$, denoted by $\operatorname{tw}(\sfM)$, is a standard measure of how tree-like $\sfM$ is (in particular, $\operatorname{tw}(\sfM)=1$ for forests). 
Many subgraph-counting algorithms (e.g., color-coding with dynamic programming on a tree-decomposition) have running time polynomial in $n$ only when $\operatorname{tw}(\sfM)=O(1)$. In our setting, $N_{\sfe}=\omega(1)$ implies $\ell=\omega(1)$, and the motifs in Figures~\ref{fig:bd-degree-graph} and~\ref{fig:refined-motif} satisfy $\operatorname{tw}(\sfM)\to\infty$. This is consistent with general hardness evidence for parameterized subgraph counting when the pattern treewidth is unbounded (see, e.g.,~\cite{meeks2016challenges}).

On the other hand, ~\cite{flum2004parameterized} proved that counting $k$-cycles parameterized by $k$ is $\#\mathrm{W}[1]$-hard, providing evidence against substantially faster exact counting algorithms of the form $f(k)\cdot n^{O(1)}$. In our setting, $N_{\sfe}=\omega(1)$ implies $\ell=\omega(1)$. Moreover, for a uniformly random permutation $\sigma\in S_\ell$, the longest cycle has length $\Theta(\ell)$ with high probability (see, e.g.,~\cite{shepp1966ordered}), so with probability $1-o(1)$ the orbit decomposition of $\sigma$ contains a cycle of length $\omega(1)$. Consequently, our construction in Figure~\ref{fig:refined-motif} typically contains long cyclic structures when $\ell=\omega(1)$, which further aligns with the above hardness evidence for accelerating exact counting in this regime.

Finally, it is natural to ask why we focus on \emph{$C$-admissible} bounded degree motifs rather than broader bounded-treewidth families, since color-coding applies to bounded-treewidth patterns. Our motivation is that, beyond polynomial-time computability, our analysis crucially exploits the availability of a large motif family (of size roughly $(c_1N_{\sfe})^{c_2N_{\sfe}}$) to amplify signal via motif multiplicity (see Condition~\ref{cond:signal-strength}). In contrast, the number of non-isomorphic patterns in bounded-treewidth classes may grow only exponentially in $N_{\sfe}$; one instance is the family of trees, whose cardinality is governed by Otter's constant~\cite{otter1948number} and hence is of order $c^{N_{\sfe}}$. 
This family-size bottleneck therefore limits how far bounded-treewidth motif classes can amplify the signal in the regime $p\ge n^{-1+\delta}$ with constant $\rho=\Omega(1)$.


\section{Numerical Results}\label{sec:num-results} 

In this section, we present numerical results on synthetic data to verify our theoretical results.
Specifically, we generate 100 pairs of graphs that are independent $\maG(n,p)$, and another 100 pairs of graphs from the correlated \ER model $\maG(n,p,\rho)$. We then evaluate the performance of our test statistic  $\maT_{\maM(N_\sfe,d)}$ on the synthetic data.

Fixing $n=100$, $p=0.05$, and $\rho=0.99$, we evaluate the statistic
$\maT_{\maM(N_\sfe,d)}$ with $N_\sfe=d=4$ on 100 pairs of graphs under each model.
Figure~\ref{fig:histogram} displays the empirical distributions: the histogram (left)
and the boxplots (right) reveal a clear shift under the correlated model relative to the
independent model, indicating separated behavior under $\maH_0$ and $\maH_1$.

\begin{figure}[ht]
\centering

\subfloat{%
  \includegraphics[width=0.4\linewidth]{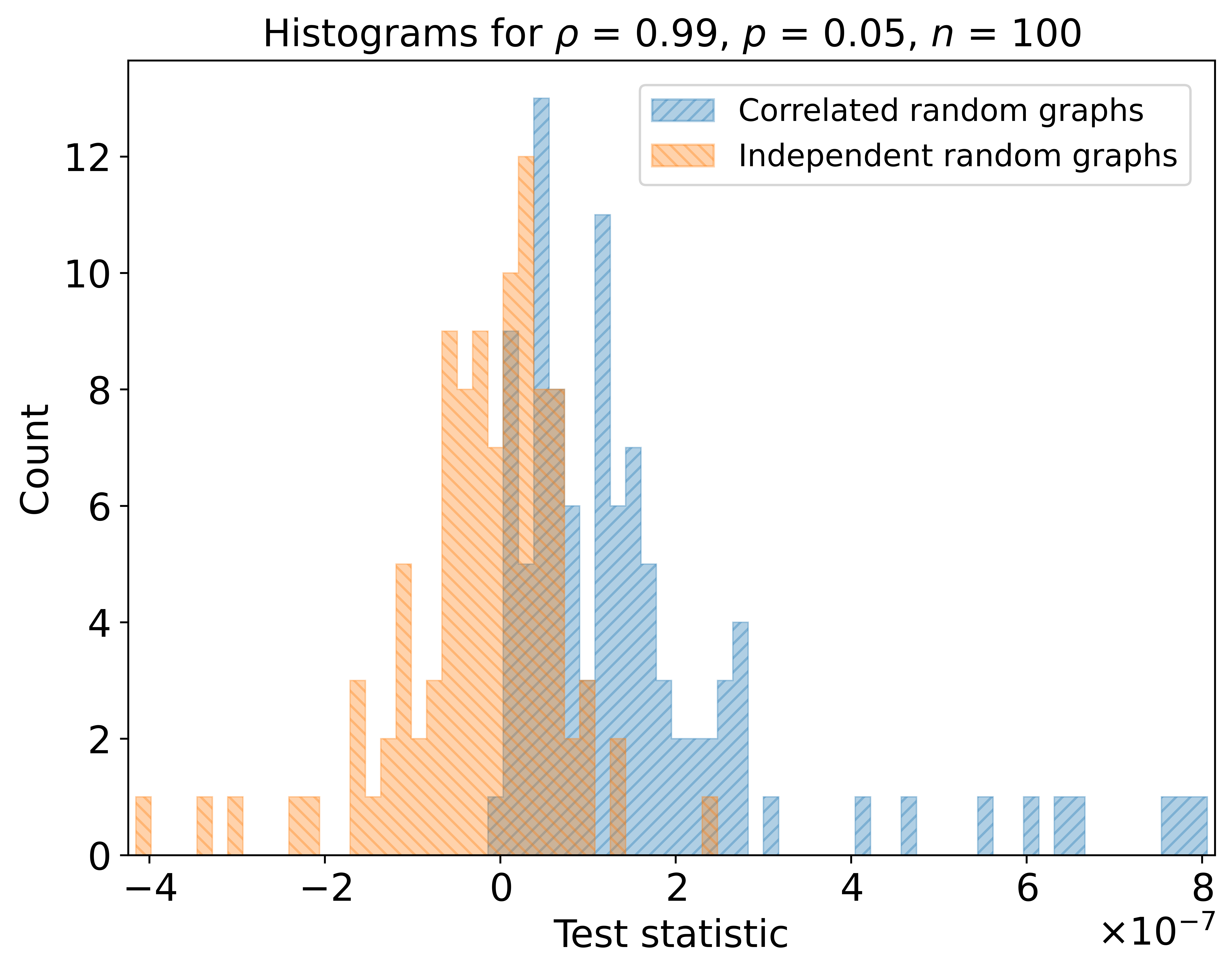}%
}\hspace{0.03\linewidth}
\subfloat{%
  \raisebox{1.2ex}{\includegraphics[width=0.4\linewidth]{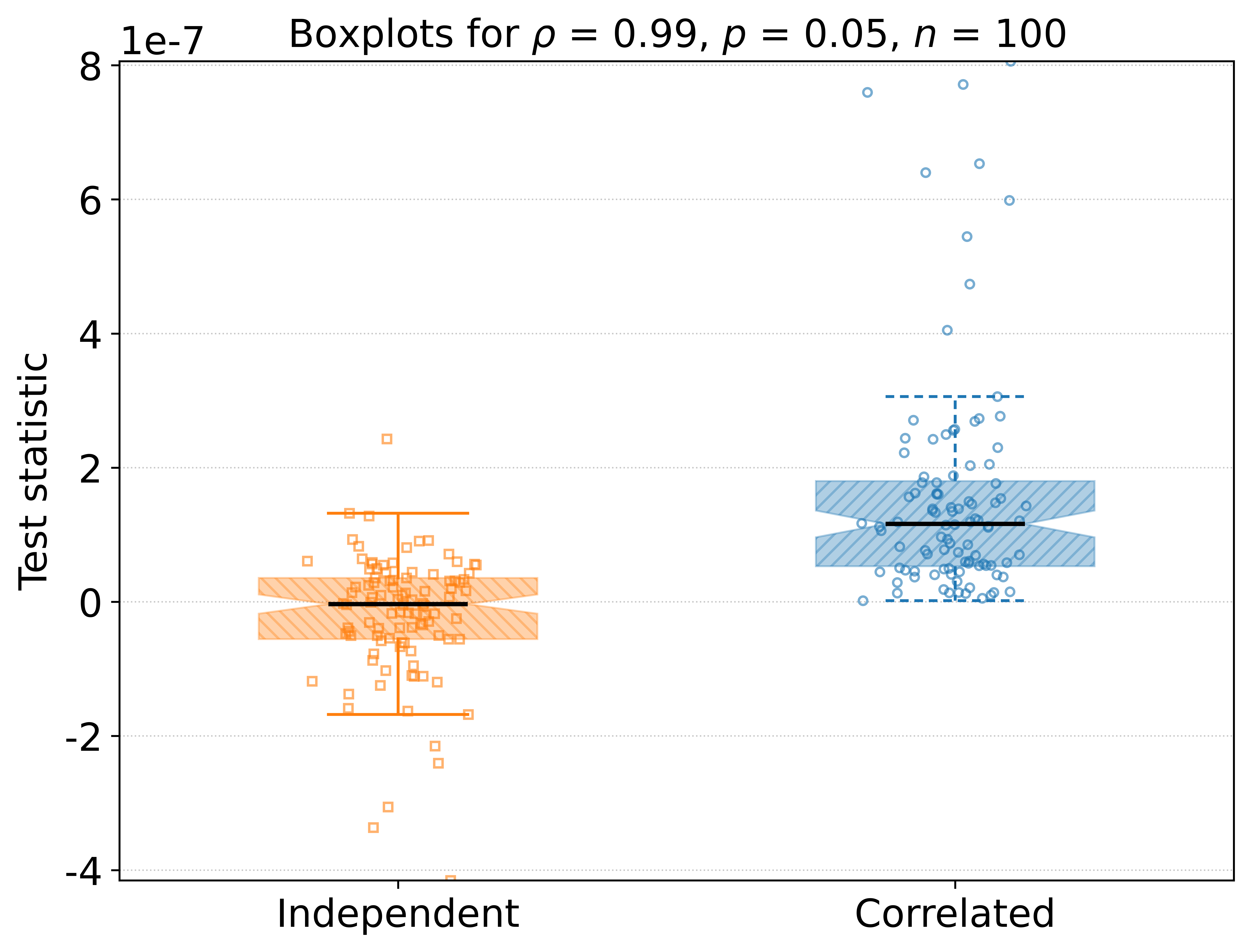}}%
}

\caption{Histograms (left) and boxplots (right) of the bounded degree motif
counting statistic $\maT_{\maM(N_\sfe,d)}$ with $N_\sfe=d=4$ for $n=100$,
$p=0.05$, and $\rho=0.99$.}
\label{fig:histogram}
\end{figure}

To compare our test statistic across settings, we plot receiver operating characteristic (ROC) curves by sweeping the detection threshold and reporting the true positive rate (one minus Type II error) against the false positive rate (Type I error) under different parameter choices. We also report the area under the ROC curve (AUC): a random classifier yields AUC = 0.5, whereas AUC = 1 corresponds to complete separation. Hence, larger AUC values indicate better discriminative performance of our statistic.

\begin{figure}[htbp] 
\centering

\subfloat{%
  \includegraphics[width=0.4\linewidth]{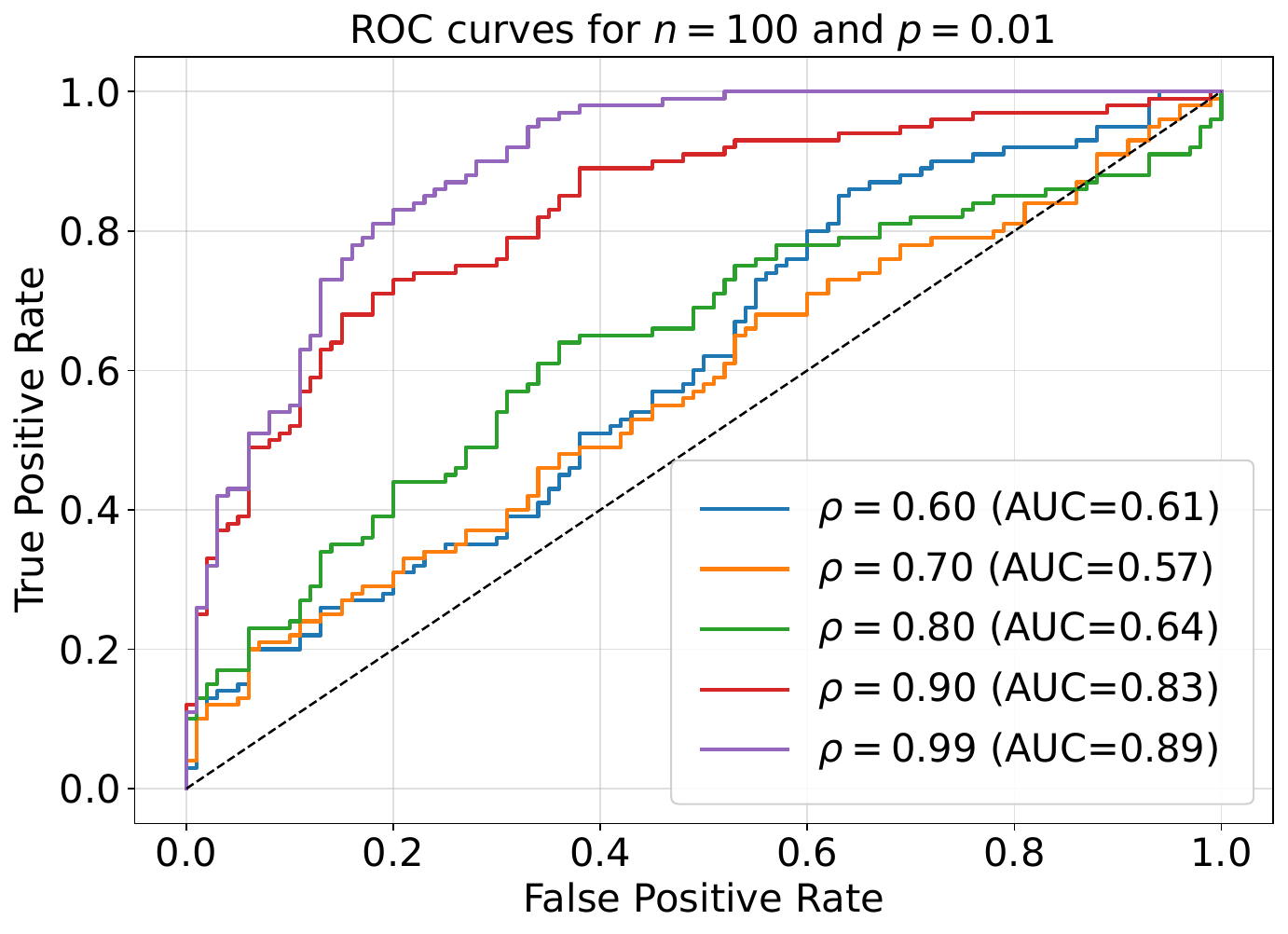}%
}\hspace{0.03\linewidth}
\subfloat{%
  \includegraphics[width=0.4\linewidth]{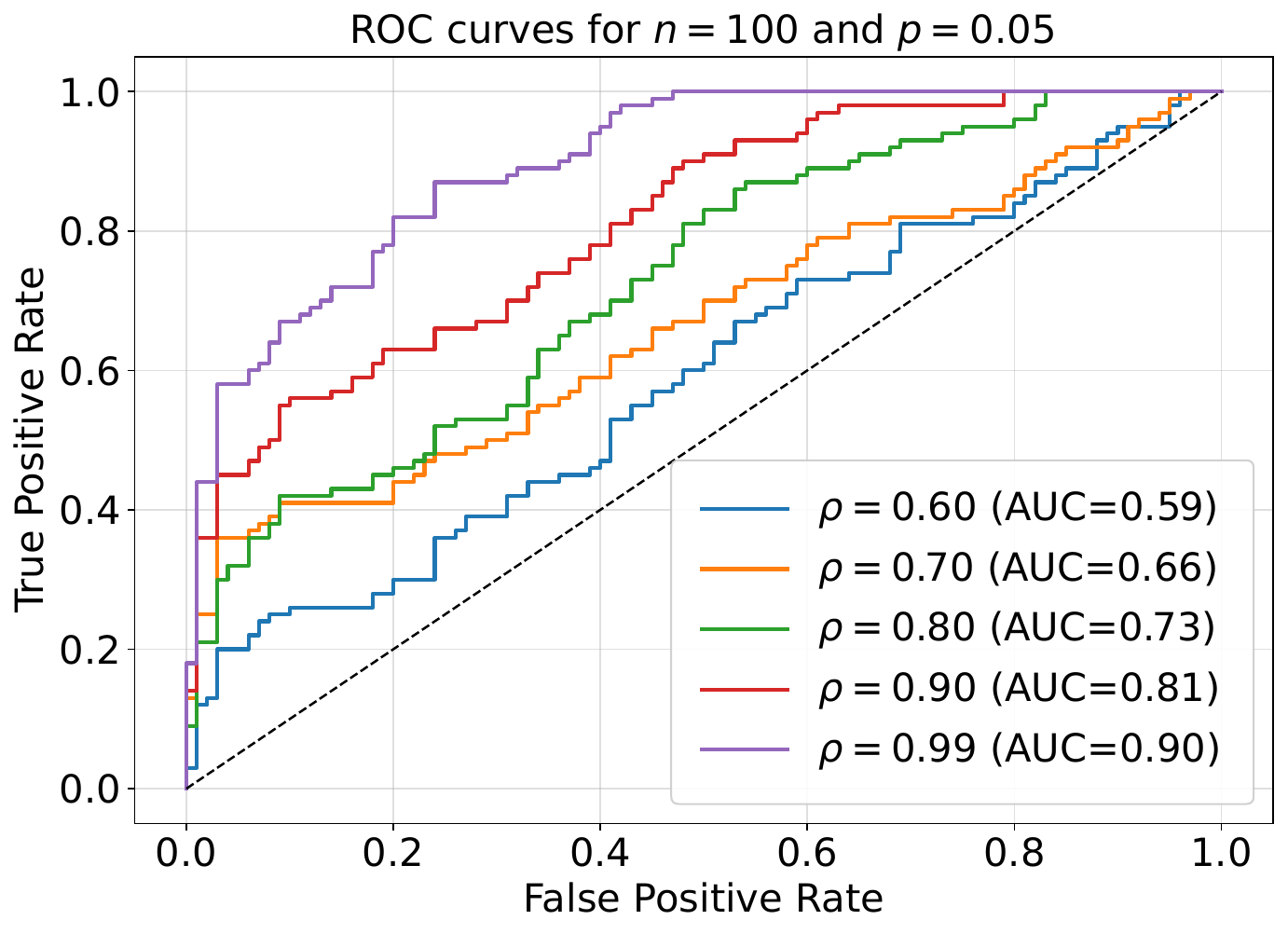}%
}

\subfloat{%
  \includegraphics[width=0.4\linewidth]{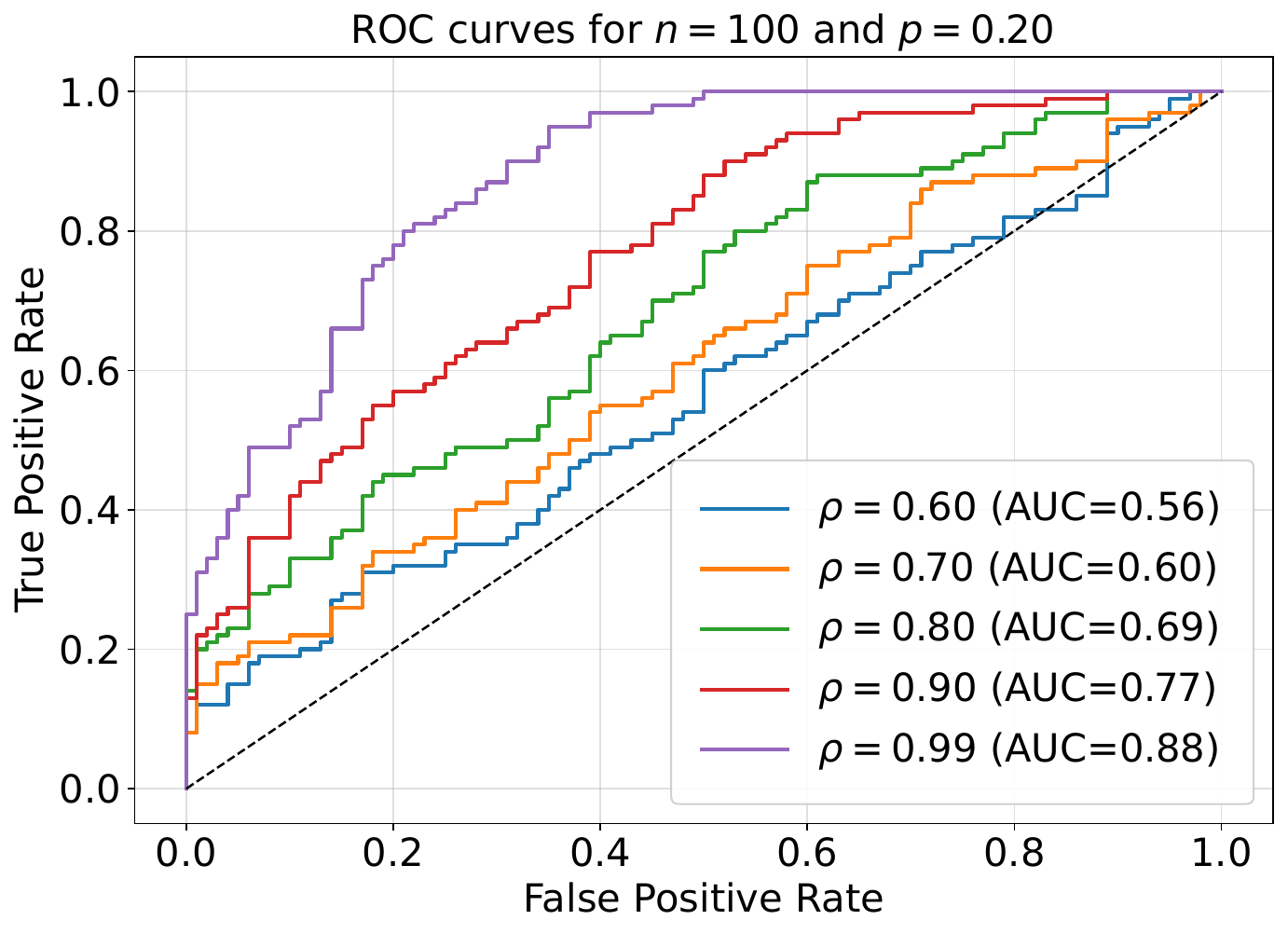}%
}\hspace{0.03\linewidth}
\subfloat{%
  \includegraphics[width=0.4\linewidth]{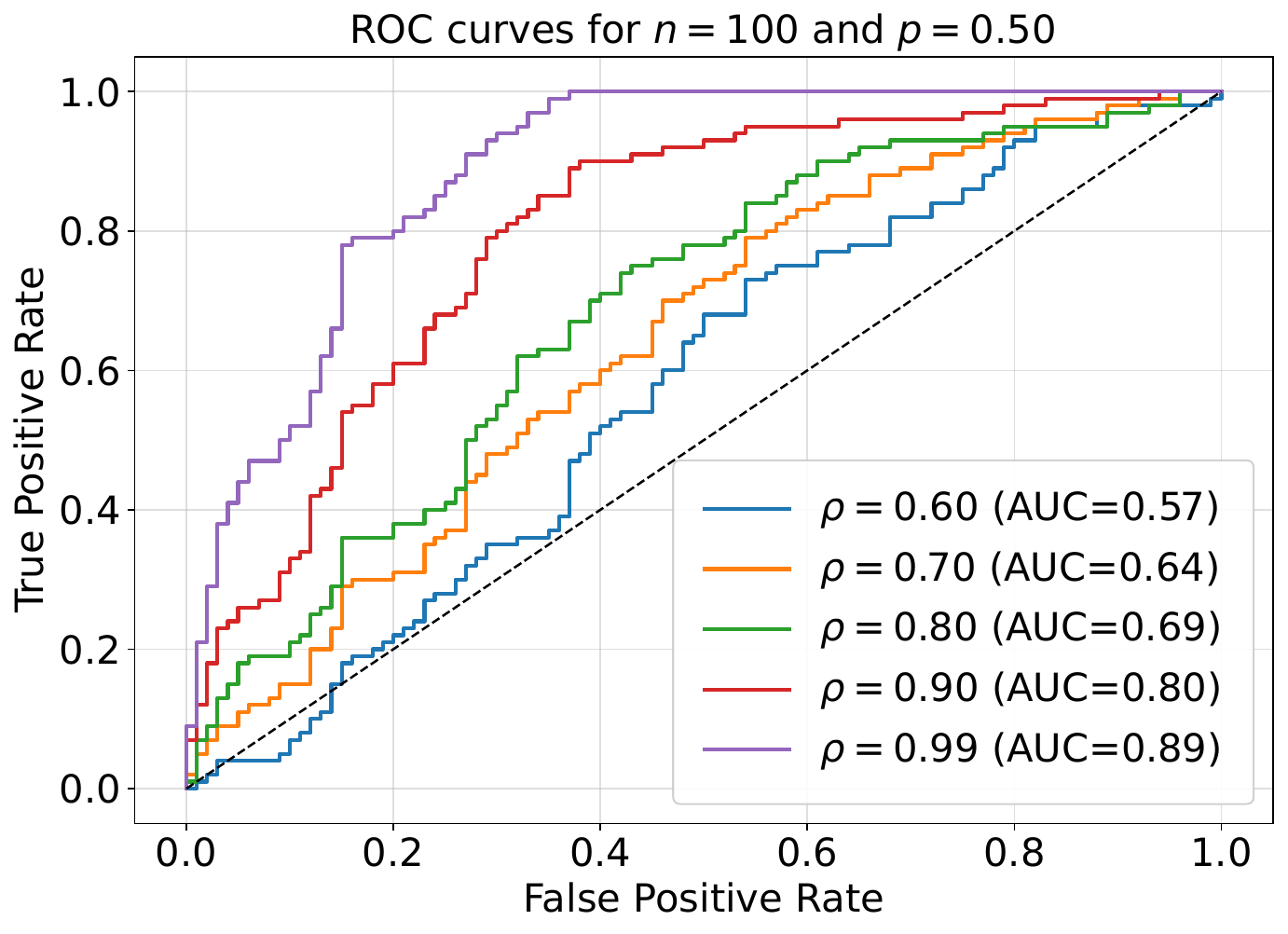}%
}

\caption{Comparison of the proposed test statistic $\maT_{\maM(N_\sfe,d)}$ with $N_\sfe = d = 4$ for fixed $p$ and varying correlation parameter $\rho\in \sth{0.6,0.7,0.8,0.9,0.99}$.}
\label{fig:ROC-E4-byrho}
\end{figure}

In Figure~\ref{fig:ROC-E4-byrho}, for each plot, we fix $n = 100, N_\sfe = d=4$ and $p\in \sth{0.01, 0.05, 0.2,0.5}$, and vary $\rho \in \sth{0.6,0.7,0.8,0.9,0.99}$. As $\rho$ increases from $0.6$ to 0.99, the ROC curves bend further toward the ideal upper-left corner (0,1) and uniformly dominate those at smaller $\rho$. Consequently, the area under the curve (AUC) grows monotonically with $\rho$ in all four settings. 
This pattern reflects improved separability of the test statistic between the null and alternative as correlation strengthens, with performance gains visible across all values of $p$ (the dashed diagonal shows the random-classifier baseline). We also plot the ROC curve for $N_\sfe=d=3$ in Figure~\ref{fig:ROC-E3-byrho}. 
Compared with the \(N_{\sfe}=3\) setting, the overall performance for $N_\sfe = 4$
exhibits noticeable improvement—most curves achieve higher true-positive rates, and the average AUCs across panels are larger—suggesting that incorporating 4-edge motifs enhances the discriminative power of the test statistic under comparable sample sizes.


\begin{figure}[htbp]
  \centering

  \subfloat{%
    \includegraphics[width=0.4\linewidth]{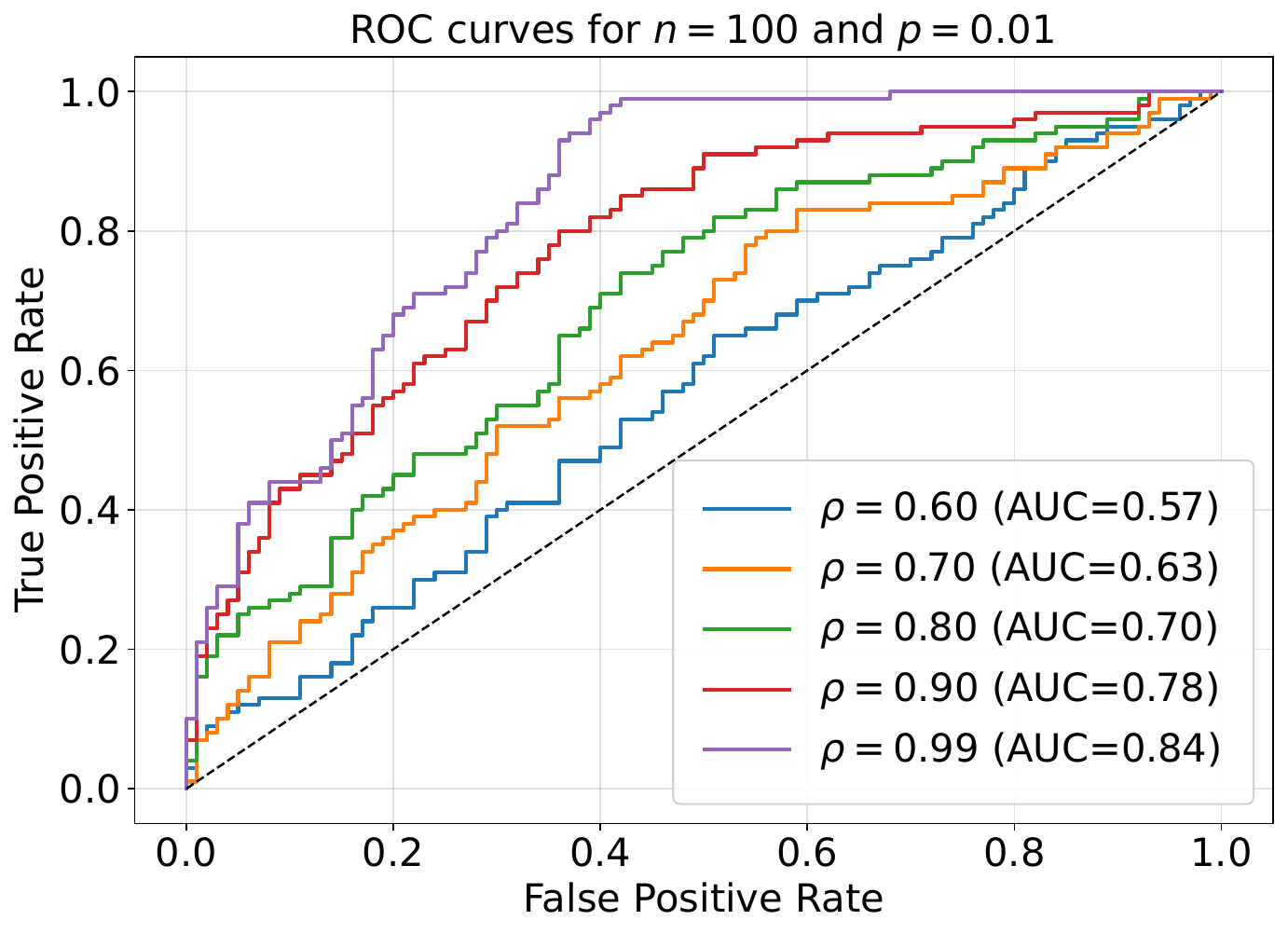}%
  }\hspace{0.03\linewidth}
  \subfloat{%
    \includegraphics[width=0.4\linewidth]{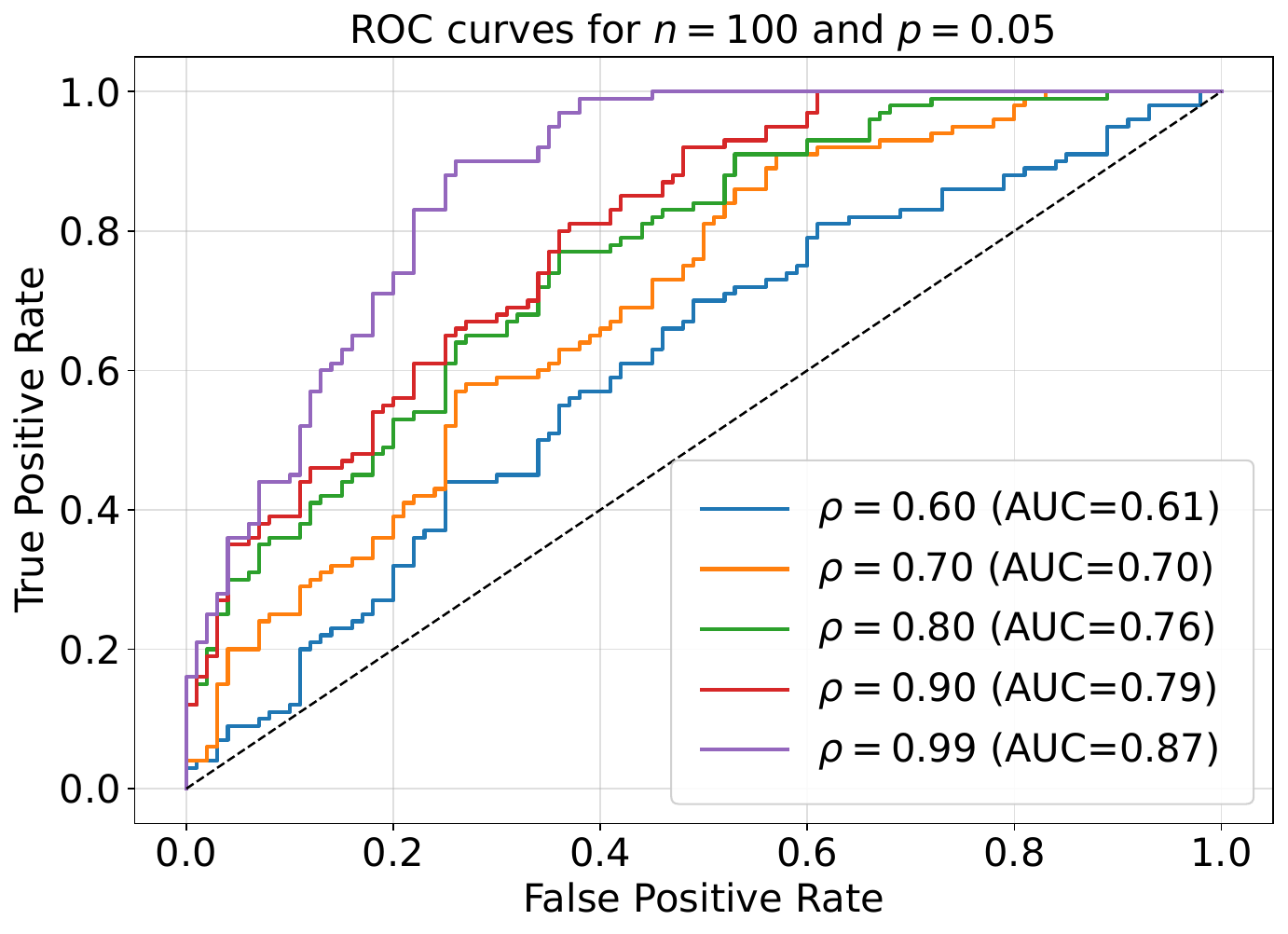}%
  }

\subfloat{%
    \includegraphics[width=0.4\linewidth]{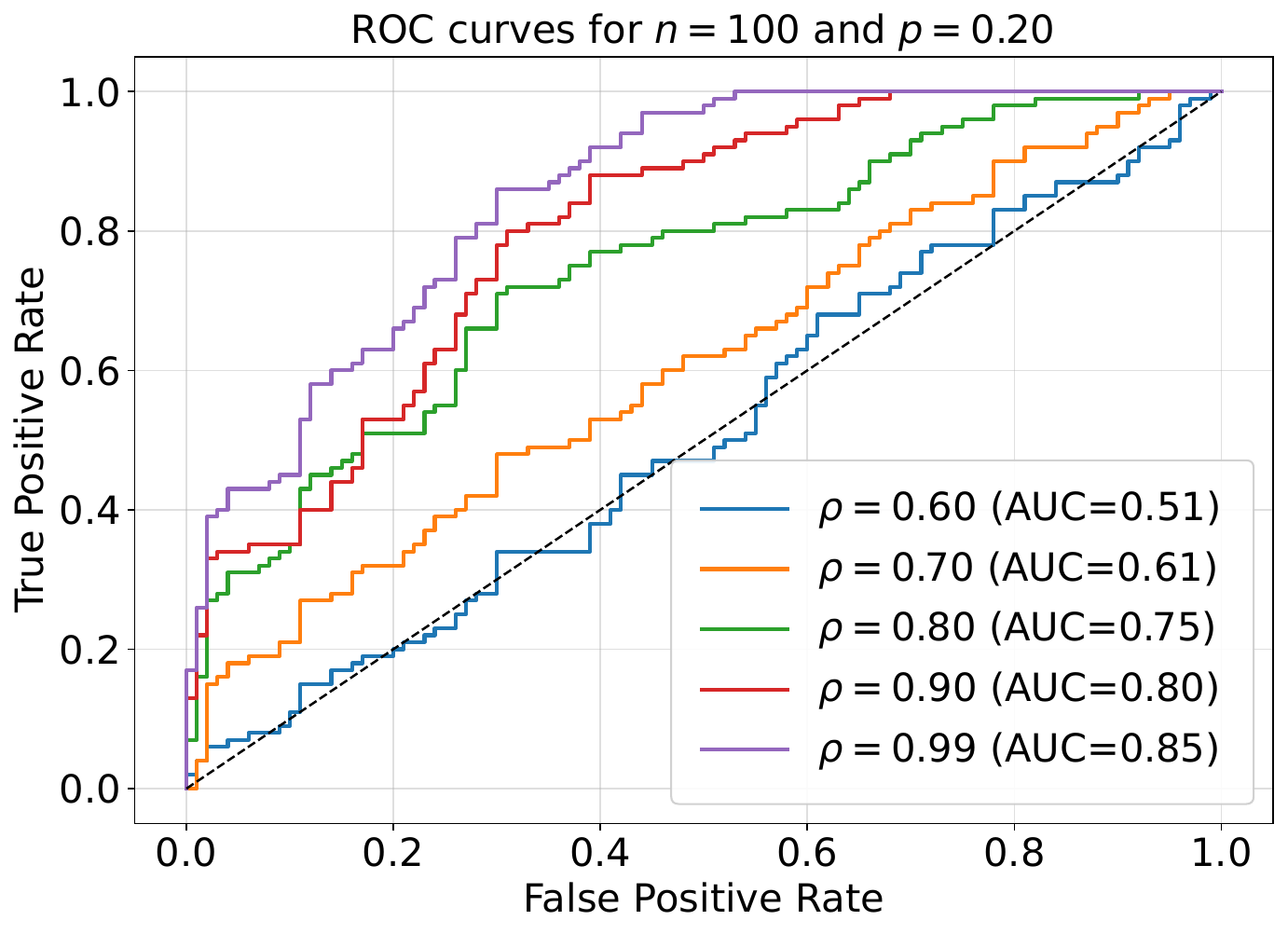}%
  }\hspace{0.03\linewidth}
  \subfloat{%
    \includegraphics[width=0.4\linewidth]{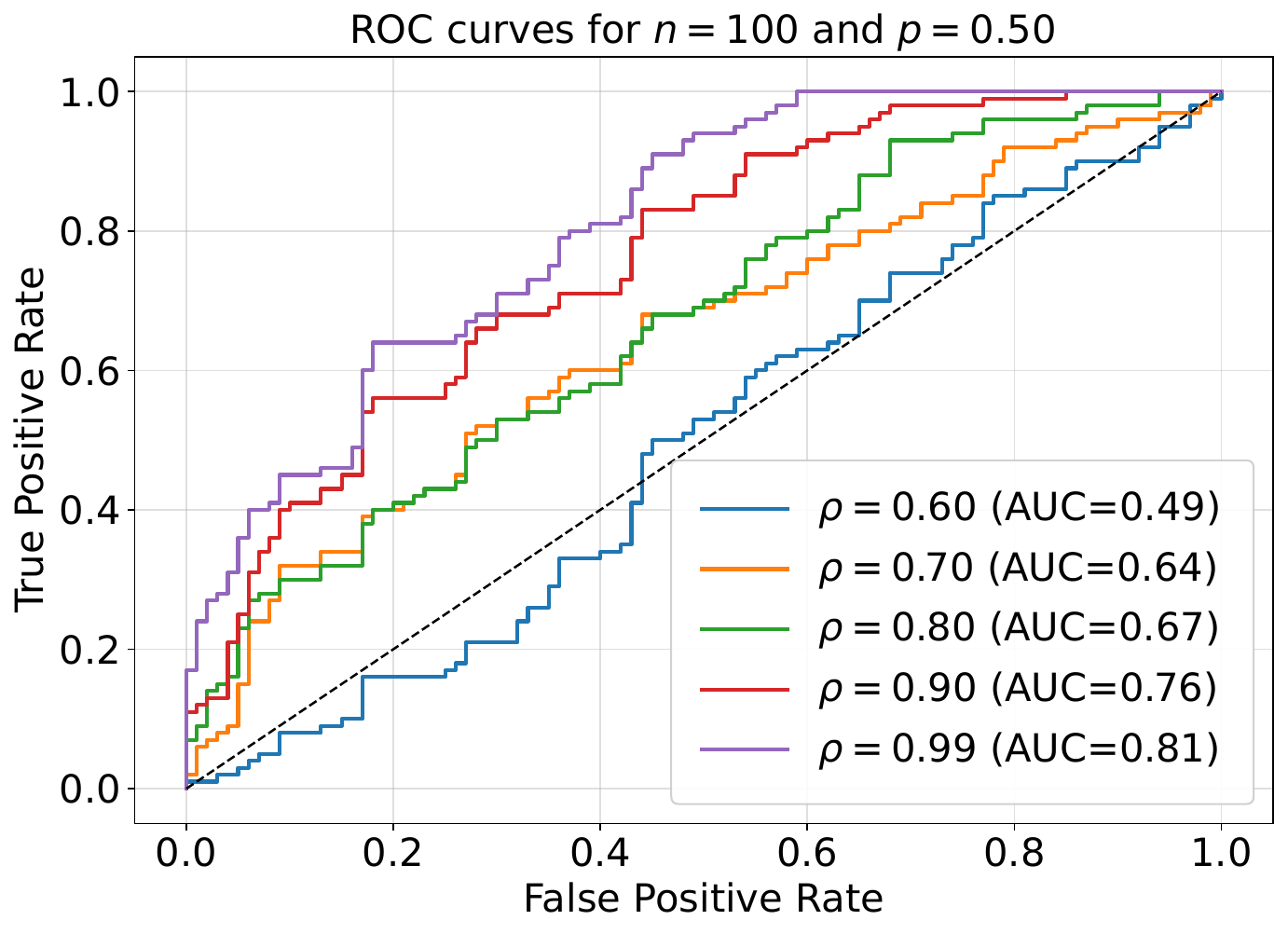}%
  }

  \caption{Comparison of the proposed test statistic $\maT_{\maM(N_\sfe,d)}$ with $N_\sfe = d = 3$ for fixed $p$ and varying correlation parameter $\rho\in \{0.6,0.7,0.8,0.9,0.99\}$.}
  \label{fig:ROC-E3-byrho}
\end{figure}

In Figure~\ref{fig:ROC-E3-byrho}, for each plot, we fix $n = 100, N_\sfe = d=4$ and $p\in \sth{0.01, 0.05, 0.2,0.5}$, and vary $\rho \in \sth{0.6,0.7,0.8,0.9,0.99}$. The qualitative behavior closely resembles that in Figure~\ref{fig:ROC-E3-byrho}: as $\rho$ increases, the ROC curves generally shift toward the ideal upper-left corner (0,1), and the AUC values increase accordingly. 
In Figure~\ref{fig:ROC-E4-byp}, we fix $n = 100$, \(N_{\sfe} = d = 4\), and $\rho \in \{0.7, 0.8, 0.9, 0.99\}$, while varying $p \in \{0.01, 0.05, 0.2, 0.5\}$. We observe that our test statistic performs consistently well for $p \in \{0.05, 0.2, 0.5\}$, whereas its performance slightly deteriorates when $p = 0.01$. This behavior is consistent with our theoretical conditions: when $n = 100$, we have $n^{-2/3} \approx 0.046$, implying that $p = 0.01$ falls below the regime $p \ge n^{-2/3}$ required in Theorems~\ref{thm:bd-degree-motif-main} and~\ref{thm:bd-degree-motif-general}. In contrast, the other three cases satisfy this condition and thus exhibit stronger empirical power, aligning well with the theoretical predictions.
\begin{figure}[htbp] 
  \centering

  \subfloat{%
    \includegraphics[width=0.4\linewidth]{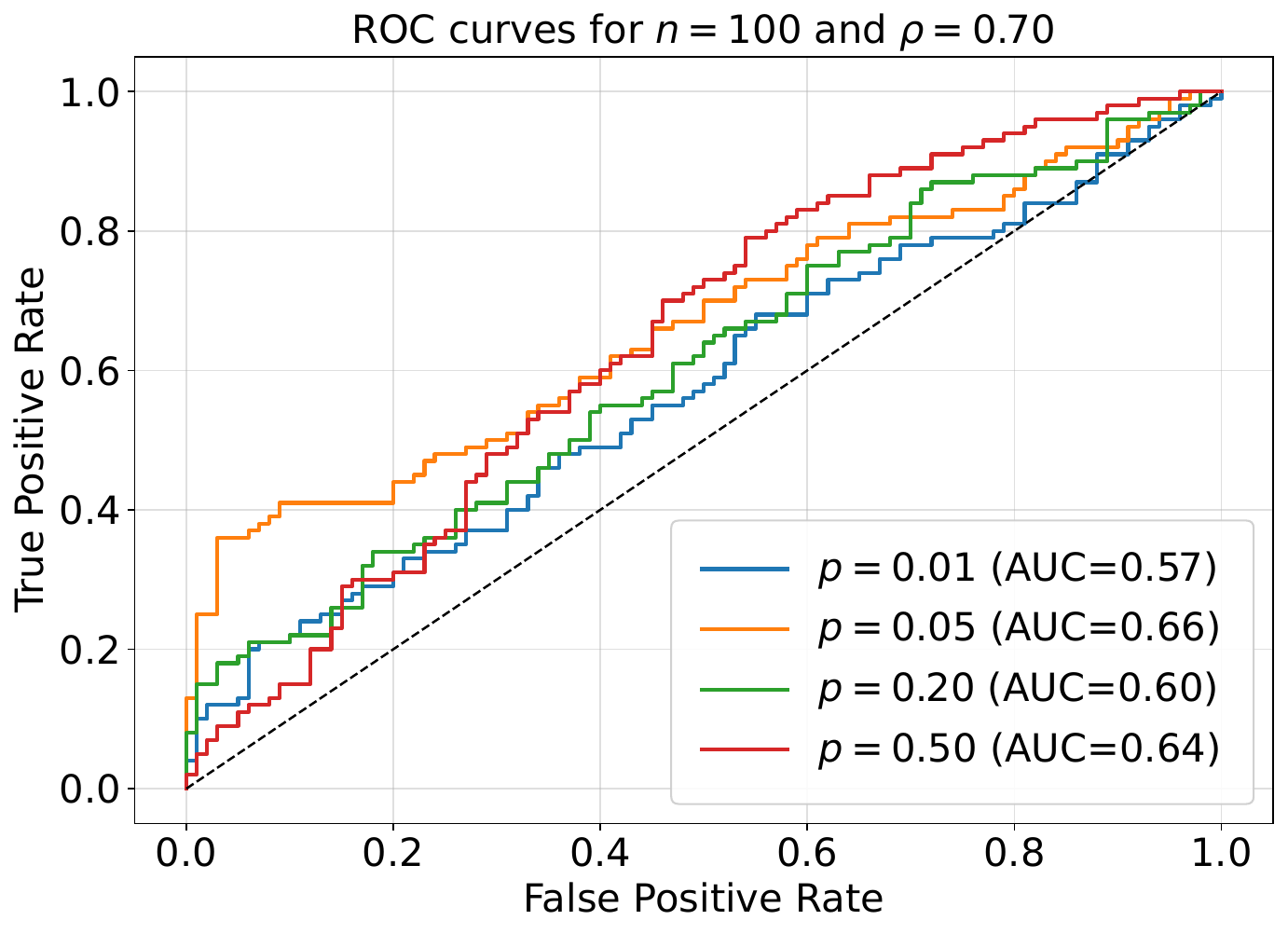}%
  }\hspace{0.03\linewidth}
  \subfloat{%
    \includegraphics[width=0.4\linewidth]{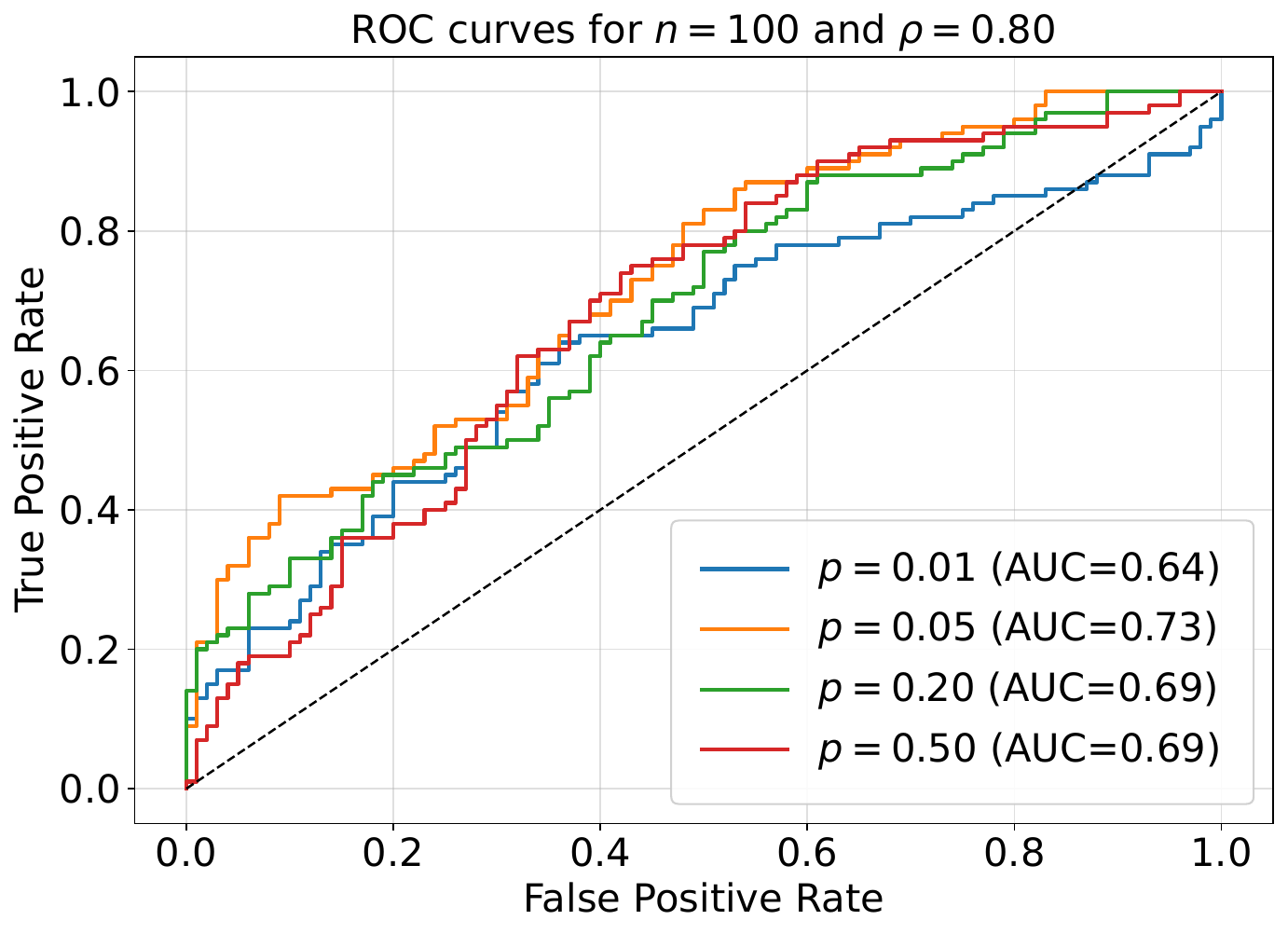}%
  }

  \subfloat{%
    \includegraphics[width=0.4\linewidth]{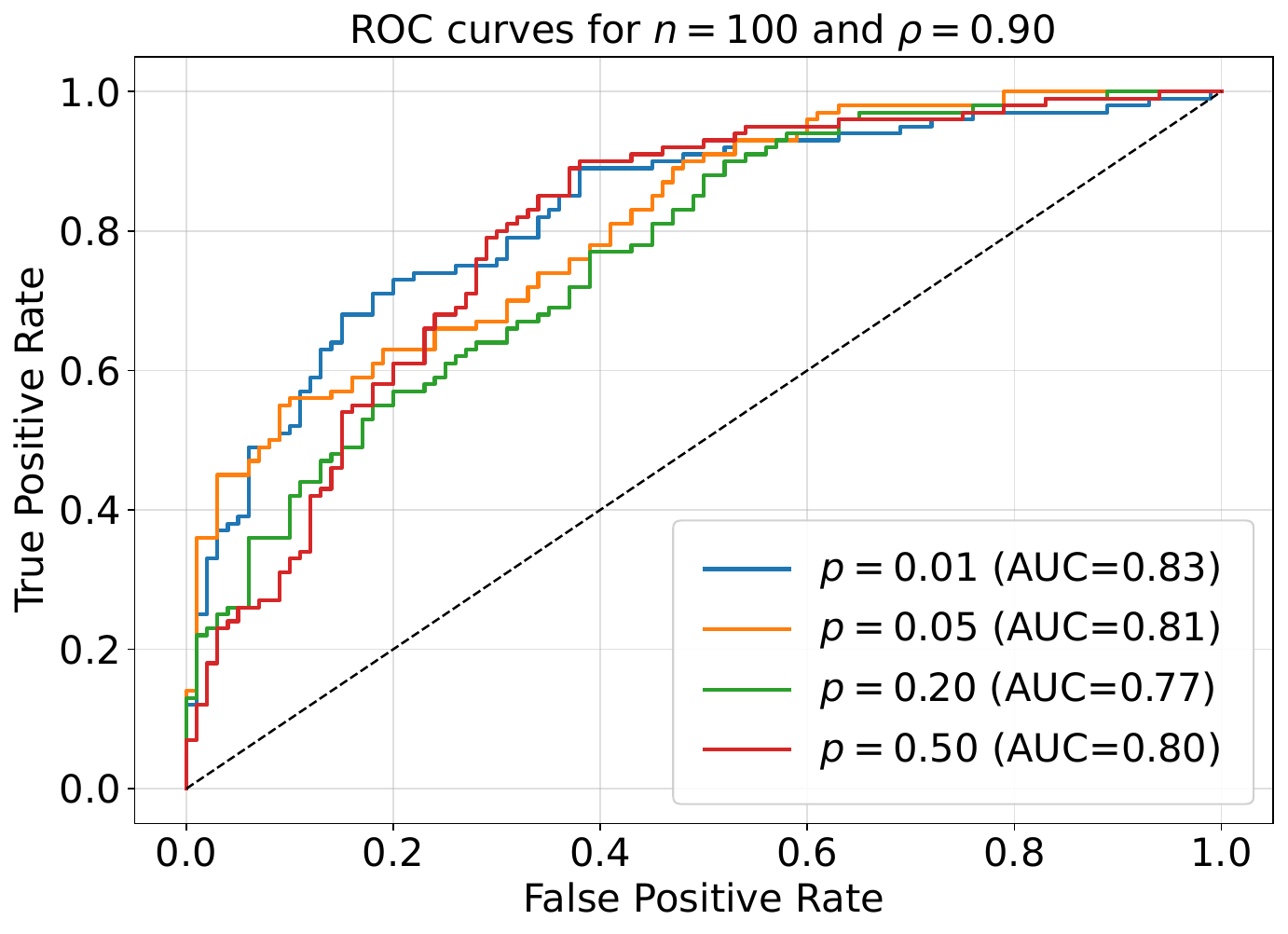}%
  }\hspace{0.03\linewidth}
  \subfloat{%
    \includegraphics[width=0.4\linewidth]{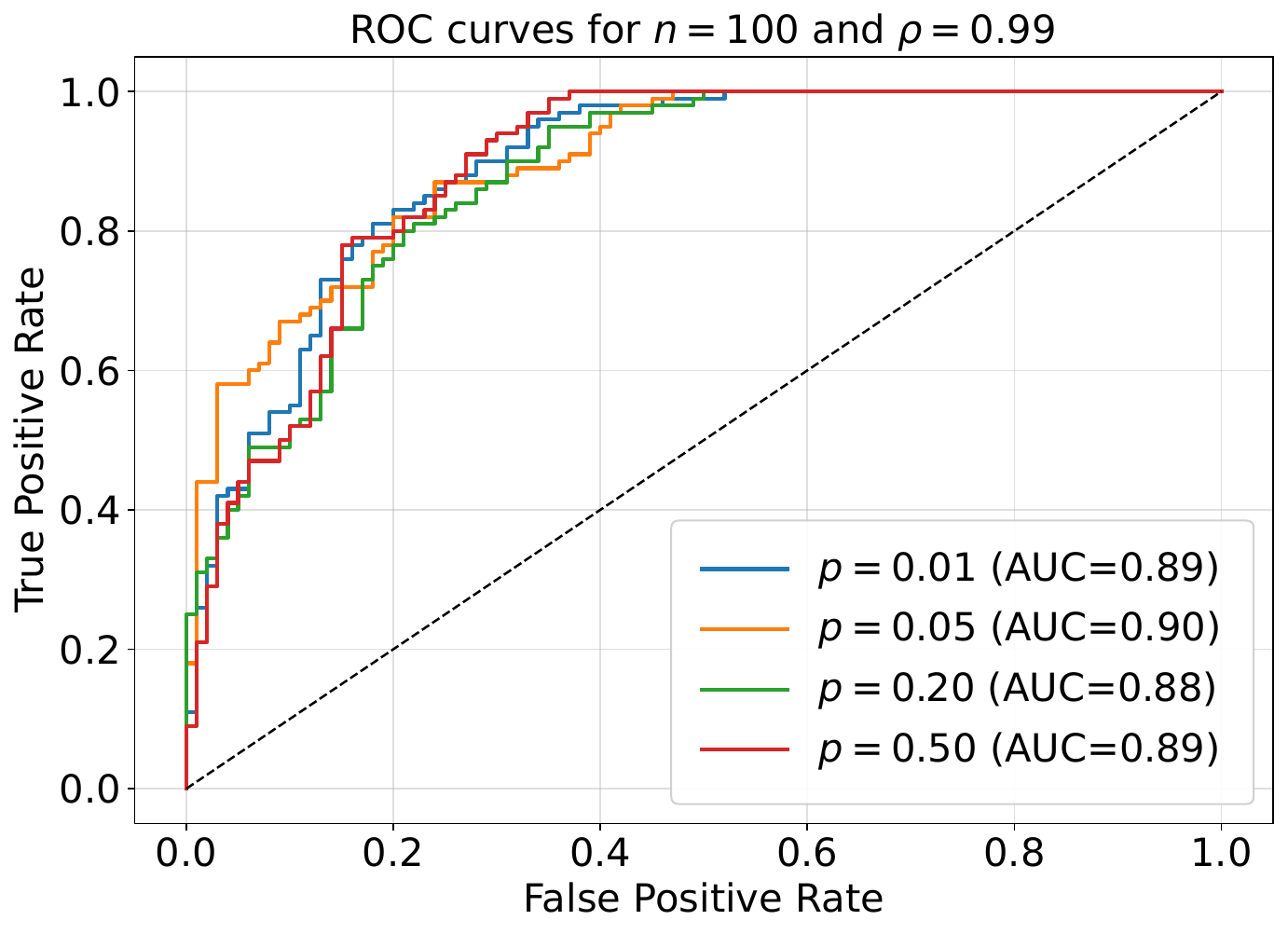}%
  }

  \caption{Comparison of the proposed test statistic $\maT_{\maM(N_\sfe,d)}$ with $N_\sfe = d = 4$ for fixed $\rho$ and varying parameter $p\in \{0.01,0.05,0.2,0.5\}$.}
  \label{fig:ROC-E4-byp}
\end{figure}

We benchmark our statistic $\maT_{\maM(N_\sfe,d)}$ against simple subgraph counting baselines (cycle counts and tree counts).
In Figure~\ref{fig:baseline_syn}, we fix $n=100$ and consider $p\in\{0.01,0.05,0.2,0.5\}$ while varying $\rho\in\{0.6,0.7,0.8,0.9,0.99\}$.
For each configuration, we compute ROC curves and report the AUC.
Across all $p$, AUC generally increases with $\rho$. Our bounded degree motif statistic consistently matches or outperforms the cycle counting and tree counting baselines in most regimes, with the advantage most pronounced at moderate to high correlations (e.g., $\rho\ge 0.8$) and for denser graphs ($p\in\{0.2,0.5\}$). These results indicate that aggregating information from bounded degree motifs yields higher detection power than relying on a single family of small subgraphs.

\begin{figure}[htbp] 
\centering

\subfloat{%
  \includegraphics[width=0.4\linewidth]{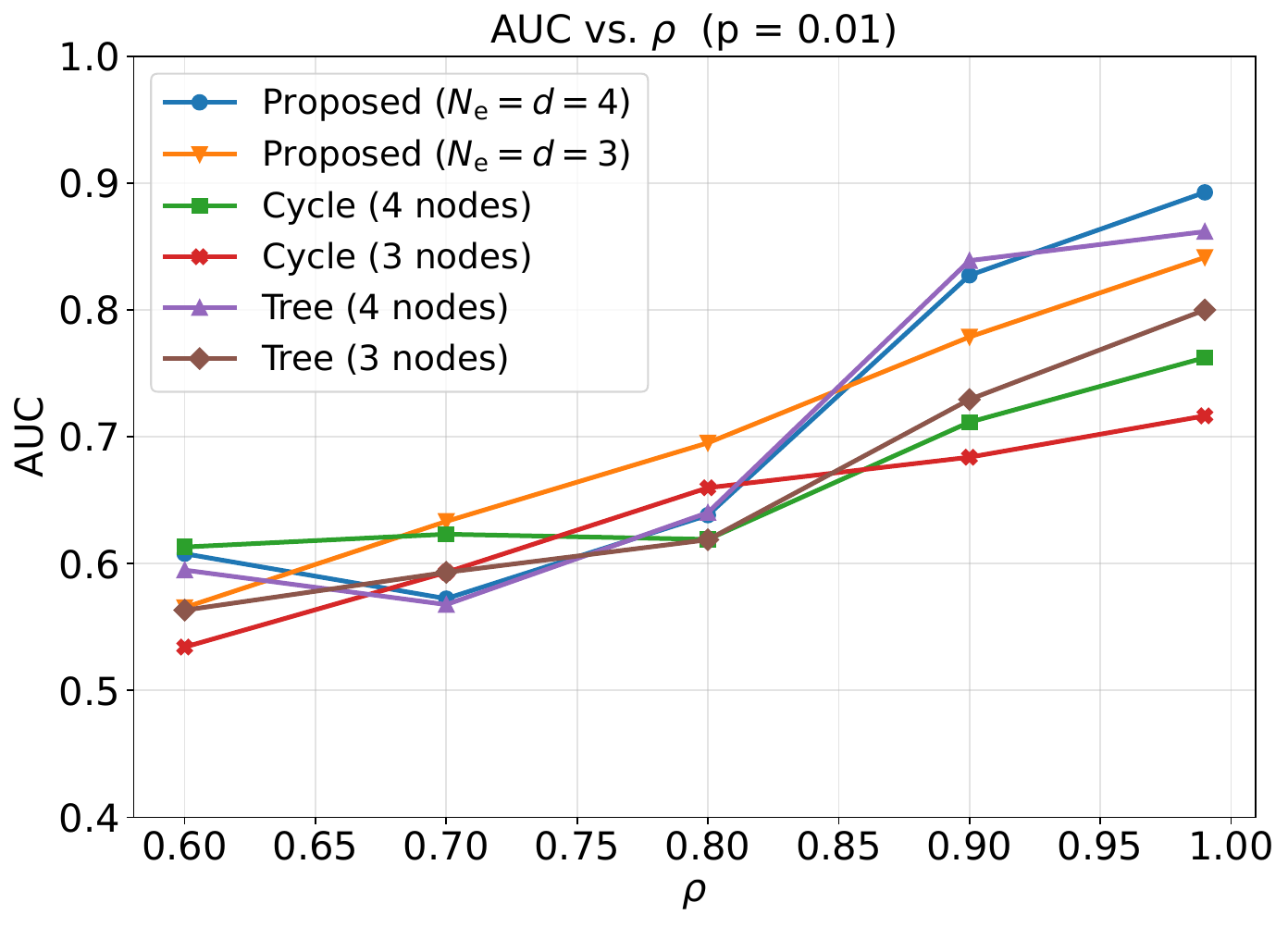}%
}\hspace{0.03\linewidth}
\subfloat{%
  \includegraphics[width=0.4\linewidth]{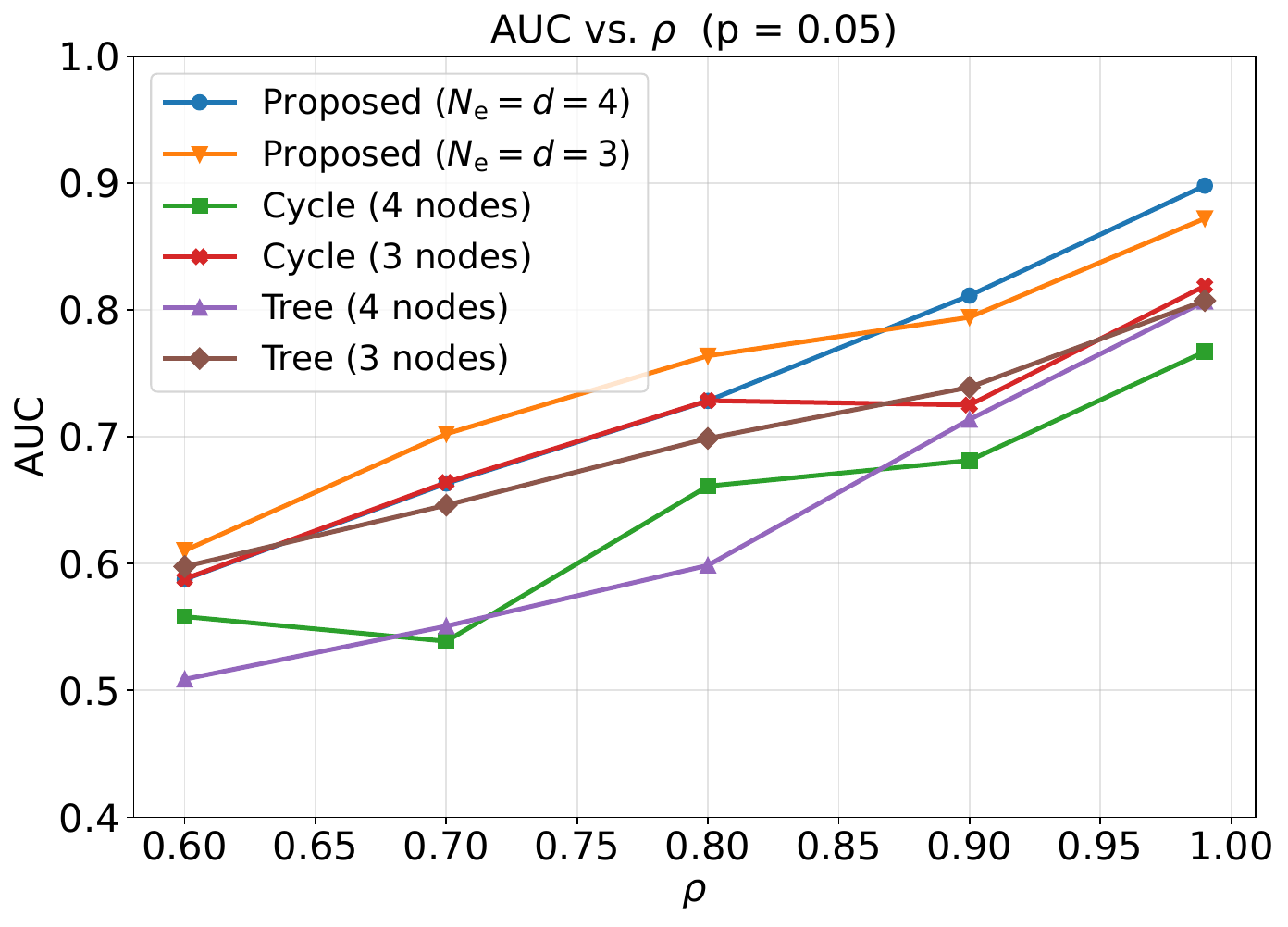}%
}

\subfloat{%
  \includegraphics[width=0.4\linewidth]{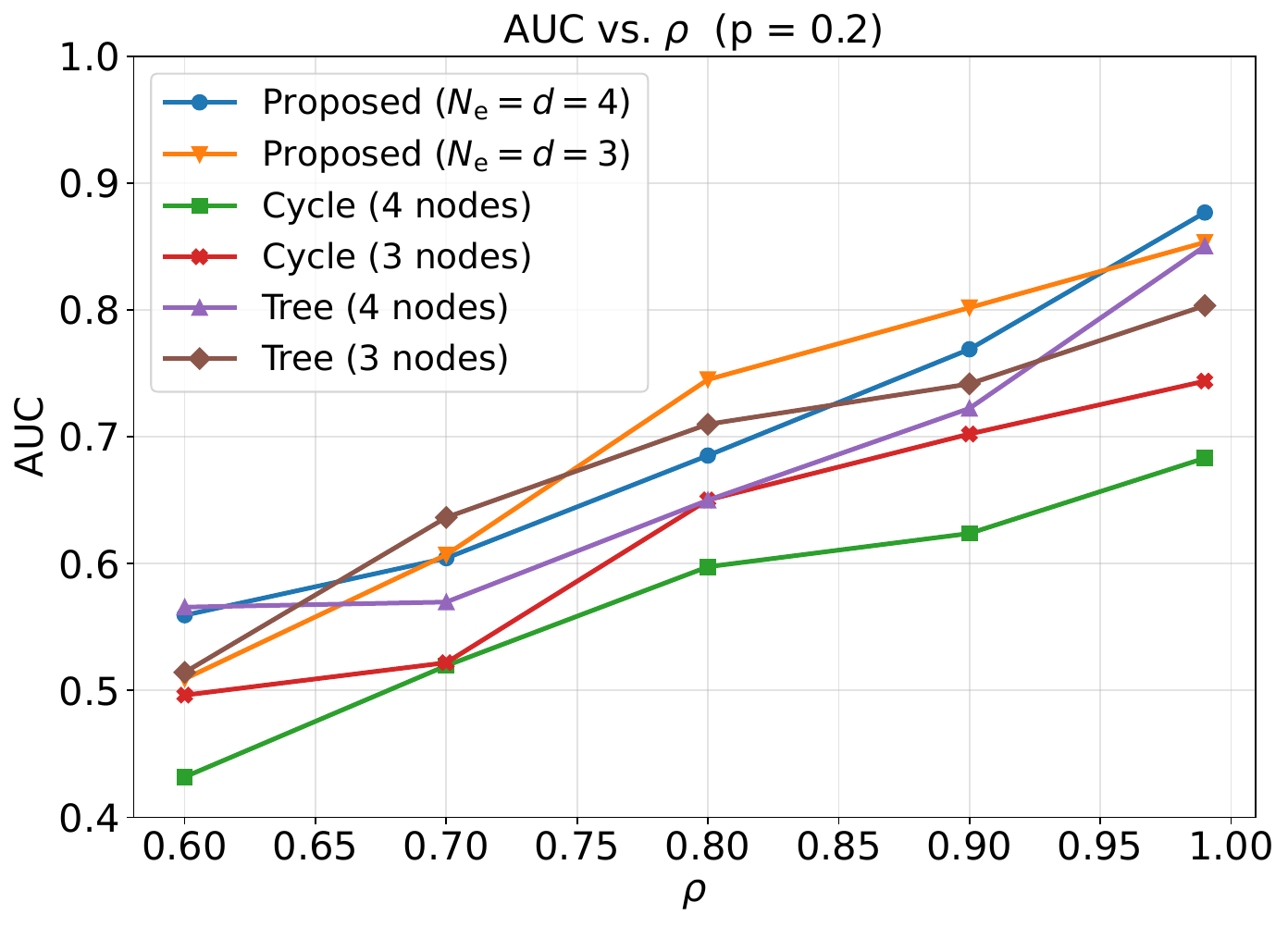}%
}\hspace{0.03\linewidth}
\subfloat{%
  \includegraphics[width=0.4\linewidth]{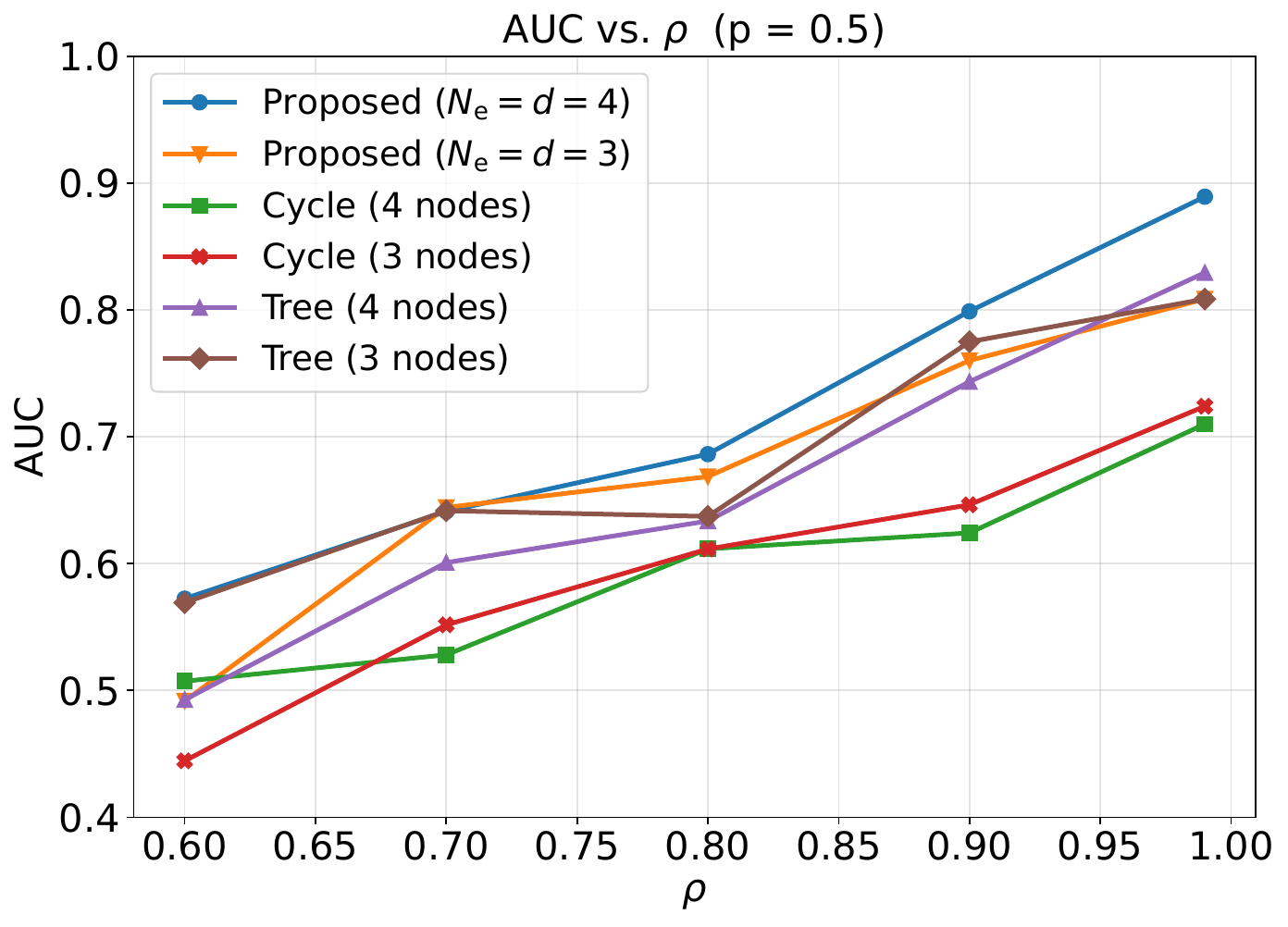}%
}

\caption{Comparison with counting-based baselines on synthetic graphs.}
\label{fig:baseline_syn}
\end{figure}



\section{Discussions and Future Directions}\label{sec:discussion}

This paper considers the hypothesis testing problem for the correlated \ER model, where under the null hypothesis two \ER graphs are independent, and under the alternative hypothesis, they are correlated through a latent permutation. We propose a polynomial-time algorithm based on counting motifs with bounded degree. Our main contributions are summarized as follows.
\begin{itemize}
\item \emph{Homomorphism number and bounded degree motifs.}
Since the homomorphism number effectively captures graph properties, we consider computing the injective homomorphism number as the test statistic. We establish the connection between homomorphism numbers and motif counting, which naturally motivates the idea of counting motifs. Instead of focusing on tree structures, which are crucial in the \ER model, we consider a more general family: bounded degree motifs. These structures frequently appear not only in graph models but also in real-world data. Notably, these motifs admit provable detection guarantees: they overcome the tree-counting barrier (which requires $\rho^2 \ge \alpha \approx 0.338$) and enable detection for any fixed constant $\rho>0$ in the regime $p \ge n^{-1+\delta}$.
\item \emph{Polynomial-time algorithm and computational hardness.}
We propose a polynomial-time algorithm that succeeds in detection for any constant $\rho $ and $p\ge n^{-1+\delta}$. This result overcomes the limitation that the correlation coefficient $\rho \ge \sqrt{\alpha}$ where $\alpha$ is the Otter's constant in tree-counting methods, as discussed in~\cite{mao2024testing}.
The bounded degree motif counting statistic achieves detection with a computational complexity of $n^{\en(\maM)}$. Moreover, our algorithm aligns with the hardness conjecture within the framework of low-degree polynomial algorithms, which conjectures that any degree-$o(\rho^{-1})$ polynomial algorithm fails for detection~\citep{ding2023low,li2025algorithmic}.
\end{itemize}

Beyond the main results, several important directions merit further investigation:
\begin{itemize}
    \item \emph{Recovery problem.} The bounded degree motif counting statistic can also be applied to the recovery problem, where two \ER graphs are assumed to be correlated through some latent permutation. Exploring this extension could lead to new insights into both exact and partial recovery guarantees.
    \item \emph{Optimal degree.} We show that our motif counting algorithm  with degree $\Theta\pth{\frac{1}{\rho^{-2d/(d-2)}}}$ achieves detection. In contrast,  it is conjectured that any algorithm with degree $o\pth{\frac{1}{\rho}}$ fails for detection~\citep{ding2023low,li2025algorithmic}. An interesting open question is determining the optimal degree for a polynomial-time algorithm in the detection problem.
    \item \emph{General graph model.} Although there have been many analyses of the correlated \ER graph, a key limitation is that the model is idealized and does not fully capture the characteristics of real-world networks. To address the generality of the model, recent works have explored various other random graph models, including partially correlated \ER model~\citep{huang2024information}, inhomogeneous model~\citep{ding2023efficiently}, correlated random geometric model~\citep{wang2022random,gong2024umeyama}, correlated stochastic block model~\citep{chen2024computational,chen2025detecting}, planted structure model~\citep{mao2024informationtheoretic}, multiple correlated \ER model~\citep{ameen2024exact}, and corrupted model~\citep{ameen2024robust}. It is of interest to explore whether our results can be extended to more general graph models.
\end{itemize}

\appendix

\section{Proof of Propositions}

\subsection{Proof of Proposition~\ref{prop:TypeI-admissible}}\label{apd:proof-prop-admi-typeI}

Recall that $\bar{G}$ defines the weighted graph with weighted edge $\beta_{uv}(\bar{G}) = \indc{uv\in E(G)} - p$ for $G\sim \maG(n,p)$ and $\inj(\sfM,\bar{G})$ defined in~\eqref{eq:def_of_inj}. Under the null hypothesis distribution $\maP_0$, $\bar{G}_1$ and $\bar{G}_2$ are independent. Therefore, for any $\sfM\in \maM$, we have \begin{align*}
    \mathbb{E}_{\maP_0} \qth{\inj(\sfM,\bar{G}_1) \inj(\sfM,\bar{G}_2)} &= \mathbb{E}_{\maP_0}\qth{\inj(\sfM,\bar{G}_1)} \mathbb{E}_{\maP_0}\qth{\inj(\sfM,\bar{G}_2)}. 
\end{align*}
    Since $\mathbb{E}_{\maP_0}\qth{\beta_e(\bar{G}_i)} = 0$ for any $e\in E(\bar{G}_i)$ and $i\in \sth{1,2}$, we have \begin{align*}
\mathbb{E}_{\maP_0}\qth{\inj(\sfM,\bar{G}_1)} &= \sum_{\substack{\varphi:V(\sfM)\mapsto V(\bar{G}_1)\\\varphi\text{ injective}}} \mathbb{E}_{\maP_0}\qth{\prod_{e\in E(\sfM)} \beta_{\varphi(e)}(\bar{G}_1)}\\&=\sum_{\substack{\varphi:V(\sfM)\mapsto V(\bar{G}_1)\\\varphi\text{ injective}}}\prod_{e\in E(\sfM)} \mathbb{E}_{\maP_0}\qth{\beta_{ \varphi(e)}(\bar{G}_1)}=0,\\
\mathbb{E}_{\maP_0}\qth{\inj(\sfM,\bar{G}_2)} &= \sum_{\substack{\varphi:V(\sfM)\mapsto V(\bar{G}_2)\\\varphi\text{ injective}}} \mathbb{E}_{\maP_0}\qth{\prod_{e\in E(\sfM)} \beta_{\varphi(e)}(\bar{G}_2)}\\&=\sum_{\substack{\varphi:V(\sfM)\mapsto V(\bar{G}_2)\\\varphi\text{ injective}}}\prod_{e\in E(\sfM)} \mathbb{E}_{\maP_0}\qth{\beta_{ \varphi(e)}(\bar{G}_1)}=0,
    \end{align*}
    where $\varphi(e)\triangleq \varphi(u)\varphi(v)$ for any edge $e=uv$. Therefore, $$\mathbb{E}_{\maP_0}\qth{\maT_{\maM}}=\sum_{\sfM\in \maM} \omega_\sfM\mathbb{E}_{\maP_0}\qth{\inj(\sfM,\bar{G}_1) \inj(\sfM,\bar{G}_2)}=0.$$ 
By Chebyshev's inequality, we have \begin{align}\label{eq:typeI-chebyshev}
  {\maP_0}\pth{\maT_{\maM}\ge \tau}= {\maP_0}\pth{\maT_{\maM}-\mathbb{E}_{\maP_0}\qth{\maT_{\maM}}\ge \tau }\le \frac{\var_{\maP_0}\qth{\maT_{\maM}}}{\tau^2}=\frac{4\var_{\maP_0}\qth{\maT_{\maM}}}{\pth{\E_{\maP_1}\qth{\maT_{\maM}}}^2}.
    \end{align}

It remains to compute $\E_{\maP_1}\qth{\maT_{\maM}}$ and $\var_{\maP_0}\qth{\maT_{\maM}}$. We first compute $\E_{\maP_1}\qth{\maT_{\maM}}$. Recall the test statistic $\maT_\maM$ defined in~\eqref{eq:est-sub-count}.
We note that \begin{align*}
    &~\E_{\maP_1}\qth{\maT_{\maM}} = \E_{\pi}\E_{\maP_1|\pi}\qth{\maT_{\maM}}\\
    =&~\sum_{\sfM\in \maM} \omega_\sfM \E_\pi\E_{\maP_1|\pi}\qth{\inj(\sfM,\bar{G}_1) \inj(\sfM,\bar{G}_2)}\\
    =&~\sum_{\sfM\in \maM}\omega_\sfM\sum_{\substack{\varphi_1:V(\sfM)\mapsto V(\bar{G}_1)\\\varphi_1\text{  injective}}}\sum_{\substack{\varphi_2:V(\sfM)\mapsto V(\bar{G}_2)\\\varphi_2\text{ injective}}} \E_\pi\E_{\maP_1|\pi}\qth{\prod_{e\in E(\sfM)} \beta_{\varphi_1(e)}(\bar{G}_1)\prod_{e\in E(\sfM)} \beta_{\varphi_2(e)}(\bar{G}_2)}.
\end{align*}
We note that for a correlated pair $(e,\pi(e))$, $\E_{\maP_1|\pi} \qth{\beta_e(\bar{G}_1)\beta_{\pi(e)}(\bar{G}_2)} = \rho p (1-p)$, otherwise we have $\E_{\maP_1|\pi} \qth{\beta_e(\bar{G}_1)\beta_{e'}(\bar{G}_2)} = 0$. Therefore, for any injections $\vp_1:V(\sfM)\mapsto V(\bar{G}_1)$ and $\vp_2:V(\sfM)\mapsto V(\bar{G}_2)$, we have \begin{align}
    \nonumber \nonumber&~\E_\pi\E_{\maP_1|\pi}\qth{\prod_{e\in E(\sfM)} \beta_{\vp_1(e)}(\bar{G}_1)\prod_{e\in E(\sfM)} \beta_{\vp_2(e)}(\bar{G}_2)} \\\nonumber =&~ \pth{\rho p (1-p)}^{\en(\sfM)}\prob{\pi\circ \vp_1(E(\sfM)) = \vp_2(E(\sfM))}\\\label{eq:Epvarphi12} 
    =&~\pth{\rho p (1-p)}^{\en(\sfM)}\cdot \frac{\aut(\sfM)(n-\vn(\sfM))!}{n!},
\end{align}
where $\vp(E(\sfM))\triangleq \sth{\vp(e):e\in E(\sfM)}$ and the last equality holds because of the following three facts: (1) $\sfM$ is connected; (2) there are $\aut(\sfM)$ options for $\pi$ on $\vp_1(V(\sfM))$ when fixing $\pi(\vp_1(V(\sfM))) = \vp_2(V(\sfM))$; and (3) there are $(n-\vn(\sfM))!$ options for mapping $V(\bar{G}_1)\backslash \vp_1(V(\sfM))$ to $V(\bar{G}_2\backslash \pi\circ\vp_1(V(\sfM)))$. We then obtain
\begin{align}
    \nonumber &~\E_{\maP_1}\qth{\maT_{\maM}}\\\nonumber =&~\sum_{\sfM\in \maM}\omega_\sfM\sum_{\substack{\varphi_1:V(\sfM)\mapsto V(\bar{G}_1)\\\varphi_1\text{ injective}}}\sum_{\substack{\varphi_2:V(\sfM)\mapsto V(\bar{G}_2)\\\varphi_2\text{ injective}}} \E_\pi\E_{\maP_1|\pi}\qth{\prod_{e\in E(\sfM)} \beta_{\varphi_1(e)}(\bar{G}_1)\prod_{e\in E(\sfM)} \beta_{\varphi_2(e)}(\bar{G}_2)}\\\nonumber 
    =&~\sum_{\sfM\in \maM}\omega_\sfM\sum_{\substack{\varphi_1:V(\sfM)\mapsto V(\bar{G}_1)\\\varphi_1\text{ injective}}}\sum_{\substack{\varphi_2:V(\sfM)\mapsto V(\bar{G}_2)\\\varphi_2\text{ injective}}} \pth{\rho p (1-p)}^{\en(\sfM)}\cdot \frac{\aut(\sfM)(n-\vn(\sfM))!}{n!}\\ \label{eq:EPmaT}=&~ \sum_{\sfM\in \maM} \omega_\sfM\, \frac{(\rho p (1-p))^{\en(\sfM)}\aut(\sfM) n!}{(n-\vn(\sfM))!},
\end{align}
where the last equality follows from the fact that there are $\frac{n!}{(n-\vn(\sfM))!}$ injections $\vp_1:V(\sfM)\mapsto V(\bar{G}_1)$ and $\frac{n!}{(n-\vn(\sfM))!}$ injections $\vp_2:V(\sfM)\mapsto V(\bar{G}_2)$.

We then compute $\var_{{\maP_0}}\qth{\maT_{\maM}}$.
Since $\E_{\maP_0}\qth{\maT_{\maM}}=0$, it is equivalent to computing $\E_{\maP_0}\qth{\maT_{\maM}^2}$. We note that \begin{align}
    \nonumber \E_{\maP_0}\qth{\maT_{\maM}^2} &= \sum_{\sfM_1,\sfM_2\in \maM} \omega_{\sfM_1}\omega_{\sfM_2} \E_{\maP_0}\qth{\inj(\sfM_1,\bar{G}_1) \inj(\sfM_1,\bar{G}_2) \inj(\sfM_2,\bar{G}_1)\inj(\sfM_2,\bar{G}_2)}\\\label{eq:varQ-HH'}
    &=\sum_{\sfM_1,\sfM_2\in \maM} \omega_{\sfM_1}\omega_{\sfM_2} \E^2_{\maP_0}\qth{\inj(\sfM_1,\bar{G}_1)\inj(\sfM_2,\bar{G}_1)} ,
\end{align}
where the last equality is because $\bar{G}_1$ and $\bar{G}_2$ are i.i.d. under ${\maP_0}$.
For any motifs $\sfM_1,\sfM_2$ with  $E(\sfM_1)\neq E(\sfM_2)$, we have $E(\sfM_1)\triangle E(\sfM_2)\neq \emptyset$. Consequently,
\begin{align*}
    &~\E_{\maP_0}\qth{\prod_{e\in E(\sfM_1)} \beta_e(\bar{G})\prod_{e\in E(\sfM_2)} \beta_e(\bar{G})}\\ =&~ \E_{\maP_0}\qth{\prod_{e\in E(\sfM_1)\triangle E(\sfM_2)} \beta_e(\bar{G})} \E_{\maP_0}\qth{\prod_{e\in E(\sfM_1)\cap E(\sfM_2)} \beta_e^2(\bar{G})} = 0,
\end{align*}
where the last equality is because $\E_{\maP_0}\qth{\prod_{e\in E(\sfM_1)\triangle E(\sfM_2)} \beta_e(\bar{G})} = 0$.
For any $\sfM_1\neq \sfM_2\in \maM$ and injective mappings $\vp_i:V(\sfM_i)\mapsto V(\bar{G}_1)$ with $i\in \sth{1,2}$, we note that $\vp_1(E(\sfM_1))\neq \vp_2(E(\sfM_2))$.
Therefore, for any $\sfM_1\neq \sfM_2\in \maM$, 
\begin{align*}
    &~\E_{\maP_0}\qth{\inj(\sfM_1,\bar{G}_1)\inj(\sfM_2,\bar{G}_1)} \\=&~ \sum_{\substack{\varphi_1:V(\sfM_1)\mapsto V(\bar{G}_1)\\\varphi_1\text{ injective}}}\sum_{\substack{\varphi_2:V(\sfM_2)\mapsto V(\bar{G}_1)\\\varphi_2\text{ injective}}} \E_{\maP_0}\qth{\prod_{e\in E(\sfM_1)} \beta_{\vp_1(e)}(\bar{G}_1)\prod_{e\in E(\sfM_2)}\beta_{\vp_2(e)}(\bar{G}_1)}\\=&~ \sum_{\substack{\varphi_1:V(\sfM_1)\mapsto V(\bar{G}_1)\\\varphi_1\text{ injective}}}\sum_{\substack{\varphi_2:V(\sfM_2)\mapsto V(\bar{G}_1)\\\varphi_2\text{ injective}}}\\&~~~~~~ \E_{\maP_0}\qth{\prod_{e\in \pth{\vp_1(E(\sfM_1))\triangle \vp_2(E(\sfM_2))}}\beta_e(\bar{G}_1)\prod_{e\in \pth{\vp_1(E(\sfM_1))\cap \vp_2(E(\sfM_2))}}\beta_e^2(\bar{G}_1)} =0
\end{align*}

For any $\sfM_1 = \sfM_2\in \maM$, \begin{align*}
    &~\E_{\maP_0}\qth{\inj(\sfM_1,\bar{G}_1)\inj(\sfM_2,\bar{G}_1)} \\=&~ \sum_{\substack{\varphi_1:V(\sfM_1)\mapsto V(\bar{G}_1)\\\varphi_1\text{ injective}}}\sum_{\substack{\varphi_2:V(\sfM_2)\mapsto V(\bar{G}_1)\\\varphi_2\text{ injective}}} \E_{\maP_0}\qth{\prod_{e\in E(\sfM_1)} \beta_{\vp_1(e)}(\bar{G}_1)\prod_{e\in E(\sfM_2)}\beta_{\vp_2(e)}(\bar{G}_1)}\\
    \overset{\mathrm{(a)}}{=} &~\sum_{\substack{\varphi_1:V(\sfM_1)\mapsto V(\bar{G}_1)\\\varphi_1\text{ injective}}}\sum_{\substack{\varphi_2:V(\sfM_2)\mapsto V(\bar{G}_1)\\\varphi_2\text{ injective}}} (p(1-p))^{\en(\sfM_1)}\indc{\vp_1(E(\sfM_1)) = \vp_2(E(\sfM_2))} \\\overset{\mathrm{(b)}}{=}&~\sum_{\substack{\varphi_1:V(\sfM_1)\mapsto V(\bar{G}_1)\\\varphi_1\text{ injective}}} (p(1-p))^{\en(\sfM_1)} \aut(\sfM_1)=\frac{(p(1-p))^{\en(\sfM_1)} \aut(\sfM_1) n!}{(n-\vn(\sfM_1))!},
\end{align*}
where $\mathrm(a)$ follows from $\E_{\maP_0}\qth{{\prod_{e\in E(\sfM_1)} \beta_{\vp_1(e)}(\bar{G}_1)\prod_{e\in E(\sfM_2)}\beta_{\vp_2(e)}(\bar{G}_1)}}=0$ for any $\vp_1(E(\sfM_1))\neq \vp_2(E(\sfM_2))$; $\mathrm{(b)}$ is because there are $\aut(\sfM_1)$ injective mappings for $\vp_1(E(\sfM_1)) = \vp_2(E(\sfM_2))$ given $\vp_1$.
Combining this with~\eqref{eq:varQ-HH'}, we obtain that \begin{align*}
    \E_{\maP_0}\qth{\maT_{\maM}^2} = \sum_{\sfM\in\maM}\pth{\frac{\omega_\sfM(p(1-p))^{\en(\sfM)} \aut(\sfM) n!}{(n-\vn(\sfM))!}}^2.
\end{align*}
Recall~\eqref{eq:EPmaT}. By picking $\omega_\sfM = \frac{\rho^{\en(\sfM)} (n-\vn(\sfM))!}{(p(1-p))^{\en(\sfM)}\aut(\sfM) n!}$, 
we have \begin{align}\label{eq:EPvarQ}
    \E_{\maP_1}\qth{\maT_{\maM}} = \var_{\maP_0}\qth{\maT_{\maM}} = \sum_{\sfM\in \maM} \rho^{2\en(\sfM)}.
\end{align}
Combining this with~\eqref{eq:typeI-chebyshev}, we obtain that \begin{align*}
    {\maP_0}\pth{\maT_{\maM}\ge \tau}\le \frac{4}{\sum_{\sfM\in \maM} \rho^{2\en(\sfM)}}.
\end{align*}

\subsection{Proof of Proposition~\ref{prop:admissible-TypeII}}\label{apd:proof-prop-admi-typeII}

By Chebyshev's inequality, the Type II error is controlled by \begin{align*}
    \maP_1\pth{\maT_{\maM}<\tau}&\le \maP_1\pth{\pth{\maT_{\maM}-\E_{\maP_1}\qth{\maT_{\maM}}}^2 >\frac{\pth{\E_{\maP_1}\qth{\maT_{\maM}}}^2}{4}}\le \frac{4\var_{\maP_1}\qth{\maT_{\maM}}}{\pth{\E_{\maP_1}\qth{\maT_{\maM}}}^2}.
\end{align*}
By selecting the weight $\omega_\sfM = \frac{\rho^{\en(\sfM)} (n-\vn(\sfM))! }{n! (p(1-p))^{\en(\sfM)}\aut(\sfM)}$, we have shown in~\eqref{eq:EPvarQ} that  $\E_{\maP_1}\qth{\maT_{\maM}}$ is characterized by the \emph{signal score} $\sum_{\sfM\in \maM}\rho^{2\en(\sfM)}$. It remains to estimate the second moment $\E_{\maP_1}\qth{\maT_{\maM}^2}$. 

Given two bounded degree motifs $\sfM_1$ and $\sfM_2$, we define a homomorphism matrix $\vp\triangleq \begin{bmatrix}\varphi_{11},\varphi_{12}\\\varphi_{21},\varphi_{22}\end{bmatrix}$, where $\vp_{ij}:V(\sfM_i)\mapsto V(\bar{G}_j)$. Let $\Phi$ be the set of homomorphism matrices $\vp$ such that $\vp_{ij}$ are injective for any $1\le i,j\le 2$. Recall that the motif counting statistic defined in~\eqref{eq:est-sub-count}.
The second moment under $\maP_1$ is given by \begin{align}\label{eq:second-moment-P}
   \E_{\maP_1}\qth{\maT_{\maM}^2}= \sum_{\sfM_1,\sfM_2\in \maM} \omega_{\sfM_1}\omega_{\sfM_2} \sum_{\vp\in \Phi}\E_{\maP_1}\qth{\prod_{i,j}\hom_{\vp_{ij}}(\sfM_i,\bar{G}_j)}.
\end{align}
Given a homomorphism matrix $\vp\in \Phi$, we define the motif $\sfH_{ij}$ induced by $\vp_{ij}$ as\begin{align}\label{eq:Hij}
    V(\sfH_{ij})\triangleq \sth{\vp_{ij}(u):u\in V(\sfM_i)},\quad E(\sfH_{ij}) = \sth{\vp_{ij}(u)\vp_{ij}(v):uv\in E(\sfM_i)}.
\end{align}
The number of node overlap on the graph $\bar{G}_j$ is defined as $\sfn_j\triangleq |V(\sfH_{1j})\cap V(\sfH_{2j})|$. The injective homomorphism matrix can be partitioned into three types according to the node overlap size:
\begin{itemize}
\item \emph{Discrepant overlap:} $\Phi_D\triangleq \sth{\vp\in \Phi:\sfn_1\notin [\sfn_2/2,2\sfn_2]}$.
    \item \emph{Balanced overlap:} $\Phi_B\triangleq \sth{\vp\in \Phi:\sfn_1\in [\sfn_2/2,2\sfn_2],\sfn_2>0}$.
    \item \emph{Null overlap:} $\Phi_N\triangleq \sth{\vp\in \Phi:\sfn_1 = \sfn_2 = 0}$.
\end{itemize}

By~\eqref{eq:second-moment-P}, a key quantity for the second moment is $\E_{\maP_1}\qth{\prod_{i,j}\hom_{\vp_{ij}}(\sfM_i,\bar{G}_j)}$, which will be characterized by different node overlap types by the following Lemma.

\begin{lemma}\label{lem:varphi-bound}
Assume $\maM$ is \emph{$C$-admissible} with constant $\epsilon_0$ for Condition~\ref{cond:subgraph} in Definition~\ref{def:admissible}. For any $\vp\in \Phi$, let $F(\vp)\triangleq \E_{\maP_1}\qth{\prod_{i,j=1}^2\frac{\hom_{\vp_{ij}}(\sfM_i,\bar{G}_j)}{\sqrt{(p(1-p))^{\en(\sfM_i)}}}}$.
\begin{itemize}
    \item If $\vp\in \Phi_D$, then   
    \begin{align}\label{eq:lem-vp-term3}
    F(\vp) = 0;
    \end{align}
    \item If $\vp\in \Phi_B$, then 
    \begin{align}\label{eq:lem-vp-term2}
     F(\vp)\le&~ \pth{\frac{2C}{n}}^{\vn(\sfM_1)+\vn(\sfM_2)-\sfn_1-\sfn_2}\qth{\indc{\sfH_{11}=\sfH_{21},\sfH_{12} = \sfH_{22}}+3n^{-\epsilon_0/2}\pth{4C}^{2C}};
    \end{align}
    \item  If $\vp\in \Phi_N$, then 
\begin{align}\label{eq:lem-vp-term1}
    F(\vp)
= \frac{1+\indc{\sfM_1= \sfM_2}}{n!}(n-\vn(\sfM_1)-\vn(\sfM_2))!\aut(\sfM_1)\aut(\sfM_2)\rho^{\en(\sfM_1)+\en(\sfM_2)}.
    \end{align}
\end{itemize}
\end{lemma}
The proof of Lemma~\ref{lem:varphi-bound} is deferred to Appendix~\ref{apd:proof-lemma-varphi-bound}. By~\eqref{eq:lem-vp-term3} in Lemma~\ref{lem:varphi-bound}, if suffices to consider $\vp \in \Phi_B$ and $\vp\in \Phi_N$.

\textbf{Balanced overlap: $\vp\in \Phi_B$.\ \ }For any $\vp\in\Phi_B$, we have $\sfn_1\in [\sfn_2/2,2\sfn_2]$ and $\sfn_2>0$,
where $\sfn_j = |V(\sfH_{1j})\cap V(\sfH_{2j})|$ and the motif $\sfH_{ij}$ induced by $\vp_{ij}$ defined in~\eqref{eq:Hij}.
We note that \begin{align}
    \nonumber \sum_{\vp_{11},\vp_{21}}\indc{\sfn_1=i} &\overset{\mathrm{(a)}}{=} \frac{n!}{(n-\vn(\sfM_1))!}\binom{\vn(\sfM_1)}{i}\binom{\vn(\sfM_2)}{i} \frac{i!(n-\vn(\sfM_1))!}{(n-\vn(\sfM_1)-\vn(\sfM_2)+i)!}\\\label{eq:sum-indc-V1V2}&\overset{\mathrm{(b)}}{\le} \frac{n!}{(n-\vn(\sfM_1))!} \vn(\sfM_1)^{2i} n^{\vn(\sfM_1)-i},
\end{align}
where $\mathrm{(a)}$ is because there are $\frac{n!}{(n-\vn(\sfM_1))!}$ choices for $\vp_{11}$, and when given $\vp_{11}$, there are $\binom{\vn(\sfM_1)}{i}\binom{\vn(\sfM_2)}{i}i!$ choices for mapping $i$ vertices from $V(\sfM_2)$ to $V(\sfH_{11})$ and $\frac{(n-\vn(\sfM_1))!}{(n-\vn(\sfM_1)-\vn(\sfM_2)+i)!}$ choices for mapping the remaining $\vn(\sfM_2)-i$ vertices to $V(\bar{G}_1)\backslash V(\sfH_{11})$; $\mathrm{(b)}$ applies $\vn(\sfM_1) = \vn(\sfM_2)$, $\binom{\vn(\sfM_1)}{i}\binom{\vn(\sfM_2)}{i}i!\le (\vn(\sfM_1))^{2i}$ and $\frac{(n-\vn(\sfM_1))!}{(n-\vn(\sfM_1)-\vn(\sfM_2)+i)!}\le n^{\vn(\sfM_1)-i}$.
Similarly, \begin{align}\label{eq:sum-indc-V3V4}
    \sum_{\vp_{12},\vp_{22}}\indc{\sfn_2=i}\le \frac{n!}{(n-\vn(\sfM_1))!} \vn(\sfM_1)^{2i} n^{\vn(\sfM_1)-i}.
\end{align}

Let  \begin{align*}
    f(\sfn_1,\sfn_2)\triangleq 3n^{-\epsilon_0/2}(p(1-p))^{\en(\sfM_1)+\en(\sfM_2)} \pth{\frac{2C}{n}}^{\vn(\sfM_1)+\vn(\sfM_2)-\sfn_1-\sfn_2} \pth{4C}^{2C}.
\end{align*}
By~\eqref{eq:lem-vp-term2} in  Lemma~\ref{lem:varphi-bound}, since $\vn(\sfM_1) = \sfn_1$ and $\vn(\sfM_2) = \sfn_2$ when $\sfH_{11} = \sfH_{21}$ and $\sfH_{12} = \sfH_{22}$, we have \begin{align}
    \nonumber &~\omega_{\sfM_1}\omega_{\sfM_2}\sum_{\vp\in \Phi_B}  \E_{\maP_1}\qth{\prod_{i,j}\hom_{\vp_{ij}}(\sfM_i,\bar{G}_j)}\\\nonumber =&~\omega_{\sfM_1}\omega_{\sfM_2}\sum_{\vp\in \Phi} \indc{0<\frac{1}{2}\sfn_2\le \sfn_1\le 2\sfn_2} \E_{\maP_1}\qth{\prod_{i,j}\hom_{\vp_{ij}}(\sfM_i,\bar{G}_j)}\\\label{eq:prop2-sum-1}\le&~\omega_{\sfM_1}\omega_{\sfM_2}  \sum_{\vp\in \Phi} \indc{0<\frac{1}{2}\sfn_2\le \sfn_1\le 2\sfn_2} f(\sfn_1,\sfn_2)\\\label{eq:prop2-sum-2}
    &+\omega_{\sfM_1}\omega_{\sfM_2} (p(1-p))^{\en(\sfM_1)+\en(\sfM_2)} \sum_{\vp\in \Phi} \indc{\sfH_{11} = \sfH_{21}} \indc{\sfH_{12} = \sfH_{22}} .
\end{align}
 We note that \begin{align*}
    \sum_{\vp\in \Phi} \indc{\sfH_{11}= \sfH_{21}}\indc{\sfH_{12} = \sfH_{22}} &= \sum_{\vp_{11},\vp_{21}}\indc{\sfH_{11} = \sfH_{21}} \sum_{\vp_{12},\vp_{22}}  \indc{\sfH_{12} = \sfH_{22}}\\ &= \pth{\frac{n! \aut(\sfM_1)}{(n-\vn(\sfM_1))!}}^2 \indc{\sfM_1=\sfM_2},
\end{align*}
where the last equality holds because $\sfH_{11} = \sfH_{21}$ implies $\sfM_1 = \sfM_2$, yielding $\frac{n!\aut(\sfM_1)}{(n-\vn(\sfM_1))!}$ choices for $\vp_{11}$ and $\vp_{21}$, and similarly $\sfH_{12} = \sfH_{22}$ implies $\sfM_1 = \sfM_2$, also yielding $\frac{n! \aut(\sfM_1)}{(n-\vn(\sfM_1))!}$ choices. 
Recall that $\omega_{\sfM_1}=\frac{\rho^{\en(\sfM_1)} (n-\vn(\sfM_1))!}{n! (p(1-p))^{\en(\sfM_1)}\aut(\sfM_1)}$ and $\omega_{\sfM_2}=\frac{\rho^{\en(\sfM_2)} (n-\vn(\sfM_2))!}{n! (p(1-p))^{\en(\sfM_2)}\aut(\sfM_2)}$.
Therefore,\begin{align}\label{eq:prop2-sum-3}
    \omega_{\sfM_1}\omega_{\sfM_2} (p(1-p))^{\en(\sfM_1)+\en(\sfM_2)} \sum_{\vp\in \Phi} \indc{\sfH_{11} = \sfH_{21}} \indc{\sfH_{12} = \sfH_{22}} = \rho^{2\en(\sfM_1)}\indc{\sfM_1 = \sfM_2}.
\end{align}
We then bound the term~\eqref{eq:prop2-sum-1}: \begin{align}
    \nonumber &~\omega_{\sfM_1}\omega_{\sfM_2}\sum_{\vp\in \Phi} \indc{0<\frac{1}{2}\sfn_2\le \sfn_1\le 2\sfn_2} f(\sfn_1,\sfn_2)\\ \nonumber=&~ \omega_{\sfM_1}\omega_{\sfM_2}\sum_{\vp\in \Phi}\sum_{i=1}^{\vn(\sfM_1)} \sum_{j=i/2}^{\min\sth{2i,\vn(\sfM_2)}} \indc{\sfn_1=i}\indc{\sfn_2=j} f(i,j)\\ \nonumber
    =&~\omega_{\sfM_1}\omega_{\sfM_2}\sum_{i=1}^{\vn(\sfM_1)} \sum_{j=i/2}^{\min\sth{2i,\vn(\sfM_2)}} f(i,j)\sum_{\vp\in \Phi}\indc{\sfn_1=i}\indc{\sfn_2=j}\\ \nonumber
    \overset{\mathrm{(a)}}{\le}&~\omega_{\sfM_1}\omega_{\sfM_2}\sum_{i=1}^{\vn(\sfM_1)} \sum_{j=i/2}^{\min\sth{2i,\vn(\sfM_2)}} f(i,j) \frac{n!\vn(\sfM_1)^{2i}n^{\vn(\sfM_1)-i}}{(n-\vn(\sfM_1))!}\frac{n!\vn(\sfM_2)^{2j}n^{\vn(\sfM_2)-j}}{(n-\vn(\sfM_2))!}\\\label{eq:prop2-sum-4}=&~\rho^{\en(\sfM_1)+\en(\sfM_2)}\pth{ 3n^{-\tfrac{\epsilon_0}{2}}\sum_{i=1}^{\vn(\sfM_1)}\sum_{j=i/2}^{\min\sth{2i,\vn(\sfM_2)}} \pth{2C}^{\vn(\sfM_1)+\vn(\sfM_2)-i-j}  \pth{4C}^{2C} \frac{\vn(\sfM_1)^{2i}\vn(\sfM_2)^{2j}}{\aut(\sfM_1)\aut(\sfM_2)}},
\end{align}
where $\mathrm{(a)}$ follows from~\eqref{eq:sum-indc-V1V2} and~\eqref{eq:sum-indc-V3V4}.
We note that \begin{align*}
    &~\sum_{i=1}^{\vn(\sfM_1)}\sum_{j=i/2}^{\min\sth{2i,\vn(\sfM_2)}} \pth{2C}^{\vn(\sfM_1)+\vn(\sfM_2)-i-j}  \pth{4C}^{2C} \frac{\vn(\sfM_1)^{2i}\vn(\sfM_2)^{2j}}{\aut(\sfM_1)\aut(\sfM_2)} \\
    \overset{\mathrm{(a)}}{\le}&~\sum_{i=1}^{\vn(\sfM_1)}\sum_{j=i/2}^{\min\sth{2i,\vn(\sfM_2)}}\pth{4C}^{\vn(\sfM_1)+\vn(\sfM_2)-i-j}\pth{4C}^{2C} \pth{4C}^{2i+2j}\\
    \overset{\mathrm{(b)}}{\le}&~\vn(\sfM_1)\vn(\sfM_2)\pth{4C}^{\vn(\sfM_1)+\vn(\sfM_2)+2C+\vn(\sfM_1)+\vn(\sfM_2)}\\
    \overset{\mathrm{(c)}}{\le}&~  \pth{4C}^{2\vn(\sfM_1)+2\vn(\sfM_2)+2C+2}\overset{\mathrm{(d)}}{\le} (4C)^{8C},
\end{align*}
where $\mathrm{(a)}$ is because  $\aut(\sfM_1)\ge 1$,  $\aut(\sfM_2)\ge 1$, and $\vn(\sfM_1),\vn(\sfM_2)\le C\le 4C$; $\mathrm{(b)}$ is because $i+j\le \vn(\sfM_1)+\vn(\sfM_2)$; $\mathrm{(c)}$ is because $\vn(\sfM_1)\vn(\sfM_2) \le C^2\le \pth{4C}^2$; $\mathrm{(d)}$ follows from $2\vn(\sfM_1)+2\vn(\sfM_2)+2C+2\le 8C$. Combining this with~\eqref{eq:prop2-sum-1}, \eqref{eq:prop2-sum-2}, \eqref{eq:prop2-sum-3}, and \eqref{eq:prop2-sum-4}, we obtain \begin{align}
    \nonumber &~\sum_{\sfM_1,\sfM_2\in \maM}\omega_{\sfM_1}\omega_{\sfM_2}\sum_{\vp\in \Phi_B} \E_{\maP_1}\qth{\prod_{i,j}\hom_{\vp_{ij}}(\sfM_i,\bar{G}_j)}\\\nonumber =&~\sum_{\sfM_1,\sfM_2\in \maM}\omega_{\sfM_1}\omega_{\sfM_2}\sum_{\vp\in \Phi} \indc{0<\frac{1}{2}\sfn_2\le \sfn_1\le 2\sfn_2} \E_{\maP_1}\qth{\prod_{i,j}\hom_{\vp_{ij}}(\sfM_i,\bar{G}_j)}\\\nonumber \le&~\sum_{\sfM_1,\sfM_2\in \maM} \rho^{\en(\sfM_1)+\en(\sfM_2)}\pth{3n^{-\epsilon_0/2} (4C)^{8C}+\indc{\sfM_1 = \sfM_2}} \\\label{eq:term2}=&~ 3n^{-\epsilon_0/2} (4C)^{8C} \pth{\sum_{\sfM\in \maM} \rho^{\en(\sfM)}}^2+\sum_{\sfM\in\maM} \rho^{2\en(\sfM)}. 
\end{align}

\textbf{Null overlap: $\vp\in \Phi_N$.\ \ }For any $\vp \in \Phi_N$, we have $\sfn_1 = \sfn_2 = 0$, where $\sfn_j = |V(\sfH_{1j})\cap V(\sfH_{2j})|$ and the motif $\sfH_{ij}$ induced by $\vp_{ij}$ defined in~\eqref{eq:Hij}.
Recall~\eqref{eq:Epvarphi12}. For $i=1,2$, 
\begin{align*}
    \E_{\maP_1}\qth{\prod_{j}\hom_{\vp_{ij}}(\sfM_i,\bar{G}_j)} &= \pth{\rho p (1-p)}^{\en(\sfM_i)} \frac{\aut(\sfM_i) (n-\vn(\sfM_i))!}{n!}.
\end{align*}
Combining this with~\eqref{eq:lem-vp-term3} in Lemma~\ref{lem:varphi-bound}, when $\sfn_1 = \sfn_2 = 0$, we have \begin{align}
    \nonumber &~\frac{\E_{\maP_1}\qth{\prod_{i,j}\hom_{\vp_{ij}}(\sfM_i,\bar{G}_j)}}{\prod_{i}\E_{\maP_1}\qth{\prod_{j}\hom_{\vp_{ij}}(\sfM_i,\bar{G}_j)} }\\\nonumber=&~ (1+\indc{\sfM_1 = \sfM_2}) \frac{n!(n-\vn(\sfM_1)-\vn(\sfM_2))! }{(n-\vn(\sfM_1))!(n-\vn(\sfM_2))!}\\\nonumber =&~(1+\indc{\sfM_1 = \sfM_2}) \prod_{i=n-\vn(\sfM_1)-\vn(\sfM_2)+1}^{n-\vn(\sfM_1)} \frac{i+\vn(\sfM_1)}{i}\\\nonumber 
    \le&~(1+\indc{\sfM_1 = \sfM_2})\prod_{i=n-\vn(\sfM_1)-\vn(\sfM_2)+1}^{n-\vn(\sfM_1)} \exp\pth{ \frac{\vn(\sfM_1)}{n-\vn(\sfM_1)-\vn(\sfM_2)+1}}\\\nonumber =&~(1+\indc{\sfM_1 = \sfM_2})\exp\pth{\frac{\vn(\sfM_1)\vn(\sfM_2)}{n-\vn(\sfM_1)-\vn(\sfM_2)+1}},
\end{align} 
where the inequality follows from $i\ge n-\vn(\sfM_1)-\vn(\sfM_2)+1$ and $1+x\le \exp(x)$ for any $x\ge 0$. Let $\kappa =\exp\pth{\frac{C^2}{n-2C+1}}$. 
For any $\sfM_1,\sfM_2\in \maM$, since $\vn(\sfM_1), \vn(\sfM_2)\le C$ for any $\sfM\in \maM$, we have 
\begin{align}\label{eq:prop2-pf-term2-1}
    \frac{\E_{\maP_1}\qth{\prod_{i,j}\hom_{\vp_{ij}}(\sfM_i,\bar{G}_j)}}{\prod_{i}\E_{\maP_1}\qth{\prod_{j}\hom_{\vp_{ij}}(\sfM_i,\bar{G}_j)} }\le \kappa (1+\indc{\sfM_1=\sfM_2}).
\end{align}
Recall that $ \E_{\maP_1}\qth{\maT_{\maM}} = \sum_{\sfM\in \maM}\rho^{2\en(\sfM)}$.
Then, \begin{align}
    \nonumber &~\sum_{\sfM_1,\sfM_2\in \maM}\omega_{\sfM_1}\omega_{\sfM_2}\sum_{\vp\in \Phi_N} \E_{\maP_1}\qth{\prod_{i,j}\hom_{\vp_{ij}}(\sfM_i,\bar{G}_j)} \\\nonumber =&~\sum_{\sfM_1,\sfM_2\in \maM}\omega_{\sfM_1}\omega_{\sfM_2}\sum_{\vp\in \Phi} \indc{\sfn_1 = \sfn_2 = 0} \E_{\maP_1}\qth{\prod_{i,j}\hom_{\vp_{ij}}(\sfM_i,\bar{G}_j)} \\\nonumber\overset{\mathrm{(a)}}{\le}&~\sum_{\sfM_1,\sfM_2\in \maM}\omega_{\sfM_1}\omega_{\sfM_2}\sum_{\vp\in \Phi} \kappa (1+\indc{\sfM_1 = \sfM_2})  \prod_{i=1}^2\E_{\maP_1}\qth{\prod_{j}\hom_{\vp_{ij}}(\sfM_i,\bar{G}_j)}\\\nonumber=&~\kappa \pth{\E_{\maP_1}\qth{\maT_{\maM}}}^2+\kappa \sum_{\sfM_1\in \maM} \omega_{\sfM_1}^2 \sum_{\vp\in \Phi} \pth{\E_{\maP_1}\qth{\prod_{j}\hom_{\vp_{1j}}(\sfM_1,\bar{G}_j)}}^2\\\label{eq:term3}=&~\kappa \pth{\sum_{\sfM\in \maM} \rho^{2\en(\sfM)}}^2+\kappa \sum_{\sfM\in \maM} \rho^{4\en(\sfM)},
\end{align}
where $\mathrm{(a)}$ follows from~\eqref{eq:prop2-pf-term2-1}; the last equality is because  $\E_{\maP_1}\qth{\prod_{j}\hom_{\vp_{1j}}(\sfM_1,\bar{G}_j)} = (\rho p (1-p))^{\en(\sfM_1)} \frac{\aut(\sfM_1) (n-\vn(\sfM_1))!}{n!}$, $|\Phi|=\pth{\frac{n!}{(n-\vn(\sfM_1))!}}^4$ and $\omega_{\sfM_1} = \frac{\rho^{\en(\sfM_1)}(n-\vn(\sfM_1))!\aut(\sfM_1)}{n!(p(1-p))^{\en(\sfM_1)}}$.
By Lemma~\ref{lem:varphi-bound}, \begin{align*}
    \sum_{\sfM_1,\sfM_2\in \maM}\omega_{\sfM_1}\omega_{\sfM_2}\sum_{\vp\in \Phi_D}  \E_{\maP_1}\qth{\prod_{i,j}\hom_{\vp_{ij}}(\sfM_i,\bar{G}_j)}=0.
\end{align*}
Combining this with~\eqref{eq:term2} and~\eqref{eq:term3}, and noting that $\Phi$ decomposes as the disjoint union
 $\Phi = \Phi_D\cup \Phi_B\cup \Phi_N$,
we obtain  \begin{align*}
    \E_{\maP_1}\qth{\maT_{\maM}^2} =&  \sum_{\sfM_1,\sfM_2\in \maM}\omega_{\sfM_1}\omega_{\sfM_2}\sum_{\vp\in \Phi_D} \E_{\maP_1}\qth{\prod_{i,j}\hom_{\vp_{ij}}(\sfM_i,\bar{G}_j)}\\&\quad+\sum_{\sfM_1,\sfM_2\in \maM}\omega_{\sfM_1}\omega_{\sfM_2}\sum_{\vp\in \Phi_B}  \E_{\maP_1}\qth{\prod_{i,j}\hom_{\vp_{ij}}(\sfM_i,\bar{G}_j)}\\&\quad+\sum_{\sfM_1,\sfM_2\in \maM}\omega_{\sfM_1}\omega_{\sfM_2}\sum_{\vp\in \Phi_N} \E_{\maP_1}\qth{\prod_{i,j}\hom_{\vp_{ij}}(\sfM_i,\bar{G}_j)}\\
    \le&~  3n^{-\epsilon_0/2} (4C)^{8C} \pth{\sum_{\sfM\in \maM} \rho^{\en(\sfM)}}^2+\sum_{\sfM\in\maM} \rho^{2\en(\sfM)}\\&\quad+\kappa \pth{\sum_{\sfM\in \maM} \rho^{2\en(\sfM)}}^2+\kappa \sum_{\sfM\in \maM} \rho^{4\en(\sfM)}.
\end{align*}
Therefore, we conclude that \begin{align*}
    &~\maP_1\pth{\maT_{\maM} <\tau }\le \frac{4\var_{\maP_1}\qth{\maT_{\maM}}}{\pth{\E_{\maP_1}\qth{\maT_{\maM}}}^2} \\=&~ \frac{4\pth{\E_{\maP_1}\qth{\maT_{\maM}^2}-\pth{\E_{\maP_1}\qth{\maT_{\maM}}}^2}}{\pth{\E_{\maP_1}\qth{\maT_{\maM}}}^2} \\
    \le&~ \frac{4\pth{3n^{-\frac{\epsilon_0}{2}} (4C)^{8C} \pth{\sum_{\sfM\in \maM} \rho^{\en(\sfM)}}^2+\sum_{\sfM\in\maM} \rho^{2\en(\sfM)}}}{\pth{\sum_{\sfM\in \maM} \rho^{2\en(\sfM)}}^2}\\&+\frac{4\pth{(\kappa-1) \pth{\sum_{\sfM\in \maM} \rho^{2\en(\sfM)}}^2+\kappa \sum_{\sfM\in \maM} \rho^{4\en(\sfM)}}}{\pth{\sum_{\sfM\in \maM} \rho^{2\en(\sfM)}}^2} \\
    \overset{\mathrm{(a)}}{\le}&~  4\Bigg(3n^{-\frac{\epsilon_0}{2}} (4C)^{8C} \Big(\frac{\sum_{\sfM\in \maM} \rho^{\en(\sfM)}}{\sum_{\sfM\in \maM} \rho^{2\en(\sfM)}}\Big)^2+\frac{\exp\big(\frac{C^2}{n-2C+1}\big)+1}{\sum_{\sfM\in \maM} \rho^{2\en(\sfM)}}+\exp\Big(\frac{C^2}{n-2C+1}\Big)-1\Bigg)\\
    \overset{\mathrm{(b)}}{\le}&~4\Bigg(3n^{-\frac{\epsilon_0}{2}} (4C)^{8C} \rho^{-2C}+\frac{\exp\big(\frac{C^2}{n-2C+1}\big)+1}{\sum_{\sfM\in \maM} \rho^{2\en(\sfM)}}+\exp\Big(\frac{C^2}{n-2C+1}\Big)-1\Bigg),
\end{align*}
where $\mathrm{(a)}$ follows from $\sum_{\sfM\in \maM} \rho^{4\en(\sfM)}\le \sum_{\sfM\in \maM} \rho^{2\en(\sfM)}$; $\mathrm{(b)}$ is because $\frac{\rho^{\en(\sfM)}}{\rho^{2\en(\sfM)}}\le \rho^{-C}$ for all $\sfM\in \maM$ implies $\pth{\frac{\sum_{\sfM\in \maM} \rho^{\en(\sfM)}}{\sum_{\sfM\in \maM} \rho^{2\en(\sfM)}}}^2\le \rho^{-2C}$.

\section{Proof of Lemmas}

\subsection{Proof of Lemma~\ref{lem:motif-family-lwbd}}\label{apd:proof-lem-motif-family-lwbd}

We first upper bound the automorphism numbers for any $\sfM\in \maM(N_\sfv,N_\sfe,d)$. Given any $\sfM\in \maM(N_\sfv,N_\sfe,d)$, let $\maX(\sfM)$ be the automorphism group of $\sfM$:
\[
\maX(\sfM)\triangleq \sth{\varphi\text{ bijection}:V(\sfM)\mapsto V(\sfM):uv\in E(\sfM)\iff \varphi(u)\varphi(v)\in E(\sfM)}.
\]


Let $V_0(\sfM)\triangleq \sth{v_{0,0},v_{0,1},v_{0,2},v_{0,3}}$.
For $b\in\{1,2\}$, define $\mathcal{Y}_b(\sfM)$ to be the set of bijective mapping $\psi:V(\sfM)\backslash V_0(\sfM)\mapsto V(\sfM)\backslash V_0(\sfM)$  such that for every $1\le i\le d-1$,
\[
\psi(v_{i,1})\in \maN(v_{0,b})\quad\text{and}\quad \psi(v_{i,j})\in \maN(\psi(v_{i,j-1}))\ \text{for all }2\le j\le \ell,
\]
where $\maN(v)$ denotes the neighbor set of $v$ in $\sfM$. Set $\mathcal{Y}(\sfM)\triangleq \mathcal{Y}_1(\sfM)\cup\mathcal{Y}_2(\sfM)$.
We note that $\deg(v)\le d$ for all $v\in V(\sfM)$, where $\deg(v)$ denotes the degree of $v$.
We conclude that each path contributes at most $d^\ell$ possibilities, hence
\[
|\mathcal{Y}_1(\sfM)|, |\mathcal{Y}_2(\sfM)|\le d^{(d-1)\ell}\quad\text{and}\quad |\mathcal{Y}(\sfM)|\le 2d^{(d-1)\ell}.
\]
Define the restriction map
\[
\Pi:\ \mathcal{X}(\sfM)\longrightarrow \mathcal{Y}(\sfM)=\mathcal{Y}_1(\sfM)\cup\mathcal{Y}_2(\sfM),\qquad
\Pi(\varphi):=\varphi\!\restriction_{V(M)\backslash V_0(M)}.
\]
This map is well-defined because any automorphism preserves adjacency, hence the image
of each path $P_i$ under $\varphi$ satisfies the constraints in the definition of $\mathcal{Y}_b(\sfM)$,
with $b\in\{1,2\}$ determined by whether $\varphi(v_{0,1})=v_{0,b}$.

We claim that $\Pi$ is injective.
Indeed, if $\Pi(\varphi_1)=\Pi(\varphi_2)$, then $\varphi_1$ and $\varphi_2$ agree on all path
vertices $V(\sfM)\backslash V_0(\sfM)$ and induce the same choice of $b$ for the central pair
$\{v_{0,1},v_{0,2}\}$. The remaining two special vertices $v_{0,0}$ and $v_{0,3}$ are uniquely
determined by adjacency (they are the extremity vertices). Hence $\varphi_1=\varphi_2$.
Therefore,
\[
\aut(\sfM)=|\mathcal{X}(\sfM)|\le |\mathcal{Y}(\sfM)|\le 2d^{(d-1)\ell}.
\]
In particular, every $\sfM\in\mathcal{M}(N_v,N_e,d)$ satisfies
\(
\aut(\sfM)\le 2d^{(d-1)\ell}.
\)

We then derive the lower bound for $|\maM(N_\sfv,N_\sfe,d)|$. For any $1\le i<j\le d-1$, there are $\ell$ edges between paths $P_i$ and $P_j$, with distinct vertices at each point. It is equivalent to picking a bijective mapping between two vertices sets $\sth{v_{i,1},\cdots, v_{i,\ell}}$ and $\sth{v_{j,1},\cdots, v_{j,\ell}}$, where we have $\ell!$ options. 
Hence, there are $(\ell !)^{\binom{d-1}{2}}$ labeled constructions in total. 
Let $S_{\ell}$ be the set of permutations from $[\ell]$ to $[\ell]$. Define $\mathcal{Z}\triangleq\prod_{1\le i<j\le d-1} S_{\ell}$. Then $|\maZ| =(\ell!)^{\binom{d-1}{2}}$.
Let $\widetilde{\mathcal{M}}$ be the set of labeled motifs on the fixed labels
\(
\{v_{i,t}: i=1,\dots,d-1,\ t=1,\dots,\ell\}\ \cup\ \{v_{0,0},v_{0,1},v_{0,2},v_{0,3}\}.
\)
For $z=(\pi_{i,j})_{i<j}\in\mathcal Z$, define $\Phi_{\mathcal Z}(z)=H(z)\in\widetilde{\mathcal M}$ on labels
$\{v_{i,t}: i=1,\dots,d-1,\ t=1,\dots,\ell\}\cup\{v_{0,0},v_{0,1},v_{0,2},v_{0,3}\}$ with
\[
E(H(z))=E_{\mathrm{core}}\cup E_{\mathrm{path}}\cup E_{\mathrm{att}}\cup E_{\mathrm{cross}}(z),
\]
where
\begin{align*}
    E_{\mathrm{core}}&=\{v_{0,0}v_{0,1},\ v_{0,2}v_{0,3}\},\\
E_{\mathrm{path}}&=\{v_{i,t}v_{i,t+1}: i=1,\dots,d-1,\ t=1,\dots,\ell-1\},\\
E_{\mathrm{att}}&=\{v_{0,b}v_{i,1}: i=1,\dots,d-1\}\ \text{with fixed }b\in\{1,2\},\\
E_{\mathrm{cross}}(z)&=\{v_{i,t}v_{j,\pi_{i,j}(t)}: 1\le i<j\le d-1,\ t\in[\ell]\}.
\end{align*}
If $\Phi_{\mathcal{Z}}(z_1)=\Phi_{\mathcal{Z}}(z_2)$, then for every pair $(i,j)$ with $1\le i<j\le d-1$, the induced
subgraphs on $\{v_{i,1},\dots,v_{i,\ell}\}\cup\{v_{j,1},\dots,v_{j,\ell}\}$ are  same, which uniquely recovers $\pi_{i,j}$. Hence $\Phi_{\mathcal{Z}}$ is injective and
$|\Phi_{\mathcal{Z}}(\mathcal{Z})|=|\mathcal{Z}|$.

For any $\sfM\in\mathcal{M}(N_\sfv,N_\sfe,d)$, let
$\mathsf{Num}(\sfM)$ be the number of labeled realizations of $\sfM$ inside $\Phi_{\mathcal{Z}}(\mathcal{Z})$.
Two such labelings differ by an automorphism of $\sfM$, so
\[
\mathsf{Num}(\sfM)\le\aut(\sfM)\le 2d^{(d-1)\ell}.
\]
Therefore,
\[
(\ell!)^{\binom{d-1}{2}}\le |\mathcal{Z}|=|\Phi_{\mathcal{Z}}(\mathcal{Z})|
=\sum_{\sfM\in\mathcal{M}(N_\sfv,N_\sfe,d)} \mathsf{Num}(\sfM)\le\ 2d^{(d-1)\ell}\,|\mathcal{M}(N_\sfv,N_\sfe,d)|.
\]
Consequently,
\begin{align*}  |\maM(N_\sfv,N_\sfe,d)|\ge \frac{1}{2d^{(d-1)\ell}} \pth{\ell!}^{\binom{d-1}{2}}&\overset{\mathrm{(a)}}{\ge}\frac{1}{2d^{(d-1)\ell}} \pth{\frac{\ell}{e}}^{\ell\binom{d-1}{2}}\\&\overset{\mathrm{(b)}}{=}\frac{1}{2}\pth{\frac{2(N_\sfe-d-1)}{ed^{d/(d-2)}(d-1)}}^{\frac{d-2}{d}\cdot (N_\sfe-d-1)},
\end{align*}
where $\mathrm{(a)}$ is because $\ell!\ge \pth{\frac{\ell}{e}}^{\ell}$ by Stirling's approximation; $\mathrm{(b)}$ follows from $\ell = \frac{N_\sfe-d-1}{\binom{d}{2}}$.

On the other hand, the upper bound of $|\maM(N_\sfv,N_\sfe,d)|$ can be directly derived by \begin{align*}
    |\maM(N_\sfv,N_\sfe,d)|\le (\ell!)^{\binom{d-1}{2}}\overset{\mathrm{(a)}}{\le} \ell^{\ell\binom{d-1}{2}}\overset{\mathrm{(b)}}{=}\pth{\frac{2(N_\sfe-d-1)}{d(d-1)}}^{\frac{d-2}{d}\cdot(N_\sfe-d-1)}, 
\end{align*}
where $\mathrm{(a)}$ follows from $\ell!\le \ell^\ell$ and $\mathrm{(b)}$ is because $\ell = \frac{N_\sfe-d-1}{\binom{d}{2}}$.

\subsection{Proof of Lemma~\ref{lem:varphi-bound}}\label{apd:proof-lemma-varphi-bound}

\textbf{Case 1: Discrepant overlap\ \ }We first consider $\vp \in \Phi_D$. Recall the motif $\sfH_{ij}$ induced by $\vp$ defined in~\eqref{eq:Hij} and $G\cap G',G\triangle G'$ defined in~\eqref{eq:intersec-graph} and~\eqref{eq:triangle-graph}.
    Let   \begin{align}\label{eq:def_of_T1T2}
        \sfI_j = \sfH_{1j}\cap \sfH_{2j},\quad \sfT_j = \sfH_{1j}\triangle \sfH_{2j},\quad \text{for any }j\in \sth{1,2}.
    \end{align}
    We note that \begin{align*}
        \E_{\maP_1}\qth{\prod_{i,j}\hom_{\vp_{ij}}(\sfM_i,\bar{G}_j)}  =\E_\pi \E_{\maP_1|\pi}\qth{\prod_{j}\pth{\prod_{e\in E(\sfI_j)} \beta_e^2(\bar{G}_j) \prod_{e\in E(\sfT_j)} \beta_e(\bar{G}_j)} }.
    \end{align*}
Given any $\pi:V(\bar{G}_1)\mapsto V(\bar{G}_2)$,   for any $e_0\in E(\sfT_1)$, if $\pi(e_0)\notin E(\sfI_2\cup \sfT_2)$, since $\E_{\maP_1|\pi}\qth{\beta_{e_0}(\bar{G}_1)} = 0$ and $\beta_{e_0}(\bar{G}_1)$ is independent with \begin{align*}
         \prod_{e\in E(\sfI_1)} \beta_e^2(\bar{G}_1) \prod_{e\in E(\sfT_1)\backslash e_0} \beta_e(\bar{G}_1)\prod_{e\in E(\sfI_2)}\beta_e^2(\bar{G}_2)\prod_{e\in E(\sfT_2) }\beta_e(\bar{G}_2),
    \end{align*}
    then $\E_{\maP_1|\pi}\qth{\prod_{j}\pth{\prod_{e\in E(\sfI_j)} \beta_e^2(\bar{G}_j) \prod_{e\in E(\sfT_j)} \beta_e(\bar{G}_j)} }=0$. 
    Therefore,  two necessary conditions for $\E_{\maP_1|\pi}\qth{\prod_{j}\pth{\prod_{e\in E(\sfI_j)} \beta_e^2(\bar{G}_j) \prod_{e\in E(\sfT_j)} \beta_e(\bar{G}_j)} }\neq 0$ are $\pi(V(\sfT_1))\subseteq V(\sfI_2\cup \sfT_2)$ and $\pi^{-1}(V(\sfT_2))\subseteq V(\sfI_1\cup \sfT_1)$. 
    Since \begin{align*}
        |\pi(V(\sfT_1))| &= |V(\sfH_{11}\triangle \sfH_{21})|\\ &= \vn(\sfM_1)+\vn(\sfM_2)-2|V(\sfH_{11})\cap V(\sfH_{21})|\\&~~~~+|V(\sfH_{11}\triangle \sfH_{21})\cap (V(\sfH_{11})\cap V(\sfH_{21}))|\\&\ge \vn(\sfM_1)+\vn(\sfM_2)-2\sfn_1\end{align*} and \begin{align*} \vn(\sfI_2\cup \sfT_2) = \vn(\sfH_{12}\cup\sfH_{22}) =  \vn(\sfH_{12})+\vn(\sfH_{22})-\sfn_2,
    \end{align*}
when $2\sfn_1<\sfn_2$, we have $|\pi(V(\sfT_1))|>\vn(\sfI_2\cup\sfT_2)$, and thus $\pi(V(\sfT_1))\nsubseteq V(\sfI_2\cup \sfT_2)$. Similarly, when $2\sfn_2<\sfn_1$, we have $\pi^{-1}(V(\sfT_2))\nsubseteq V(\sfI_1\cup \sfT_1)$. Therefore, for any $\vp\in \Phi_D$, 
\begin{align*}
        \E_{\maP_1}\qth{\prod_{i,j}\hom_{\vp_{ij}}(\sfM_i,\bar{G}_j)} = 0.
    \end{align*}

\textbf{Case 2: Balanced overlap.\ \ }We then focus on $\vp \in \Phi_B$. 
For any bijective mapping $\pi:V(\bar{G}_1)\mapsto V(\bar{G}_2)$, let 
\begin{align*}
    \sfE_{i,j}\triangleq\sth{e\in E(\sfH_{11}\cup \sfH_{21}):\sum_{k=1}^2 \indc{e\in E(\sfH_{k1})} = i,\sum_{k=1}^2\indc{\pi(e)\in E(\sfH_{k2})} = j },\quad \forall 0\le i,j\le 2.
\end{align*}
Define $\maS\triangleq \sth{(1,1),(1,2),(2,1),(2,2),(0,2),(2,0)}$. We note that \begin{align}
    \nonumber &~\E_{\maP_1|\pi}\qth{\prod_{i,j}\frac{\hom_{\vp_{ij}}(\sfM_i,\bar{G}_j)}{\sqrt{(p(1-p))^{\en(\sfM_i)}}}}   \\\label{eq:upbd-vp1-4-moment}
    =&~\E_{\maP_1|\pi}\qth{\prod_{(i,j)\in \maS}\prod_{e\in \sfE_{i,j}} \pth{\frac{\beta_e(\bar{G}_1)}{\sqrt{p(1-p)}}}^i\pth{\frac{\beta_{\pi(e)}(\bar{G}_2)}{\sqrt{p(1-p)}}}^{j}} \indc{\pi(\sfT_1)\subseteq \sfI_2\cup \sfT_2,\pi^{-1}(\sfT_2)\subseteq \sfI_1\cup \sfT_1}, 
\end{align}
where the last equality follows from the fact that two necessary conditions for $$\E_{\maP_1|\pi}\qth{\prod_{j}\pth{\prod_{e\in E(\sfI_j)} \beta_e^2(\bar{G}_j) \prod_{e\in E(\sfT_j)} \beta_e(\bar{G}_j)} }\neq 0$$ are $\pi(\sfT_1)\subseteq \sfI_2\cup \sfT_2$ and $\pi^{-1}(\sfT_2)\subseteq \sfI_1\cup \sfT_1$.
For a correlated pair $(e,\pi(e))$, we have $\E_{\maP_1|\pi} \qth{\frac{\beta_{e}(\bar{G}_1)\beta_{\pi(e)}(\bar{G}_2)}{p(1-p)}} = \rho\le 1$ and \begin{align*}
    \E\qth{\frac{\beta_e^2(\bar{G}_1)}{p(1-p)}} = \E\qth{\frac{\beta_{\pi(e)}^2(\bar{G}_2)}{p(1-p)}}=1.
\end{align*}
Combining with~\eqref{eq:upbd-vp1-4-moment} and Lemma~\ref{lem:moment-bound}, we obtain 
\begin{align}
    \nonumber &~\E_{\maP_1|\pi}\qth{\prod_{(i,j)\in \maS}\prod_{e\in \sfE_{i,j}} \pth{\frac{\beta_e(\bar{G}_1)}{\sqrt{p(1-p)}}}^i\pth{\frac{\beta_{\pi(e)}(\bar{G}_2)}{\sqrt{p(1-p)}}}^{j}} \indc{\pi(\sfT_1)\subseteq \sfI_2\cup \sfT_2,\pi^{-1}(\sfT_2)\subseteq \sfI_1\cup \sfT_1}\\\nonumber
    \le&~ \pth{\pth{\prod_{\substack{(i,j)\in \maS\\ i+j=2}} \prod_{e\in \sfE_{i,j}}1 }\pth{\prod_{\substack{(i,j)\in \maS\\ i+j=3}} \prod_{e\in \sfE_{i,j}} \frac{1}{\sqrt{p}} }\pth{\prod_{e\in \sfE_{2,2}} \frac{1}{p}}} \indc{\pi(\sfT_1)\subseteq \sfI_2\cup \sfT_2,\pi^{-1}(\sfT_2)\subseteq \sfI_1\cup \sfT_1}\\\nonumber
    =&~\pth{\frac{1}{\sqrt{p}}}^{|\sfE_{1,2}|+|\sfE_{2,1}|+2|\sfE_{2,2}|} \indc{\pi(\sfT_1)\subseteq \sfI_2\cup \sfT_2,\pi^{-1}(\sfT_2)\subseteq \sfI_1\cup \sfT_1}\\\label{eq:upbd-vp1-4-sqrtp}=&~\pth{\frac{1}{\sqrt{p}}}^{(|\sfE_{1,2}|+| \sfE_{2,2}|)+(|\sfE_{2,1}|+ |\sfE_{2,2}|)} \indc{\pi(\sfT_1)\subseteq \sfI_2\cup \sfT_2,\pi^{-1}(\sfT_2)\subseteq \sfI_1\cup \sfT_1}.
\end{align}

Let $\sfS_1 = \sfI_1\cap \pth{\pi^{-1}(\sfI_2\cup \sfT_2)}$ and $\sfS_2 = \pi(\sfI_1\cup \sfT_1) \cap \sfI_2$. Since $E(\sfI_1)\cap E(\sfT_1)=E(\sfI_2)\cap E(\sfT_2) = \emptyset$, we have \begin{align*}
    |\sfE_{2,1}|+ |\sfE_{2,2}| = \en(\sfS_1),\quad |\sfE_{1,2}|+|\sfE_{2,2}| = \en(\sfS_2).
\end{align*}
We then verify $\pi(\sfS_1\cup \sfT_1) = \sfS_2\cup\sfT_2$ when $\pi(\sfT_1)\subseteq \sfI_2\cup \sfT_2$ and $\pi^{-1}(\sfT_2)\subseteq \sfI_1\cup \sfT_1$. Since $\pi(\sfT_1)\subseteq \sfI_2\cup \sfT_2$, we have $\pi(\sfT_1)\cap  (\sfI_2\cup \sfT_2 )= \pi(\sfT_1)$.
Therefore, \begin{align*}
    \pi(\sfS_1\cup \sfT_1) &= \pth{\pi(\sfI_1)\cap (\sfI_2\cup\sfT_2)} \cup (\pi(\sfT_1))\\ &= \pth{\pi(\sfI_1)\cap (\sfI_2\cup\sfT_2)}\cup \pth{\pi(\sfT_1)\cap (\sfI_2\cup\sfT_2)}=\pi(\sfI_1\cup \sfT_1) \cap (\sfI_2\cup \sfT_2).
\end{align*}
Similarly, $\sfS_2\cup \sfT_2 = \pi(\sfI_1\cup \sfT_1) \cap (\sfI_2\cup \sfT_2)$, and thus we have $\pi(\sfS_1\cup \sfT_1)= \sfS_2\cup \sfT_2 $ when $\pi(\sfT_1)\subseteq \sfI_2\cup \sfT_2$ and $\pi^{-1}(\sfT_2)\subseteq \sfI_1\cup \sfT_1$. Combining this with~\eqref{eq:upbd-vp1-4-moment} and~\eqref{eq:upbd-vp1-4-sqrtp}, we obtain that \begin{align}
    \nonumber &~\E_\pi\E_{\maP_1|\pi}\qth{\prod_{i,j}\frac{\hom_{\vp_{ij}}(\sfM_i,\bar{G}_j)}{\sqrt{(p(1-p))^{\en(\sfM_i)}}}}\\\nonumber\le&~ \E_\pi\qth{\pth{\frac{1}{\sqrt{p}}}^{\en(\sfS_1)+\en(\sfS_2)}\indc{\pi(\sfS_1\cup\sfT_1) = \sfS_2\cup \sfT_2}}\\\nonumber
    \overset{\mathrm{(a)}}{\le}&~\E_\pi\qth{\max_{\ti{\sfS}_1\subseteq \sfI_1,\ti{\sfS}_2\subseteq \sfI_2}\qth{\pth{\frac{1}{\sqrt{p}}}^{\en(\ti{\sfS}_1)+\en(\ti{\sfS}_2)}\indc{\pi(\ti{\sfS}_1\cup\sfT_1) = \ti{\sfS}_2\cup \sfT_2} }}\\\label{eq:sum-prob-pi-p}
    \overset{\mathrm{(b)}}{\le}&~ \sum_{\ti{\sfS}_1\subseteq \sfI_1}\sum_{\ti{\sfS}_2\subseteq \sfI_2} \qth{\pth{\frac{1}{\sqrt{p}}}^{\en(\ti{\sfS}_1)+\en(\ti{\sfS}_2)}\prob{\pi(\ti{\sfS}_1\cup\sfT_1) = \ti{\sfS}_2\cup \sfT_2} },
\end{align}
where $\mathrm{(a)}$ is because $\sfS_1\subseteq \sfI_1$ and $\sfS_2\subseteq \sfI_2$; $\mathrm{(b)}$ applies the union bound.

Recall that $\sfT_1 = \sfH_{11}\triangle \sfH_{21}$ and $\sfT_2 = \sfH_{12}\triangle \sfH_{22}$. 
We note that $\sfH_{ij}$ is connected for all $i,j \in \{1,2\}$, since each $\sfM \in \maM$ is connected.
Since $\ti{\sfS}_i\subseteq \sfI_i = \sfH_{1i}\cap \sfH_{2i}$ for $i\in \sth{1,2}$, by Lemma~\ref{lem:aux-lem},\begin{align*}
    |V(\ti{\sfS}_1\cup \sfT_1)|&\ge \vn(\sfM_1)+\vn(\sfM_2) -2\sfn_1+\vn(\ti{\sfS}_1)+\indc{\ti{\sfS}_1=\emptyset, \sfH_{11}\neq \sfH_{21}},\\
    |V(\ti{\sfS}_2\cup \sfT_2)|&\ge \vn(\sfM_1)+\vn(\sfM_2)-2\sfn_2+\vn(\ti{\sfS}_2)+\indc{\ti{\sfS}_2=\emptyset, \sfH_{12}\neq \sfH_{22}},\\
    \prob{\pi(\ti{\sfS}_1\cup \sfT_1) = \ti{\sfS}_2\cup \sfT_2}&\le \pth{\frac{\max\sth{\vn(\ti{\sfS}_1\cup \sfT_1),\vn(\ti{\sfS}_2\cup \sfT_2)}}{n}}^{\frac{\vn(\ti{\sfS}_1\cup \sfT_1)+\vn(\ti{\sfS}_2\cup \sfT_2)}{2}}.
\end{align*}   

We note that $\max\sth{\vn(\ti{\sfS}_1\cup \sfT_1),\vn(\ti{\sfS}_2\cup \sfT_2)}\le \vn(\sfM_1)+\vn(\sfM_2)\le 2C$.  Therefore, we obtain that \begin{align*}
    &~\prob{\pi(\ti{\sfS}_1\cup \sfT_1) = \ti{\sfS}_2\cup \sfT_2}\\\le&~ \pth{\frac{2C}{n}}^{\frac{\vn(\ti{\sfS}_1\cup \sfT_1)+\vn(\ti{\sfS}_2\cup \sfT_2)}{2}}\\\le&~ \pth{\frac{2C}{n}}^{\vn(\sfM_1)+\vn(\sfM_2)-\sfn_1-\sfn_2+\frac{1}{2}\pth{\vn(\ti{\sfS}_1)+\vn(\ti{\sfS}_2)+\indc{\ti{\sfS}_1=\emptyset, \sfH_{11}\neq \sfH_{21}}+\indc{\ti{\sfS}_2=\emptyset, \sfH_{12}\neq \sfH_{22}}}}.
\end{align*}
Combining this with~\eqref{eq:sum-prob-pi-p}, we have \begin{align}
    \nonumber &~\E_\pi\E_{\maP_1|\pi}\qth{\prod_{i,j}\frac{\hom_{\vp_{ij}}(\sfM_i,\bar{G}_j)}{\sqrt{(p(1-p))^{\en(\sfM_i)}}}}
    \\\nonumber \le&~\pth{\frac{2C}{n}}^{\vn(\sfM_1)+\vn(\sfM_2)-\sfn_1-\sfn_2} \\\label{eq:upbd-vp-term-PhiB}&~\cdot\sum_{\ti{\sfS}_1\subseteq \sfI_1}\sum_{\ti{\sfS}_2\subseteq \sfI_2} \pth{\frac{2C}{n}}^{\frac{1}{2}\pth{\vn(\ti{\sfS}_1)+\vn(\ti{\sfS}_2)+\indc{\ti{\sfS}_1=\emptyset, \sfH_{11}\neq \sfH_{21}}+\indc{\ti{\sfS}_2=\emptyset, \sfH_{12}\neq \sfH_{22}}}}\pth{\frac{1}{\sqrt{p}}}^{\en(\ti{\sfS}_1)+\en(\ti{\sfS}_2)},
\end{align}
where \begin{align*}
    &~\sum_{\ti{\sfS}_1\subseteq \sfI_1}\sum_{\ti{\sfS}_2\subseteq \sfI_2} \pth{\frac{2C}{n}}^{\frac{1}{2}\pth{\vn(\ti{\sfS}_1)+\vn(\ti{\sfS}_2)+\indc{\ti{\sfS}_1=\emptyset, \sfH_{11}\neq \sfH_{21}}+\indc{\ti{\sfS}_2=\emptyset, \sfH_{12}\neq \sfH_{22}}}}\pth{\frac{1}{\sqrt{p}}}^{\en(\ti{\sfS}_1)+\en(\ti{\sfS}_2)}\\
    =&~\Bigg[\sum_{\ti{\sfS}_1\subseteq \sfI_1} \pth{\frac{2C}{n}}^{\frac{\vn(\ti{\sfS}_1)+\indc{\ti{\sfS}_1 = \emptyset, \sfH_{11}\neq \sfH_{21}}}{2}} \pth{\frac{1}{\sqrt{p}}}^{\en(\ti{\sfS}_1)}\Bigg]\\&\cdot \Bigg[\sum_{\ti{\sfS}_2\subseteq \sfI_2} \pth{\frac{2C}{n}}^{\frac{\vn(\ti{\sfS}_2)+\indc{\ti{\sfS}_2 = \emptyset, \sfH_{12}\neq \sfH_{22}}}{2}} \pth{\frac{1}{\sqrt{p}}}^{\en(\ti{\sfS}_2)}\Bigg].
\end{align*}

We note that \begin{align*}
    &~\sum_{\ti{\sfS}_1\subseteq \sfI_1} \pth{\frac{2C}{n}}^{\frac{\vn(\ti{\sfS}_1)+\indc{\ti{\sfS}_1 = \emptyset, \sfH_{11}\neq \sfH_{21}}}{2}} \pth{\frac{1}{\sqrt{p}}}^{\en(\ti{\sfS}_1)}\\ =&~ \pth{\frac{2C}{n}}^{\frac{\indc{\sfH_{11}\neq \sfH_{21}}}{2}}+\sum_{\ti{\sfS}_1\subseteq \sfI_1,\ti{\sfS}_1\neq \emptyset} \pth{\frac{2C}{n}}^{\vn(\ti{\sfS}_1)/2}\pth{\frac{1}{\sqrt{p}}}^{\en(\ti{\sfS}_1)}\\
    \overset{\mathrm{(a)}}{\le}&~\pth{\frac{2C}{n}}^{\frac{\indc{\sfH_{11}\neq \sfH_{21}}}{2}}+\sum_{\ti{\sfS}_1\subseteq \sfI_1,\ti{\sfS}_1\neq \emptyset} (2C)^C n^{-\vn(\ti{\sfS}_1)/2}p^{-\en(\ti{\sfS}_1)/2}\\
    \overset{\mathrm{(b)}}{\le}&~\pth{\frac{2C}{n}}^{\frac{\indc{\sfH_{11}\neq \sfH_{21}}}{2}}+\sum_{\ti{\sfS}_1\subseteq \sfI_1,\ti{\sfS}_1\neq \emptyset} (2C)^C n^{-\epsilon_0/2} \\\overset{\mathrm{(c)}}{\le}&~ \indc{\sfH_{11}=\sfH_{21}}+n^{-\epsilon_0/2}(2C)^C(2^C-1)+\pth{\frac{2C}{n}}^{1/2}\\\overset{\mathrm{(d)}}{\le}&~\indc{\sfH_{11}=\sfH_{21}}+n^{-\epsilon_0/2}(4C)^C,
\end{align*}
where $\mathrm{(a)}$ is because $\frac{\vn(\ti{\sfS}_1)}{2}\le C$; 
$\mathrm{(b)}$ follows from the Condition~\ref{cond:subgraph} for \emph{$C$-admissible} motif family $\maM$; $\mathrm{(c)}$ is because there are at most $2^C-1$ choices for $\ti{\sfS}_1\subseteq \sfI_1$ with $\ti{\sfS}_1\neq \emptyset$; $\mathrm{(d)}$ follows because choosing $\sfM'$ with $\vn(\sfM')=1$ in Condition~\ref{cond:subgraph} implies $\epsilon_0<1$, and thus $(2C/n)^{1/2}\le n^{-\epsilon_0/2}(2C)^C$ as $C = o\pth{\frac{\log n}{\log\log n}}$. 
Similarly, we have \begin{align*}
    \sum_{\ti{\sfS}_2\subseteq \sfI_2} \pth{\frac{2C}{n}}^{\frac{\vn(\ti{\sfS}_2)+\indc{\ti{\sfS}_2 = \emptyset, \sfH_{12}\neq \sfH_{22}}}{2}} \pth{\frac{1}{\sqrt{p}}}^{\en(\ti{\sfS}_2)}\le \indc{\sfH_{12} = \sfH_{22}}+n^{-\epsilon_0/2}(4C)^C.
\end{align*}
Combining this with~\eqref{eq:upbd-vp-term-PhiB}, we obtain \begin{align*}
    &~\E_{\maP}\qth{\prod_{i,j}\frac{\hom_{\vp_{ij}}(\sfM_i,\bar{G}_j)}{\sqrt{(p(1-p))^{\en(\sfM_i)}}}}
    \\\nonumber \le&~\pth{\frac{2C}{n}}^{\vn(\sfM_1)+\vn(\sfM_2)-\sfn_1-\sfn_2} \pth{\indc{\sfH_{11} = \sfH_{21}}+n^{-\epsilon_0/2}(4C)^C}\pth{\indc{\sfH_{12} = \sfH_{22}}+n^{-\epsilon_0/2}(4C)^C}\\
    \le&~\pth{\frac{2C}{n}}^{\vn(\sfM_1)+\vn(\sfM_2)-\sfn_1-\sfn_2}\pth{\indc{\sfH_{11}=\sfH_{21},\sfH_{12} = \sfH_{22}}+2n^{-\epsilon_0/2}\pth{4C}^C+n^{-\epsilon_0}\pth{4C}^{2C}}\\
    \le&~\pth{\frac{2C}{n}}^{\vn(\sfM_1)+\vn(\sfM_2)-\sfn_1-\sfn_2}\pth{\indc{\sfH_{11}=\sfH_{21},\sfH_{12} = \sfH_{22}}+3n^{-\epsilon_0/2}\pth{4C}^{2C}}.
\end{align*}

\textbf{Case 3: Null overlap.\ \ }We finally consider the case $\vp\in \Phi_N$, where $\sfn_1=\sfn_2 =0$. 
We note that $\sfH_{11}\cap \sfH_{21}=\sfH_{12}\cap \sfH_{22}= \emptyset$ under this case. Therefore,\begin{align*}
        &~\E_{\maP_1}\qth{\prod_{i,j}\frac{\hom_{\vp_{ij}}(\sfM_i,\bar{G}_j)}{\sqrt{(p(1-p))^{\en(\sfM_i)}}}}\\ =&~ \E_\pi \E_{\maP_1|\pi}\qth{  \prod_{e\in E(\sfH_{11}\cup \sfH_{21})} \frac{\beta_e(\bar{G}_1)}{\sqrt{p(1-p)}} \prod_{e\in E(\sfH_{12}\cup \sfH_{22}) }\frac{\beta_e(\bar{G}_2)}{\sqrt{p(1-p)}}}\\=&~\E_\pi \qth{\rho^{\en(\sfM_1)+\en(\sfM_2)}\indc{\pi(E(\sfH_{11}\cup \sfH_{21})) = E(\sfH_{12}\cup \sfH_{22})}}.
    \end{align*}
We note that for any $\sfM_1,\sfM_2\in \maM$, the motifs $\sfM_1,\sfM_2$ are connected. Consequently, four motifs $\sfH_{11},\sfH_{12},\sfH_{21}$, and $\sfH_{22}$ induced by $\sfM_1$ and $\sfM_2$ are all connected. 
Given $\pi(E(\sfH_{11}\cup \sfH_{21})) = E(\sfH_{12}\cup \sfH_{22})$ and $\sfM_1 = \sfM_2$, we must have $\pi(E(\sfH_{11})) = E(\sfH_{12}),\pi(E(\sfH_{21}))=E(\sfH_{22})$ or $\pi(E(\sfH_{11})) = E(\sfH_{22}),\pi(E(\sfH_{21})=E(\sfH_{12})$. When $\pi(E(\sfH_{11}\cup \sfH_{21})) = E(\sfH_{12}\cup \sfH_{22})$ and $\sfM_1 \neq \sfM_2$, we only have $\pi(E(\sfH_{11})) = E(\sfH_{12}),\pi(E(\sfH_{21}))=E(\sfH_{22})$.
For two connected motifs $\sfH$ and $\sfH'$, we note that $\pi(E(\sfH)) = \pi(E(\sfH'))$ is equivalent to $\pi(\sfH) = \pi(\sfH')$.
Therefore,\begin{align*}
    &\E_{\maP_1}\qth{\prod_{i,j}\frac{\hom_{\vp_{ij}}(\sfM_i,\bar{G}_j)}{\sqrt{(p(1-p))^{\en(\sfM_i)}}}}\\
    =&~\E_\pi \qth{\rho^{\en(\sfM_1)+\en(\sfM_2)}\indc{\pi(E(\sfH_{11}\cup \sfH_{21})) = E(\sfH_{12}\cup \sfH_{22})}}\\
    =&~\rho^{\en(\sfM_1)+\en(\sfM_2)}\prob{\pi(\sfH_{11}) = \sfH_{12},\pi(\sfH_{21})=\sfH_{22}}\\&\quad+\rho^{\en(\sfM_1)+\en(\sfM_2)}\prob{\pi(\sfH_{11}) = \sfH_{22},\pi(\sfH_{21})=\sfH_{12}}\indc{\sfM_1 = \sfM_2}\\
    =&~\rho^{\en(\sfM_1)+\en(\sfM_2)}\pth{\frac{(n-\vn(\sfM_1)-\vn(\sfM_2))! \aut(\sfM_1)\aut(\sfM_2)}{n!}}(1+\indc{\sfM_1=\sfM_2}).
\end{align*}

\section{Auxiliary results}

\begin{lemma}\label{lem:moment-bound}
    For any bijective mapping $\pi:V(\bar{G}_1)\mapsto V(\bar{G}_2)$ and a correlated pair $(e,\pi(e))$, where $e\in V(\bar{G}_1)$, we have \begin{align*}
        \E_{\maP_1|\pi} \qth{\beta_e^2(\bar{G}_1) \beta_{\pi(e)}(\bar{G}_2)} = \E_{\maP_1|\pi} \qth{\beta_e(\bar{G}_1) \beta_{\pi(e)}^2(\bar{G}_2)}&\le \pth{p(1-p)}^{3/2}\cdot \sqrt{\frac{1}{p}},\\
        \E_{\maP_1|\pi}\qth{\beta_e^2(\bar{G}_1) \beta_{\pi(e)}^2(\bar{G}_2)}&\le (p(1-p))^2\cdot \frac{1}{p}. 
    \end{align*}
\end{lemma}

\begin{proof}
    We note that \begin{align*}
        \E_{\maP_1|\pi} \qth{\beta_e^2(\bar{G}_1) \beta_{\pi(e)}(\bar{G}_2)} &= \sum_{i,j\in \sth{0,1}} \prob{\beta_e(\bar{G}_1) = i-p,\beta_{\pi(e)}(\bar{G}_2) = j-p} (i-p)^2 (j-p)\\
        &=p(1-p)(1-2p)\rho.
    \end{align*}
    Since $0<p\le \frac{1}{2}$, we have \begin{align*}
        \E_{\maP_1|\pi} \qth{\beta_e^2(\bar{G}_1) \beta_{\pi(e)}(\bar{G}_2)} &=p(1-p)(1-2p)\rho\\&\le p(1-p)\sqrt{1-4p^2+4p}\\
        &\le p(1-p)\sqrt{1-p} = \pth{p(1-p)^{3/2}}\cdot \sqrt{\frac{1}{p}}.
    \end{align*}
    Similarly, $\E_{\maP_1|\pi} \qth{\beta_e^2(\bar{G}_1) \beta_{\pi(e)}(\bar{G}_2)}\le \pth{p(1-p)^{3/2}}\cdot \sqrt{\frac{1}{p}}$.

    For $\E_{\maP_1|\pi}\qth{\beta_e^2(\bar{G}_1) \beta_{\pi(e)}^2(\bar{G}_2)}$, we have \begin{align*}
        \E_{\maP_1|\pi}\qth{\beta_e^2(\bar{G}_1) \beta_{\pi(e)}^2(\bar{G}_2)} &= \sum_{i,j\in \sth{0,1}} \prob{\beta_e(\bar{G}_1) = i-p,\beta_{\pi(e)}(\bar{G}_2) = j-p} (i-p)^2 (j-p)^2\\
        &=p^2(1-p)^2\pth{1+\frac{\rho (2p-1)^2}{p(1-p)}}\\&\le p^2 (1-p)^2 \pth{\frac{p(1-p)+(2p-1)^2}{p(1-p)}}\le p^2(1-p)^2\cdot \frac{1}{p},
    \end{align*}
    where the last inequality is because $\frac{p(1-p)+(2p-1)^2}{p(1-p)}=\frac{3p^2-2p+1-p}{p(1-p)}\le \frac{1}{p}$.
\end{proof}

\begin{lemma}\label{lem:aux-lem}
    Let $\sfM_1,\sfM_2\in \maM$ and $\sfH_1\subseteq \sfM_1,\sfH_2\subseteq \sfM_2$ be two connected subgraphs of $\sfM_1$ and $\sfM_2$, respectively.
    \begin{enumerate} 
        \item Let $\pi$ be sampled uniformly from all bijections between $V(G_1)$ and $V(G_2)$. We have 
    \begin{align*}
        \prob{\pi(\sfH_1) = \sfH_2}\le \min\pth{\pth{\frac{\vn(\sfH_1)}{n}}^{\vn(\sfH_1)},\pth{\frac{\vn(\sfH_2)}{n}}^{\vn(\sfH_2)}}.
    \end{align*}
    Furthermore, \begin{align*}
        \prob{\pi(\sfH_1)=\sfH_2}\le \pth{\frac{\max\pth{\vn(\sfH_1),\vn(\sfH_2)}}{n}}^{\frac{\vn(\sfH_1)+\vn(\sfH_2)}{2}}.
    \end{align*}
    \item If $|V(\sfH_1)\cap V(\sfH_2)|\ge 1$, then for any subgraph $\sfH_0\subseteq \sfH_1\cap \sfH_2$, \begin{align*}
        |V((\sfH_1\triangle \sfH_2)\cup \sfH_0)|\ge \vn(\sfH_1)+\vn(\sfH_2)-2|V(\sfH_1)\cap V(\sfH_2)|+\vn(\sfH_0)+\indc{\sfH_0=\emptyset, \sfH_1\neq \sfH_2},
    \end{align*}
    where $\sfH_0 =\emptyset$ denotes the empty subgraph with no vertices and no edges.
    \end{enumerate}
\end{lemma}
\begin{proof}
    (1) 
    On the one hand, \begin{align*}
        \prob{\pi(\sfH_1) = \sfH_2} &= \frac{(n-\vn(\sfH_1))!\aut(\sfH_1)}{n!} \indc{\sfH_1= \sfH_2}\\&\overset{\mathrm{(a)}}{\le} \frac{(n-\vn(\sfH_1))! (\vn(\sfH_1))!}{n!} = \prod_{i=1}^{\vn(\sfH_1)} \frac{i}{n-\vn(\sfH_1)+i}\\&\overset{\mathrm{(b)}}{\le} \prod_{i=1}^{\vn(\sfH_1)} \frac{\vn(\sfH_1)}{n} = \pth{\frac{\vn(\sfH_1)}{n}}^{\vn(\sfH_1)},
    \end{align*}
    where $\mathrm{(a)}$ is because $\aut(\sfH_1)\le (\vn(\sfH_1))!$ and $\mathrm{(b)}$ is because $\frac{i}{n-\vn(\sfH_1)+i}\le \frac{\vn(\sfH_1)}{n}$ for any $1\le i\le \vn(\sfH_1)$.
    On the other hand,
    \begin{align*}
        \prob{\pi(\sfH_1) = \sfH_2} &= \frac{(n-\vn(\sfH_1))!\aut(\sfH_1)}{n!} \indc{\sfH_1= \sfH_2}\\&\le\frac{(n-\vn(\sfH_2))!\aut(\sfH_2)}{n!} \indc{\sfH_1= \sfH_2}  = \pth{\frac{\vn(\sfH_2)}{n}}^{\vn(\sfH_2)}.
    \end{align*}
    Therefore,\begin{align*}
        \prob{\pi(\sfH_1) = \sfH_2}&\le \min\pth{\pth{\frac{\vn(\sfH_1)}{n}}^{\vn(\sfH_1)},\pth{\frac{\vn(\sfH_2)}{n}}^{\vn(\sfH_2)}}\\&\le\min\pth{\pth{\frac{\max\pth{\vn(\sfH_1),\vn(\sfH_2)}}{n}}^{\vn(\sfH_1)},\pth{\frac{\max\pth{\vn(\sfH_1),\vn(\sfH_2)}}{n}}^{\vn(\sfH_2)}} \\&=\pth{\frac{\max\pth{\vn(\sfH_1),\vn(\sfH_2)}}{n}}^{\max\pth{\vn(\sfH_1),\vn(\sfH_2)}}\\&\le   \pth{\frac{\max\pth{\vn(\sfH_1),\vn(\sfH_2)}}{n}}^{\frac{\vn(\sfH_1)+\vn(\sfH_2)}{2}}.
    \end{align*}

    (2) We note that \begin{align*}
        |V((\sfH_1\triangle \sfH_2)\cup \sfH_0)| =&~ |V(\sfH_1\triangle \sfH_2)\cup V(\sfH_0)|\\ =&~ |V(\sfH_1\triangle \sfH_2)|+\vn(\sfH_0) - |V(\sfH_1\triangle \sfH_2)\cap V(\sfH_0)|\\
        =&~\vn(\sfH_1)+\vn(\sfH_2)-2|V(\sfH_1)\cap V(\sfH_2)|\\&+|V(\sfH_1\triangle \sfH_2)\cap (V(\sfH_1)\cap V(\sfH_2))|\\&+\vn(\sfH_0) - |V(\sfH_1\triangle \sfH_2)\cap V(\sfH_0)|.
    \end{align*}
    It suffices to prove \begin{align}\label{eq:aux-1}
        |V(\sfH_1\triangle \sfH_2)\cap (V(\sfH_1)\cap V(\sfH_2))| - |V(\sfH_1\triangle \sfH_2)\cap V(\sfH_0)|\ge \indc{\sfH_0=\emptyset,\sfH_1\neq \sfH_2}.
    \end{align}
    Since $\sfH_0\subseteq \sfH_1\cap \sfH_2$, we obtain that $V(\sfH_0)\subseteq V(\sfH_1\cap \sfH_2)\subseteq V(\sfH_1)\cap V(\sfH_2)$, and thus \begin{align*}
        |V(\sfH_1\triangle \sfH_2)\cap (V(\sfH_1)\cap V(\sfH_2))| - |V(\sfH_1\triangle \sfH_2)\cap V(\sfH_0)|\ge 0.
    \end{align*}
    It remains to prove~\eqref{eq:aux-1} when $\sfH_0 = \emptyset$ and $\sfH_1\neq \sfH_2$. We note that $|V(\sfH_1\triangle \sfH_2)\cap V(\sfH_0)| = 0$ when $\sfH_0 =\emptyset$.
    It suffices to show $|V(\sfH_1\triangle \sfH_2)\cap (V(\sfH_1)\cap V(\sfH_2))|\ge 1$ when $\sfH_1\neq \sfH_2$.
    If $\sfH_1\cap \sfH_2=\emptyset$, then $\sfH_1\triangle \sfH_2 = \sfH_1\cup \sfH_2$, and thus \begin{align*}
        |V(\sfH_1\triangle \sfH_2)\cap (V(\sfH_1)\cap V(\sfH_2))| &= |V(\sfH_1\cup  \sfH_2)\cap (V(\sfH_1)\cap V(\sfH_2))|\\&=|V(\sfH_1)\cap V(\sfH_2)|\ge 1.
    \end{align*}
    If $\sfH_1\cap \sfH_2\neq \emptyset$, then $V(\sfH_1\cap \sfH_2)\neq \emptyset$. Since $|V(\sfH_1)\cap V(\sfH_2)|\ge 1$, $\sfH_1\cup \sfH_2$ are connected. Recall that $\sfH_1\neq \sfH_2$, and thus $V(\sfH_1\triangle \sfH_2)\neq \emptyset$. Therefore,\begin{align*}
        |V(\sfH_1\triangle \sfH_2)\cap (V(\sfH_1)\cap V(\sfH_2))|\ge |V(\sfH_1\triangle \sfH_2) \cap V(\sfH_1\cap \sfH_2)|\ge 1,
    \end{align*}
    where the last inequality follows from the fact that $\sfH_1\cup \sfH_2 = (\sfH_1\triangle \sfH_2)\cup (\sfH_1\cap \sfH_2)$ is connected and $V(\sfH_1\triangle \sfH_2),V(\sfH_1\cap \sfH_2)\neq \emptyset$.
\end{proof}
\bibliographystyle{alpha}
\bibliography{main}

@inproceedings{narayanan2008robust,
  title={Robust de-anonymization of large sparse datasets},
  author={Narayanan, Arvind and Shmatikov, Vitaly},
  booktitle={2008 IEEE Symposium on Security and Privacy (SP 2008)},
  pages={111--125},
  year={2008},
  organization={IEEE}
}

@article{chen2025detecting,
  title={Detecting correlation efficiently in stochastic block models: breaking {Otter's} threshold by counting decorated trees},
  author={Chen, Guanyi and Ding, Jian and Gong, Shuyang and Li, Zhangsong},
  journal={arXiv preprint arXiv:2503.06464},
  year={2025}
}

@article{chen2024computational,
  title={A computational transition for detecting correlated stochastic block models by low-degree polynomials},
  author={Chen, Guanyi and Ding, Jian and Gong, Shuyang and Li, Zhangsong},
  journal={arXiv preprint arXiv:2409.00966},
  year={2024}
}

@article{ameen2024exact,
  title={Exact random graph matching with multiple graphs},
  author={Ameen, Taha and Hajek, Bruce},
  journal={arXiv preprint arXiv:2405.12293},
  year={2024}
}

@article{bandeira2019computational,
  title={Computational hardness of certifying bounds on constrained {PCA} problems},
  author={Bandeira, Afonso S and Kunisky, Dmitriy and Wein, Alexander S},
  journal={arXiv preprint arXiv:1902.07324},
  year={2019}
}

@article{dhawan2025detection,
  title={Detection of dense subhypergraphs by low-degree polynomials},
  author={Dhawan, Abhishek and Mao, Cheng and Wein, Alexander S},
  journal={Random Structures \& Algorithms},
  volume={66},
  number={1},
  pages={e21279},
  year={2025},
  publisher={Wiley Online Library}
}

@article{sohn2025sharp,
  title={Sharp phase transitions in estimation with low-degree polynomials},
  author={Sohn, Youngtak and Wein, Alexander S},
  journal={arXiv preprint arXiv:2502.14407},
  year={2025}
}

@article{schramm2022computational,
  title={Computational barriers to estimation from low-degree polynomials},
  author={Schramm, Tselil and Wein, Alexander S},
  journal={The Annals of Statistics},
  volume={50},
  number={3},
  pages={1833--1858},
  year={2022},
  publisher={Institute of Mathematical Statistics}
}

@inproceedings{li2025algorithmic,
  title={Algorithmic contiguity from low-degree conjecture and applications in correlated random graphs},
  author={Li, Zhangsong},
  booktitle={Approximation, Randomization, and Combinatorial Optimization. Algorithms and Techniques (APPROX/RANDOM 2025)},
  pages={1--30},
  year={2025},
  organization={Schloss Dagstuhl--Leibniz-Zentrum f{\"u}r Informatik}
}

@article{du2025optimal,
  title={Optimal recovery of correlated {E}rd{\H{o}}s-{R}{\'e}nyi graphs},
  author={Du, Hang},
  journal={arXiv preprint arXiv:2502.12077},
  year={2025}
}

@article{piccioli2022aligning,
  title={Aligning random graphs with a sub-tree similarity message-passing algorithm},
  author={Piccioli, Giovanni and Semerjian, Guilhem and Sicuro, Gabriele and Zdeborov{\'a}, Lenka},
  journal={Journal of Statistical Mechanics: Theory and Experiment},
  volume={2022},
  number={6},
  pages={063401},
  year={2022},
  publisher={IOP Publishing}
}

@article{ganassali2024statistical,
  title={Statistical limits of correlation detection in trees},
  author={Ganassali, Luca and Massouli{\'e}, Laurent and Semerjian, Guilhem},
  journal={The Annals of Applied Probability},
  volume={34},
  number={4},
  pages={3701--3734},
  year={2024},
  publisher={Institute of Mathematical Statistics}
}

@inproceedings{mao2021random,
  title={Random graph matching with improved noise robustness},
  author={Mao, Cheng and Rudelson, Mark and Tikhomirov, Konstantin},
  booktitle={Conference on Learning Theory},
  pages={3296--3329},
  year={2021},
  organization={PMLR}
}

@article{umeyama1988eigendecomposition,
  title={An eigendecomposition approach to weighted graph matching problems},
  author={Umeyama, Shinji},
  journal={IEEE Transactions on Pattern Analysis and Machine Intelligence},
  volume={10},
  number={5},
  pages={695--703},
  year={1988},
  publisher={IEEE}
}

@article{conte2004thirty,
  title={Thirty years of graph matching in pattern recognition},
  author={Conte, Donatello and Foggia, Pasquale and Sansone, Carlo and Vento, Mario},
  journal={International Journal of Pattern Recognition and Artificial Intelligence},
  volume={18},
  number={03},
  pages={265--298},
  year={2004},
  publisher={World Scientific}
}

@article{muller1977edge,
  title={The edge reconstruction hypothesis is true for graphs with more than $n\log_2 n$ edges},
  author={M{\"u}ller, Vladim{\'\i}r},
  journal={Journal of Combinatorial Theory, Series B},
  volume={22},
  number={3},
  pages={281--283},
  year={1977},
  publisher={Elsevier}
}

@book{lovasz2012large,
  title={Large networks and graph limits},
  author={Lov{\'a}sz, L{\'a}szl{\'o}},
  volume={60},
  year={2012},
  publisher={American Mathematical Soc.}
}

@inproceedings{berg2005shape,
  title={Shape matching and object recognition using low distortion correspondences},
  author={Berg, Alexander C and Berg, Tamara L and Malik, Jitendra},
  booktitle={2005 IEEE Computer Society Conference on Computer Vision and Pattern Recognition (CVPR'05)},
  volume={1},
  pages={26--33},
  year={2005},
  organization={IEEE}
}

@article{singh2008global,
  title={Global alignment of multiple protein interaction networks with application to functional orthology detection},
  author={Singh, Rohit and Xu, Jinbo and Berger, Bonnie},
  journal={Proceedings of the National Academy of Sciences},
  volume={105},
  number={35},
  pages={12763--12768},
  year={2008},
  publisher={National Acad Sciences}
}

@article{fan2019spectral,
  title={Spectral graph matching and regularized quadratic relaxations {I}: The {G}aussian model},
  author={Fan, Zhou and Mao, Cheng and Wu, Yihong and Xu, Jiaming},
  journal={Foundations of Computational Mathematics},
  volume={23},
  pages={1511-1565},
  number={5},
  year={2023}
}

@article{vogelstein2015fast,
    author = {Vogelstein, Joshua T and Conroy, John M and Lyzinski, Vince and Podrazik, Louis J and Kratzer, Steven G and Harley, Eric T and Fishkind, Donniell E and Vogelstein, R Jacob and Priebe, Carey E},
    journal = {PLOS ONE},
    title = {Fast approximate quadratic programming for graph matching},
    year = {2015},
    volume = {10},
    number = {4},
    pages = {e0121002}
}

@inproceedings{haghighi2005robust,
  title={Robust textual inference via graph matching},
  author={Haghighi, Aria and Ng, Andrew Y and Manning, Christopher D},
  booktitle={Proceedings of Human Language Technology Conference and Conference on Empirical Methods in Natural Language Processing},
  pages={387--394},
  year={2005}
}

@article{paul1959random,
  title={On random graphs {I}},
  author={Erd{\H{o}}s, Paul and R{\'e}nyi, Alfr{\'e}d},
  journal={Publicationes Mathematicae (Debrecen)},
  volume={6},
  pages={290--297},
  year={1959}
}

@article{hall2023partial,
  title={Partial recovery in the graph alignment problem},
  author={Hall, Georgina and Massouli{\'e}, Laurent},
  journal={Operations Research},
  volume={71},
  number={1},
  pages={259--272},
  year={2023},
  publisher={INFORMS}
}

@article{wu2022settling,
  title={Settling the sharp reconstruction thresholds of random graph matching},
  author={Wu, Yihong and Xu, Jiaming and Yu, Sophie H},
  journal={IEEE Transactions on Information Theory},
  volume={68},
  number={8},
  pages={5391--5417},
  year={2022},
  publisher={IEEE}
}

@inproceedings{ganassali2021impossibility,
  title={Impossibility of partial recovery in the graph alignment problem},
  author={Ganassali, Luca and Massouli{\'e}, Laurent and Lelarge, Marc},
  booktitle={Conference on Learning Theory},
  pages={2080--2102},
  year={2021},
  organization={PMLR}
}

@article{ding2023matching,
  title={Matching recovery threshold for correlated random graphs},
  author={Ding, Jian and Du, Hang},
  journal={The Annals of Statistics},
  volume={51},
  number={4},
  pages={1718--1743},
  year={2023},
  publisher={Institute of Mathematical Statistics}
}

@article{cullina2016improved,
  title={Improved achievability and converse bounds for {\uppercase{E}}rd{\H{o}}s-{\uppercase{R}}{\'e}nyi graph matching},
  author={Cullina, Daniel and Kiyavash, Negar},
  journal={ACM SIGMETRICS Performance Evaluation Review},
  volume={44},
  number={1},
  pages={63--72},
  year={2016},
  publisher={ACM New York, NY, USA}
}

@article{cullina2020partial,
  title={Partial recovery of {E}rd{\H{o}}s-{R}{\'e}nyi graph alignment via k-core alignment},
  author={Cullina, Daniel and Kiyavash, Negar and Mittal, Prateek and Poor, H Vincent},
  journal={ACM SIGMETRICS Performance Evaluation Review},
  volume={48},
  number={1},
  pages={99--100},
  year={2020},
  publisher={ACM New York, NY, USA}
}

@article{barak2019nearly,
  title={({N}early) efficient algorithms for the graph matching problem on correlated random graphs},
  author={Barak, Boaz and Chou, Chi-Ning and Lei, Zhixian and Schramm, Tselil and Sheng, Yueqi},
  journal={Advances in Neural Information Processing Systems},
  volume={32},
  pages={9186–9194},
  year={2019}
}

@article{wu2023testing,
  title={Testing correlation of unlabeled random graphs},
  author={Wu, Yihong and Xu, Jiaming and Yu, Sophie H},
  journal={The Annals of Applied Probability},
  volume={33},
  number={4},
  pages={2519--2558},
  year={2023},
  publisher={Institute of Mathematical Statistics}
}

@article{mao2024testing,
  title={Testing network correlation efficiently via counting trees},
  author={Mao, Cheng and Wu, Yihong and Xu, Jiaming and Yu, Sophie H},
  journal={The Annals of Statistics},
  volume={52},
  number={6},
  pages={2483--2505},
  year={2024},
  publisher={Institute of Mathematical Statistics}
}

@article{ding2023detection,
  title={Detection threshold for correlated {\uppercase{E}}rd{\H{o}}s-{\uppercase{R}}{\'e}nyi graphs via densest subgraph},
  author={Ding, Jian and Du, Hang},
  journal={IEEE Transactions on Information Theory},
  year={2023},
volume={69},
number={8},
pages={5289--5298},
  publisher={IEEE}
}

@article{dai2019analysis,
  title={Analysis of a canonical labeling algorithm for the alignment of correlated {\uppercase{E}}rd{\H{o}}s-{\uppercase{R}}{\'e}nyi graphs},
  author={Dai, Osman Emre and Cullina, Daniel and Kiyavash, Negar and Grossglauser, Matthias},
  journal={Proceedings of the ACM on Measurement and Analysis of Computing Systems},
  volume={3},
  number={2},
  pages={1--25},
  year={2019},
  publisher={ACM New York, NY, USA}
}

@article{ding2021efficient,
  title={Efficient random graph matching via degree profiles},
  author={Ding, Jian and Ma, Zongming and Wu, Yihong and Xu, Jiaming},
  journal={Probability Theory and Related Fields},
  volume={179},
  pages={29--115},
  year={2021},
  publisher={Springer}
}

@inproceedings{ganassali2020tree,
  title={From tree matching to sparse graph alignment},
  author={Ganassali, Luca and Massouli{\'e}, Laurent},
  booktitle={Conference on Learning Theory},
  pages={1633--1665},
  year={2020},
  organization={PMLR}
}

@article{mao2025random,
  author  = {Mao, Cheng and Wu, Yihong and Xu, Jiaming and Yu, Sophie H.},
  title   = {Random graph matching at {Otter’s} threshold via counting chandeliers},
  journal = {Operations Research},
  volume  = {74},
  number  = {1},
  pages   = {430--445},
  year    = {2026}
}

@article{mckay1991asymptotic,
  title={Asymptotic enumeration by degree sequence of graphs with degrees o (n 1/2)},
  author={McKay, Brendan D and Wormald, Nicholas C},
  journal={Combinatorica},
  volume={11},
  number={4},
  pages={369--382},
  year={1991},
  publisher={Springer}
}

@article{liebenau2017asymptotic,
  title={Asymptotic enumeration of graphs by degree sequence, and the degree sequence of a random graph},
  author={Liebenau, Anita and Wormald, Nick},
  journal={arXiv preprint arXiv:1702.08373},
  year={2017}
}

@article{ding2023polynomial,
  title={A polynomial-time iterative algorithm for random graph matching with non-vanishing correlation},
  author={Ding, Jian and Li, Zhangsong},
  journal={arXiv preprint arXiv:2306.00266},
  year={2023}
}

@article{ding2023efficiently,
  title={Efficiently matching random inhomogeneous graphs via degree profiles},
  author={Ding, Jian and Fei, Yumou and Wang, Yuanzheng},
  volume = {53},
journal = {The Annals of Statistics},
number = {4},
pages = {1808 -- 1832},
year={2025}
}

@inproceedings{wang2022random,
  title={Random graph matching in geometric models: the case of complete graphs},
  author={Wang, Haoyu and Wu, Yihong and Xu, Jiaming and Yolou, Israel},
  booktitle={Conference on Learning Theory},
  pages={3441--3488},
  year={2022},
  organization={PMLR}
}

@InProceedings{ameen2024robust,
  title = 	 {Robust graph matching when nodes are corrupt},
  author =       {Ameen, Taha and Hajek, Bruce},
  booktitle = 	 {Proceedings of the 41st International Conference on Machine Learning},
  pages = 	 {1276--1305},
  year = 	 {2024},
  volume = 	 {235},
  publisher =    {PMLR}
}

@article{mao2024informationtheoretic,
  title={Information-theoretic thresholds for planted dense cycles},
  author={Mao, Cheng and Wein, Alexander S and Zhang, Shenduo},
  journal={IEEE Transactions on Information Theory},
  year={2024},
  volume={71},
  number={2},
  pages={1266-1282},
  publisher={IEEE}
}

@inproceedings{kunisky2024tensor,
  title={Tensor cumulants for statistical inference on invariant distributions},
  author={Kunisky, Dmitriy and Moore, Cristopher and Wein, Alexander S},
  booktitle={2024 IEEE 65th Annual Symposium on Foundations of Computer Science (FOCS)},
  pages={1007--1026},
  year={2024},
  organization={IEEE}
}

@phdthesis{hopkins2018statistical,
  title={Statistical inference and the sum of squares method},
  author={Hopkins, Samuel},
  year={2018},
  school={Cornell University}
}

@inproceedings{hopkins2017efficient,
  title={Efficient {B}ayesian estimation from few samples: community detection and related problems},
  author={Hopkins, Samuel B and Steurer, David},
  booktitle={2017 IEEE 58th Annual Symposium on Foundations of Computer Science (FOCS)},
  pages={379--390},
  year={2017},
  organization={IEEE}
}

@inproceedings{hopkins2017power,
  title={The power of sum-of-squares for detecting hidden structures},
  author={Hopkins, Samuel B and Kothari, Pravesh K and Potechin, Aaron and Raghavendra, Prasad and Schramm, Tselil and Steurer, David},
  booktitle={2017 IEEE 58th Annual Symposium on Foundations of Computer Science (FOCS)},
  pages={720--731},
  year={2017},
  organization={IEEE}
}

@article{ding2023low,
  title={Low-degree hardness of detection for correlated Erd{\H{o}}s--R{\'e}nyi graphs},
  author={Ding, Jian and Du, Hang and Li, Zhangsong},
  journal={The Annals of Statistics},
  volume={53},
  number={5},
  pages={1833--1856},
  year={2025},
  publisher={Institute of Mathematical Statistics}
}

@article{huang2026information,
  title={Information-Theoretic and Computational Limits of Correlation Detection under Graph Sampling},
  author={Huang, Dong and Yang, Pengkun},
  journal={arXiv preprint arXiv:2601.13966},
  year={2026}
}

@article{aflalo2015convex,
  title={On convex relaxation of graph isomorphism},
  author={Aflalo, Yonathan and Bronstein, Alexander and Kimmel, Ron},
  journal={Proceedings of the National Academy of Sciences},
  volume={112},
  number={10},
  pages={2942--2947},
  year={2015},
  publisher={National Academy of Sciences}
}

@article{du2025algorithmic,
  title={The algorithmic phase transition of random graph alignment problem},
  author={Du, Hang and Gong, Shuyang and Huang, Rundong},
  journal={Probability Theory and Related Fields},
volume={191},
  pages={1233--1288},
  year={2025},
  publisher={Springer}
}

@article{szekely2007measuring,
  title={Measuring and testing dependence by correlation of distances},
  author={Sz{\'e}kely, G{\'a}bor J and Rizzo, Maria L and Bakirov, Nail K},
  journal={The Annals of Statistics},
  volume={35},
  number={6},
  pages={2769--2794},
  year={2007}
}

@article{gretton2007kernel,
  title={A kernel statistical test of independence},
  author={Gretton, Arthur and Fukumizu, Kenji and Teo, Choon and Song, Le and Sch{\"o}lkopf, Bernhard and Smola, Alex},
  journal={Advances in Neural Information Processing Systems},
  volume={20},
  pages={585--592},
  year={2007}
}

@article{gong2024umeyama,
      title={The Umeyama algorithm for matching correlated Gaussian geometric models in the low-dimensional regime}, 
      author={Shuyang Gong and Zhangsong Li},
      journal={arXiv preprint arXiv:2402.15095},
      year={2024}
}

@article{otter1948number,
  title={The number of trees},
  author={Otter, Richard},
  journal={Annals of Mathematics},
  volume={49},
  number={3},
  pages={583--599},
  year={1948},
  publisher={JSTOR}
}

@article{huang2024information,
  title = 	 {Information-Theoretic Thresholds for the Alignments of Partially Correlated Graphs},
  author =       {Huang, Dong and Song, Xianwen and Yang, Pengkun},
  journal = 	 {IEEE Transactions on Information Theory},
  year = 	 {2025},
  volume = {71},
  number = {12},
  pages = {9674--9697},
  publisher = {IEEE}
}

@article{kendall1938new,
  title={A new measure of rank correlation},
  author={Kendall, Maurice G},
  journal={Biometrika},
  volume={30},
  number={1-2},
  pages={81--93},
  year={1938},
  publisher={Oxford University Press}
}

@article{pearson1900x,
  title={Note on regression and inheritance in the case of two parents},
  author={Pearson, Karl},
  journal={Proceedings of the Royal Society of London},
  volume={58},
  pages={240--242},
  year={1895}
}

@article{rohe2011spectral,
  title={Spectral clustering and the high-dimensional stochastic block model},
  author={Rohe, Karl and Chatterjee, Sourav and Yu, Bin},
  journal={Annals of Statistics},
  volume={39},
  number={4},
  pages={1878--1915},
  year={2011}
}

@article{meeks2016challenges,
  title={The challenges of unbounded treewidth in parameterised subgraph counting problems},
  author={Meeks, Kitty},
  journal={Discrete Applied Mathematics},
  volume={198},
  pages={170--194},
  year={2016},
  publisher={Elsevier}
}

@article{flum2004parameterized,
  title={The parameterized complexity of counting problems},
  author={Flum, J{\"o}rg and Grohe, Martin},
  journal={SIAM Journal on Computing},
  volume={33},
  number={4},
  pages={892--922},
  year={2004},
  publisher={SIAM}
}

@article{shepp1966ordered,
  title={Ordered cycle lengths in a random permutation},
  author={Shepp, Lawrence A and Lloyd, Stuart P},
  journal={Transactions of the American Mathematical Society},
  volume={121},
  number={2},
  pages={340--357},
  year={1966}
}

@article{jin2025counting,
  title={Counting Cycles with Deepseek},
  author={Jin, Jiashun and Ke, Tracy and Sui, Bingcheng and Wang, Zhenggang},
  journal={arXiv preprint arXiv:2505.17964},
  year={2025}
}

@article{lee2019network,
  title={Network dependence testing via diffusion maps and distance-based correlations},
  author={Lee, Youjin and Shen, Cencheng and Priebe, Carey E and Vogelstein, Joshua T},
  journal={Biometrika},
  volume={106},
  number={4},
  pages={857--873},
  year={2019},
  publisher={Oxford University Press}
}

@article{fujita2017correlation,
  title={Correlation between graphs with an application to brain network analysis},
  author={Fujita, Andr{\'e} and Takahashi, Daniel Yasumasa and Balardin, Joana Bisol and Vidal, Maciel Calebe and Sato, Jo{\~a}o Ricardo},
  journal={Computational Statistics \& Data Analysis},
  volume={109},
  pages={76--92},
  year={2017},
  publisher={Elsevier}
}

@article{ganassali2022spectral,
  title={Spectral alignment of correlated Gaussian matrices},
  author={Ganassali, Luca and Lelarge, Marc and Massouli{\'e}, Laurent},
  journal={Advances in Applied Probability},
  volume={54},
  number={1},
  pages={279--310},
  year={2022},
  publisher={Cambridge University Press}
}

@article{maier2025asymmetric,
  title={Asymmetric graph alignment and the phase transition for asymmetric tree correlation testing},
  author={Maier, Jakob and Massouli{\'e}, Laurent},
  journal={arXiv preprint arXiv:2504.02299},
  year={2025}
}

@article{varma2025graph,
  title={Graph alignment via Birkhoff relaxation},
  author={Varma, Sushil Mahavir and Waldspurger, Ir{\`e}ne and Massouli{\'e}, Laurent},
  journal={arXiv preprint arXiv:2503.05323},
  year={2025}
}

@book{lehmann2005testing,
  title={Testing statistical hypotheses},
  author={Lehmann, Erich Leo and Romano, Joseph P},
  year={2005},
  publisher={Springer}
}

@article{zhu2017projection,
  title={Projection correlation between two random vectors},
  author={Zhu, Liping and Xu, Kai and Li, Runze and Zhong, Wei},
  journal={Biometrika},
  volume={104},
  number={4},
  pages={829--843},
  year={2017},
  publisher={Oxford University Press}
}

@article{spearman1987proof,
  title={The proof and measurement of association between two things},
  author={Spearman, Charles},
  journal={The American Journal of Psychology},
  volume={15},
  pages={72-101},
  year={1904},
  publisher={JSTOR}
}

@article{cullina2017exact,
  title={Exact alignment recovery for correlated {\uppercase{E}}rd{\H{o}}s-{\uppercase{R}}{\'e}nyi graphs},
  author={Cullina, Daniel and Kiyavash, Negar},
  journal={arXiv preprint arXiv:1711.06783},
  year={2017}
}

@InProceedings{huang2025sample,
  title = 	 {Sample complexity of correlation detection in the {G}aussian {W}igner model},
  author =       {Huang, Dong and Yang, Pengkun},
  booktitle = 	 {Proceedings of the 42nd International Conference on Machine Learning},
  pages = 	 {26020--26040},
  year = 	 {2025},
  volume = 	 {267},
  publisher =    {PMLR},
  }
\end{document}